\documentclass[onecolumn,12pt,journal]{IEEEtran}
\usepackage{graphicx} 

\usepackage[english]{babel}
\usepackage{graphicx}
\usepackage{amsmath}
\usepackage{amsfonts}
\usepackage{amssymb}
\usepackage{amsfonts,eucal,bm}
\usepackage{a4wide}
\usepackage{pgfplots} 
\pgfplotsset{compat=newest}
\usepackage{tikz}
\usepackage{tikzscale}
\usepackage{amsmath}

\newtheorem{theorem}{Theorem}
\newtheorem{lemma}{Lemma}
\newtheorem{definition}{Definition}

\newtheorem{example}{Example}
\newtheorem{remark}{Remark}

\newcommand{\wt}{ \mathrm{\, wt \,} }
\begin{document}

\title{ Soft Decision Decoding of Recursive Plotkin Constructions Based on Hidden Code Words}

\author{Martin Bossert,~\IEEEmembership{Fellow,~IEEE}}

\date{Mai 2024}
\maketitle

\begin{abstract}
	The Plotkin construction combines two codes to a code of doubled length.
	It can be applied recursively.
	The class of Reed--Muller (RM) codes is a particular example.
	Also, a special class of generalized concatenated codes (GCC) can be described as recursive Plotkin
	construction. 
	Exploiting a property of the code words
	constructed by the recursive Plotkin construction,
	we present novel soft-decision decoders. These are based on the
	decoding of hidden code words which are inherent contained in the constructed code words
	and can be uncovered by adding particular parts of the overall code word.  
	The main idea is to use more than one decoding variant where each variant starts with
	the decoding of a different hidden code word.  
	The final decoding decision selects the best of the decisions of the used variants.
	The more variants are used the closer the performance 
	gets to the  maximum-likelihood (ML) decoding performance. This is 
	verified by an ML-bound for the cases where the ML performance is not known.
	The decoding algorithms use only additions, comparisons, and sign operations.
	Further, due to the recursive structure, only relatively short codes 
	have to be decoded, thus, the decoding complexity is very low. 
	In addition, we introduce two novel classes of half-rate codes based on 
	recursive Plotkin constructions with RM codes. 
\end{abstract}

\section{Introduction and Notations}
Even though the history of error-correcting codes and their decoding lasts already for decades, 
there is still the need for codes which have low complexity and good performance decoding methods.
Some part of the history was dominated by iterative decoding methods which
posses good performance but also have a large complexity and often a varying decoding time.
Recently, Reed--Muller (RM) codes became of interest again since they can be decoded
with low complexity and show a relatively good performance.
RM codes can be constructed by a recursive Plotkin construction.
However, the construction is more general and not limited to RM codes.
A set of component codes is combined according to the Plotkin construction
to obtain the overall code word of a recursive Plotkin construction.
The addition of particular parts of the overall code word uncovers a code word
of some component code which we call a hidden code word. 
We present here a novel decoding strategy for recursive Plotkin code constructions which 
uses these hidden code words.
Starting the decoding with a particular hidden code word
leads to one decoding variant.
We use different variants 
and the final decoding decision is the best decoding decision of the used variants
which is decided by calculationg the correlation of the decision with the received vector.
Due to the random nature of a channel the hidden code words have
different channel situations and thus, some have a better situation than others.
Having decoded a hidden code word correctly the remaining recursive decoding steps
benefits from error cancellation due to the fact that the hidden code word is known.
The variants are independent of each other and can be calculated in parallel.
The more variants are used the closer the performance approaches the maximum-likelihood (ML) decoding performance.
Due to the recursive structure we have to decode only relatively short codes
or codes with small dimension.
This can be done by calculating the correlation which all code words and has still low complexity
since a correlation needs only additions.
In fact, all other operations used need also only sign operations, comparison and additions of real numbers 
which supports the low decoding complexity. 

We will recall the related literature.
The Plotkin construction was introduced
by Plotkin~\cite{Plotkin} in $1960$ 
and can be used to create longer codes by combination of shorter codes.
Among others, a recursive construction of Reed--Muller (RM) codes is possible.
RM codes were introduced by Muller \cite{Muller} in $1954$ 
and in the same year Reed \cite{Reed} published a decoding scheme. 
The number of publications which describe decoders for RM codes is huge and therefore we 
will restrict to particular literature.
Namely, to those which are based on a recursive structure where the involved component codes 
are codes which could be used also indendent of the recursion.
Recursive decoding with soft-decision for first order RM codes (simplex codes) 
was presented in \cite{Litsyn} in $1985$.
Then, two years later in $1987$, a correlation decoder which uses the recursive structure
of RM codes was described in \cite{Karyakin}.
Both of this early recursive decoders were restricted to relatively short lengths.
In $1995$, a decoding method for RM codes \cite{Boss95} was published
where the RM code was split into two RM codes (according to the Plotkin construction) 
and each of the two RM codes was split again into two RM codes and so on. 
The decoding was done if one of the splitted codes was a repetition or a parity-check code. 
The idea used in \cite{Boss95} was based on generalized code concatenation (GCC) \cite{BZ} from the year $1974$.
However, the decoding had a very low decoding complexity but also a performance which was much worse than ML.
In the dissertation \cite{Stolte} from $2002$, among others, also recursive decoding methods for RM codes are
studied, however, the decoding performance was also worse than ML decoding.
In $2006$, a decoding method was published \cite{Dumer} which used recursive lists
and shows very good performance. Since several lists are necessary, the complexity 
increases rapidly.
Another decoder which is based on the automorphism group of RM codes
was published in $2021$  \cite{Geiselhart}. 
In fact, the used permutations exploit the random nature of the channel
since the decoders face different channel situations, but in a different way than the novel decoding
strategy we will present here.
In \cite{Freudenberger} from $2021$, Plotkin constructions with BCH codes are studied 
which already shows that the construction is not limited to RM codes only
and that the BCH codes give better results than the RM codes.
In $2022$, a recursive decoder \cite{Kamenev} was published which
starts the decoding with the larger rate code in a Plotkin construction.
Since several decoding trials are made by flipping unreliable bits it is mainly suited for 
small minimum distances. 

We introduce a few necessary definitions and notations. 
Let  $\mathbb F_2$ be the binary field and $\mathbb F_2^n$
the $n$-dimensional binary vector space.
We restrict to binary linear codes   
$\mathcal C (n,k,d) \subseteq \mathbb F_2^n $
which are linear  $k$-dimensional subspaces where $n$ is the length, $k$ the dimension, and
$d$ the minimum distance of the code.
For a code word  $\mathbf c \in \mathcal C$ 
the norm is the Hamming weight $\mathrm{wt} (\mathbf c)$ which is the number of ones in the vector.
The metric is the
Hamming distance $\mathrm{dist} (\mathbf a, \mathbf b) = \mathrm{wt} (\mathbf a + \mathbf b)$. 
The minimum distance $d$ of a code is the smallest Hamming distance between any
two different code words which for linear codes is identical to the smallest weight of the nonzero code words.
A binary symbol is transmitted over a binary symmetric channel (BSC) and 
with probability $p$ the bit is flipped from $0$ to $1$ or from $1$ to $0$.
Thus, with probability $1-p$ the transmission of a binary symbol is error free.
Transmitting a code word of length $n$ over a BSC the expected number of errors is $n p$.
For a code with odd minimum distance $d$ 
a bounded minimum distance (BMD) decoder
can uniquely correct $\tau \leq (d-1)/2$ errors.
For soft-decision decoding we use
the Gaussian channel and  binary phase shift keying (BPSK)
which will be introduced in Sec.~\ref{sec:softdec}.

The remaining paper is organized as follows.
In Sec. \ref{sec:defplot} we recall the recursive Plotkin construction and  
introduce the hidden code words. In addition, we relate the construction to RM codes and to GCC.
Decoding functions and aspects of soft-decision decoding are treated in Sec. \ref{sec:softdec}.
ML decoding with correlation and the ML bound are defined. 
Then, in Sec. \ref{sec:decstrat}, we introduce the {\em join} operation
in order to combine blocks of the received vector related to different hidden code words.
Then we show that these combinations of blocks cancel errors. Afterwards, $15$ decoding variants are 
introduced based on the decoding functions and structured by the hidden code word of the first decoding step.
Finally, the recursive decoding is described.
In Sec. \ref{sec:decperform} we first introduce some realizations of the decoding functions
and then analyze the performance of the variants using a half-rate RM code of length $32$.
We count the necessary operations to show the complexity of the decoders.
In addition, modified code constructions which are not RM codes are introduced and their performance is simulated.
Two of these constructions will be generalized to get two new classes of half-rate codes based on RM codes
which are not RM codes itself.
The performance of recursive decoding is shown in Sec. \ref{sec:recurplot}
using a half-rate RM code of length $128$ which will be decoded by decoding code words of length $8$.
Additionally, we show the performance of representatives of the novel code classes.
Finally, it is shown that the algorithms perform well also 
for low- and  high-rate codes.
At the end, we give some conclusions and mention some open problems.

\section{The Plotkin Construction}\label{sec:defplot}
The original Plotkin construction \cite{Plotkin} from 1960 combines  two codes
which we will denote by $\mathcal C_0$ and  $\mathcal C_1$.
These codes are combined to a code  $\mathcal C$ of double length by choosing two code words 
$\mathbf u_0 \in \mathcal C_0$  and $\mathbf u_1 \in \mathcal C_1$  
and appending the code word
 $\mathbf u_0+\mathbf u_1 $ to the code word
$\mathbf u_0 $ which results in the code word
$\mathbf c = | \mathbf u_0 | \mathbf u_0+\mathbf u_1 |$ of the new code $\mathcal C$.
The following known Theorem 
gives the parameters of the resulting code $\mathcal C$.
There exist various proofs, however, here we have choosen one which 
gives insight into the presented decoding strategy.

\begin{theorem}[Plotkin Construction]\label{uuplusvrm}%
Let $\mathcal C_0(n,k_0,d_0)$ and 
$\mathcal C_1(n,k_1,d_1)$ be two binary linear component codes.
Then, the code $\mathcal C = \{ \mathbf c = |\mathbf u_0|\mathbf u_0+\mathbf u_1|:\ \mathbf u_0\in \mathcal C_0,\ \mathbf u_1 \in \mathcal C_1 \}$ has length $2n$, dimension $k= k_0 + k_1$, and minimum distance $d=\min\{ 2d_0, d_1\}$.
\end{theorem}

\begin{proof}
The length $2n$ and the dimension  $k= k_0 + k_1$ are obvious.
In order to prove the minimum distance we will describe a possible decoder of $\mathcal C$.
Assume $\mathbf c \in \mathcal C$ was transmitted over a BSC and $\mathbf r = \mathbf c + \mathbf e $
was received. Knowing that  $\mathbf c = |\mathbf u_0|\mathbf u_0+\mathbf u_1|$
we can split the error in two halfs $\mathbf e = |\mathbf e_0|\mathbf e_1|$, thus,  
	$\mathbf r = |\mathbf u_0 + \mathbf e_0|\mathbf u_0+\mathbf u_1+\mathbf e_1|$. 
	Adding the left to the right half we get  
\[\mathbf u_0 + \mathbf e_0 + \mathbf u_0+\mathbf u_1+\mathbf e_1
= \mathbf u_1 + \mathbf e_0+\mathbf e_1.\]
Since  $\wt(\mathbf e) = \wt(\mathbf e_0) + \wt(\mathbf e_1)
\geq  \wt(\mathbf e_0 + \mathbf e_1)$ a BMD decoder can decode 
$\mathbf u_1$ correctly if the number of errors $\tau = \wt (\mathbf e) \leq \frac{d_1-1}{2}$ and thus, we know $\mathbf u_1$ and the minimum distance $d$ is $d_1$ for this step of decoding.

Now, we can add the decoded $\mathbf u_1$ to the right half of $\mathbf r$ and get
$\mathbf u_0+\mathbf e_1$. Therefore, we have the two vectors 
\[\mathbf u_0 + \mathbf e_0  \quad \mbox{and} \quad \mathbf u_0+\mathbf e_1\]
which is the same code word $\mathbf u_0$ disturbed in the left 
vector by $\mathbf e_0$
and in the right by $\mathbf e_1$. Now, we assume that
$\tau =\wt (\mathbf e)= d_0 - 1$ errors have occurred
and $d_0$ is even. 
Let $j$ be the number of errors in the left part. For the case
$j < d_0-1-j$ the left part contains $j=0,1, \ldots, <\frac{d_0-1}{2}$
errors. Thus, the decoder for the left part can find the correct $\mathbf u_0$,
since the decoder for the right part needs to correct with at least $ \geq j+1$ errors 
to find a valid code word with weight  $ \geq d_0$. 
This follows from the known fact that the error on the channel added to the estimated error by the decoder
	is a code word of the used code, thus, the sum of the weigths of both must be at least $d_0$.
	So if we decide for the decoder 
which corrects the smaller number of errors we have $\mathbf u_0$.
In the second case
when $j > d_0-1-j$ the right decoder finds the correct $\mathbf u_0$
since the decoder for the left part needs at least to correct with $\geq j+1$ errors.
The same arguments hold if $\tau = \wt (\mathbf e) < d_0 - 1$. 
A code which can correct $d_0-1$ errors uniquely has minimum distance $ 2 d_0$.
\end{proof}

\subsection{Recursive Plotkin Construction}
It is known that the Plotkin construction can be applied recursively \cite{Boss-eng}.
We recall the recursion in order to introduce the double Plotkin construction.
Given are four component codes $\mathcal C_0$,$\mathcal C_1$,$\mathcal C_2$, and $\mathcal C_3$
and the two Plotkin constructions  $|\mathbf u_0|\mathbf u_0+\mathbf u_1|$ and  $ |\mathbf u_2|\mathbf u_2+\mathbf u_3|$.
If we  apply the Plotkin construction for these two codes we get 
a code of length  $4n$

\begin{equation}\label{doublepl}
 |\mathbf u_0|\mathbf u_0+\mathbf u_1 |\mathbf u_0+\mathbf u_2|\mathbf u_0 + \mathbf u_1 +\mathbf u_2+\mathbf u_3|. 
\end{equation}
which we refer to as double Plotkin construction.
This can be continued to get a triple Plotkin construction using a second double Plotkin constructed code
and so on. 
Any double Plotkin construction can be viewed as a Plotkin construction. 
However, a Plotkin construction can only be viewed as a double Plotkin construction
if the component codes are Plotkin constructions themselves.
\begin{example}[Recursive Plotkin Construction]\label{ex:recursive}
	We use the binary component codes  $\mathcal C_0 (8,7,2)$,  $\mathcal C_1 (8,4,4)$,
	$\mathcal C_2 (8,4,4)$, and $\mathcal C_3 (8,1,8)$.
In a first step we use the Plotkin construction  for the two code words 
$\mathcal C_0$ and $\mathcal C_1$ and for the codes
$\mathcal C_2$ and $\mathcal C_3$ which gives the two code words
$$
\mathbf v_0 = |\mathbf u_0|\mathbf u_0+\mathbf u_1| \ \mathrm{and} \  \mathbf v_1 = |\mathbf u_2|\mathbf u_2+\mathbf u_3| .
$$
Now we use the Plotkin construction for these two code words which results in the code word
$$
 |\mathbf v_0|\mathbf v_0+\mathbf v_1 | = |\mathbf u_0|\mathbf u_0+\mathbf u_1 |\mathbf u_0+\mathbf u_2|\mathbf u_0 + \mathbf u_1 +\mathbf u_2+\mathbf u_3| 
$$
which is a code word of the code $\mathcal C (32,16,8)$, since according to Theorem \ref{uuplusvrm}
the code $|\mathbf u_0|\mathbf u_0+\mathbf u_1 |$ 
has the parameters $(16,11,4)$ and the code 
$|\mathbf u_2|\mathbf u_2+\mathbf u_3|$ 
has the parameters $(16,5,8)$. A Plotkin construction of these two codes gives $\mathcal C (32,16,8)$. 
\end{example}
In a recursive Plotkin construction, when the used codes have certain subcode properties,
 there exist different hidden code words of the used codes which are not only
 $\mathbf u_1$, $\mathbf u_2$, and $\mathbf u_3$.
These hidden code words can be uncovered when combining the four blocks in different ways which is
shown in the following lemma.
\begin{lemma}[Hidden Code Words]\label{subcodedouble}
Given a double Plotkin construction with the component codes  $\mathcal C_0$, $\mathcal C_2
\subseteq \mathcal C_1$, and $\mathcal C_3 \subset \mathcal C_2$ and a code word
$$\mathbf c =|\mathbf u_0|\mathbf u_0+\mathbf u_1|\mathbf u_0+\mathbf u_2|\mathbf u_0+\mathbf u_1+\mathbf u_2+\mathbf u_3|
=|\mathbf a_0|\mathbf a_1|\mathbf a_2|\mathbf a_3|$$
	then the uncovered code word 
	$\mathbf a_0+\mathbf a_1 + \mathbf a_2+\mathbf a_3$ is from code $\mathcal C_3$.
The two uncovered code words
	$\mathbf a_0+\mathbf a_2$ and $\mathbf a_1+\mathbf a_3$ 
		are from the code $ \mathcal C_2$ and the four uncovered code words
	$\mathbf a_0+\mathbf a_1, \  \mathbf a_0+\mathbf a_3, \ \mathbf a_1+\mathbf a_2, \  \mathbf a_2+\mathbf a_3$ 
	are from the code
	$ \mathcal C_1$. 
In case of $\mathcal C_1 = \mathcal C_2$ all six code words are from this code. If additionally $\mathcal C_1 \subset \mathcal C_0$
all blocks $\mathbf a_i$ are from code $\mathcal C_0$. 	
\end{lemma}
\begin{proof}
The code word $\mathbf a_0+\mathbf a_2=\mathbf u_2$ is from $\mathcal C_2$.
The code word $\mathbf a_1+\mathbf a_3=\mathbf u_2+\mathbf u_3$ is from $\mathcal C_2$,
since $\mathcal C_3 \subset \mathcal C_2$.
The code word $\mathbf a_0+\mathbf a_1=\mathbf u_1$ 
is from the code $\mathcal C_1$.
The code word  $\mathbf a_0+\mathbf a_3=\mathbf u_1+\mathbf u_2 + \mathbf u_3$
is  also from $\mathcal C_1$, since 
	$\mathcal C_3 \subset \mathcal C_2 \subseteq \mathcal C_1$.
The code word $\mathbf a_1+\mathbf a_2=\mathbf u_1+\mathbf u_2$ is from the same code, since
$\mathcal C_2 \subseteq \mathcal C_1$.
Finally, $\mathbf a_2+\mathbf a_3=\mathbf u_1+\mathbf u_3$ 
is from $\mathcal C_1$ because $\mathcal C_3 \subset \mathcal C_1$.	
Since all code words are either from $\mathcal C_1$ or $\mathcal C_2$
they are all from the same code if  $\mathcal C_1 = \mathcal C_2$.
When $\mathcal C_1 \subset \mathcal C_0$
then $\mathbf u_0 \in  \mathcal C_0$,  
	$\mathbf u_0  + \mathbf u_1 \in  \mathcal C_0$, 
	$\mathbf u_0  + \mathbf u_2 \in  \mathcal C_0$ because  $\mathcal C_2
\subseteq \mathcal C_1$, and 
	$\mathbf u_0  + \mathbf u_1 + \mathbf u_2 + \mathbf u_3 \in  \mathcal C_0$. 
According to the construction of the code words of a double Plotkin construction
	$\mathbf a_0+\mathbf a_1 + \mathbf a_2+\mathbf a_3 = \mathbf u_3$ is from code $\mathcal C_3$.
\end{proof}

The following example illustrates this lemma.

\begin{example}[Hidden Code Words]\label{ex:rm25}
Again we consider the four binary codes  $\mathcal C_0 (8,7,2)$,  $\mathcal C_1 (8,4,4)$,
	$\mathcal C_2 (8,4,4)$, and $\mathcal C_3 (8,1,8)$ and the double Plotkin construction
$$
 |\mathbf u_0|\mathbf u_0+\mathbf u_1 |\mathbf u_0+\mathbf u_2|\mathbf u_0 + \mathbf u_1 +\mathbf u_2+\mathbf u_3|.
$$
Adding the first to the third block and the second to the fourth one corresponds to the Plotkin construction and  
we get $|\mathbf u_2|\mathbf u_2+\mathbf u_3|$ which is a code word from $(16,5,8)$. 
However, we may also add the first to the second and the third to the fourth block
and get $|\mathbf u_1|\mathbf u_1+\mathbf u_3|$  
which is also from the code $(16,5,8)$. 
Since the codes $\mathcal C_1$ and $\mathcal C_2$ are identical 
it is further possible to   add the second to the third and the first to the third and we get
$|\mathbf u_1+\mathbf u_2|\mathbf u_1+\mathbf u_2+\mathbf u_3|$  
which is again from code $(16,5,8)$.
Another interpretation are the six code words from code  $(8,4,4)$  which
are obtained by combining any two of the four blocks of the double Plotkin construction.
These uncovered code words are $\mathbf u_1$, $\mathbf u_2$,  $\mathbf u_1+\mathbf u_2$, $\mathbf u_1+\mathbf u_3$, 
$\mathbf u_2+\mathbf u_3$, and $\mathbf u_1+\mathbf u_2+\mathbf u_3$. 
Since $\mathcal C_3 \subset \mathcal C_2 =  \mathcal C_1 \subset \mathcal C_0$
each block is a code word from a parity-check code.
\end{example}
The combination of different blocks is the basis for the presented novel decoding strategy.
Due to the statistical nature of a channel
the number of errors as well as their position is most likely different in any of the four blocks.  
Especially in the case when each position has a reliability value  these values
will be different in each of the four blocks which will be exploited by soft-decision decoding later.

Before describing the decoding
we will relate the recursive Plotkin construction to binary Reed-Muller (RM) codes 
and to generalized code concatenation (GCC).
The GCC point of view will give insight into the decoding.
\subsection{Reed--Muller Codes}
The binary RM codes can be described in many ways and we will 
use the Plotkin construction recursively to get the codes according to Definition \ref{rmcodesrm} \cite[p. 106]{Boss-eng}.
\begin{definition}[Binary RM-Code]\label{rmcodesrm}
	The RM code ${\mathcal R}(r,m)$ 
with order $r$ has length $n=2^m$,  dimension 
$k = 1 + \binom{m}{1} + \binom{m}{2} + \dots + \binom{m}{r}$
and  minimum distance $d=2^{m-r}$. That is 
 ${\mathcal R}(r,m)$ is a $(2^m, k, 2^{m-r})$ code.
\end{definition}

We can use RM codes in the Plotkin construction. If we use
$\mathcal C_0 ={\mathcal R}(r+1,m)$ 
and $\mathcal C_1 = {\mathcal R}(r,m)$ 
we get $\mathcal C={\mathcal R}(r+1,m+1)$, thus
\begin{equation*}
{\mathcal R}(r+1,m+1) = \left\{|\mathbf u_0|\mathbf u_0 + \mathbf u_1|:\ \mathbf u_0 \in {\mathcal R}(r+1,m),\ \mathbf u_1 \in {\mathcal R}(r,m)
\right\}.
\end{equation*}
Since $d_0= 2^{m-r-1}$  and  $d_1= 2^{m-r}$ we get $d = 2^{m-r}$. 
For the length we get $2 \cdot 2^m = 2^{m+1}$.
For the dimension it holds that
\begin{equation*}
\binom{m+1}{r+1} = \binom{m}{r+1} + \binom{m}{r}
\end{equation*}
and thus $k=k_0+k_1$ and the parameters 
are those of the RM code $\mathcal R(r+1,m+1)$.
\begin{figure}[htb]
  \center
\includegraphics[width=0.95\textwidth]{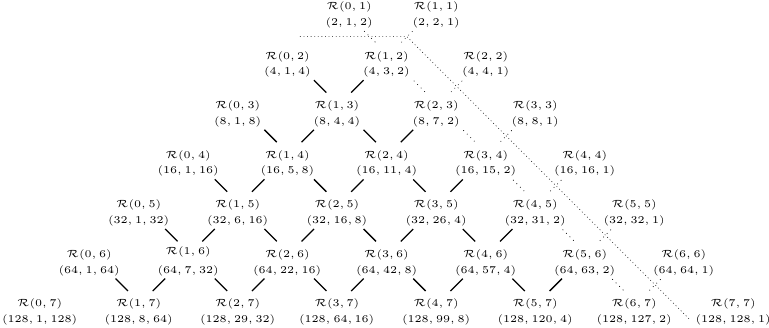}
  \caption{Recursive Plotkin Construction of RM codes}\label{fig:rm-tree}
\end{figure}
The recursive calculation of RM codes up to length $128$ 
is illustrated in Figure~\ref{fig:rm-tree} \cite[p. 106]{Boss-eng}.
The codes $\mathcal R(0,m)$ are repetition codes, $\mathcal R(1,m)$ are simplex codes, $\mathcal R(m-2,m)$
are extended Hamming codes, $\mathcal R(m-1,m)$ are parity-check codes, and $\mathcal R(m,m)$ are all binary vectors
of length $2^m$.
\begin{remark} The recursive structure of RM codes
is the basis for the soft- and hard-decision decoding algorithm
given in \cite{Boss95}. A list-decoding extension can be found in \cite{Boss-eng}.
\end{remark}
\subsection{Generalized Code Concatenation}
The concept of generalized code concatenation (GCC) uses an inner code which is partitioned into subcodes 
which may have increased minimum distance.
Outer codes protect the labeling of this partitioning.
Numerous theoretical results exist (\cite[pp. 287--389]{Boss-eng}) and it would be 
beyond the scope of this paper to explain the theory in detail.
However, the connection to the Plotkin constructions gives some insight and the basic idea
can be explained by the following example which describes a Plotkin and a double Plotkin 
construction as GCC.

\begin{example}[GCC and Plotkin Construction] \\
	We choose as inner code $\mathcal B(2,2,1)=\{00,01,10,11 \}$ and partition it (as shown in Fig. \ref{fig:partsingl})
into two subsets, namely $\{00,11 \}$ and  $\{01,10 \}$ where the first set is labeled by a zero and the second by a one.
These two subsets can again be partitioned into two vectors labeled again by zero and one.
This yields the mapping of the labeling to the vectors: $00 \leftrightarrow 00$,  $01 \leftrightarrow 11$, 
 $10 \leftrightarrow 01$, and  $11 \leftrightarrow 10$ which 
corresponds to $u_0|u_0+u_1$.
Now we use the two outer codes $\mathcal C_1(4,1,4)$ and $\mathcal C_0(4,3,2)$.
The bits $u_{1,i}, i = 0, 1,2, 3$ of a code word from  $\mathcal C_1(4,1,4)$ 
 are the first labels and the bits $u_{0,i}, i = 0, 1, 2, 3$ of  a code word from  $\mathcal C_0(4,3,2)$ 
are the second labels.
Assume the information bits are $1011$. First, we encode the two outer codes and get the code words   $1111$ and $0110$
where we used the first information bit for the repetition code and the remaining three for the parity-check code.
The labels are now $10 \leftrightarrow 01$, $11 \leftrightarrow 10$, $11 \leftrightarrow 10$, and $10 \leftrightarrow 01$.
The code word of the GCC is $01 \  10 \  10 \ 01$.
We map the left of the two  bits  to $c_0, c_1, c_2, c_3 = 0110$ and  the right to $c_{4},c_5,c_6, c_7 = 1001$.
We see that this is the same code word which we would get by Plotkin construction
 $\mathbf c = |\mathbf u_0 |\mathbf u_0 +\mathbf u_1|$.

In Fig. \ref{fig:partdouble} the inner code is $\mathcal B(4,4,1)$ which is partitioned into
two subsets containing the vectors of odd and even Hamming weight.
Each of these two subsets is again partitioned into two subsets and so on.
If the partitioning is done according to
$u_0|u_0+u_1|u_0+u_2|u_0+u_1+u_2+u_3$ the last labeling is protected with Hamming distance four. 
Since we have four labels of the partitioning we need four outer codes.
If we choose as outer codes $\mathcal C_0=(8,7,2)$, $\mathcal C_1=(8,4,4)$,
	$\mathcal C_2=(8,4,4)$ and $\mathcal C_3=(8,1,8)$ 
we get the same code as in Ex. \ref{ex:rm25}.

\end{example}

\begin{figure}[htb]
\centerline{\begin{tikzpicture}[scale=0.6]

\draw (1.5,1.3) node {$u_1$};
\draw (0,0) -- (4,2);
\draw (8,0) -- (4,2);

\draw (-1.5,-0.8) node {$u_0$};
\draw (0,0) -- (-2,-2);
\draw (0,0) -- (2,-2);
\draw (8,0) -- (6,-2);
\draw (8,0) -- (10,-2);

\draw[fill] (4,2) circle (2pt);
\draw[fill] (0,0) circle (2pt);
\draw[fill] (8,0) circle (2pt);
\draw[fill] (-2,-2) circle (2pt);
\draw[fill] (2,-2) circle (2pt);
\draw[fill] (6,-2) circle (2pt);
\draw[fill] (10,-2) circle (2pt);

\draw (-2,-3) node {$00$};
\draw (2,-3) node {$11$};
\draw (6,-3) node {$01$};
\draw (10,-3) node {$10$};

\end{tikzpicture}}
 \caption{Partitioning of $(2,2,1)$}\label{fig:partsingl}
\end{figure}
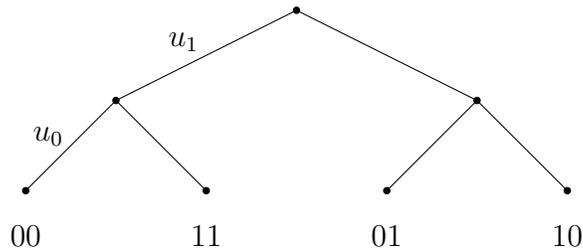

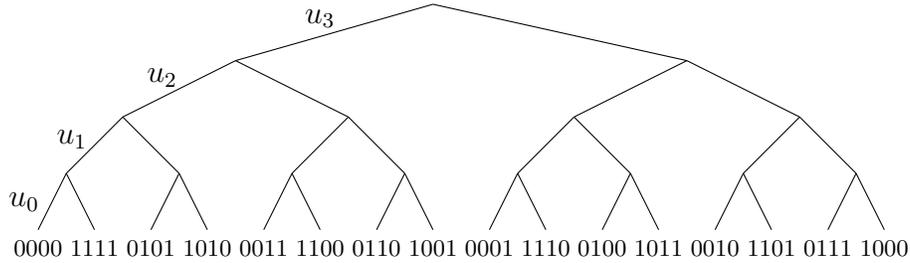
\begin{figure}[htb]
\centerline{\begin{tikzpicture}[scale=.75]

\draw (6,-0.25) node {$u_3$};
\draw (8,0) -- (4.5,-1);
\draw (8,0) -- (12.5,-1);

\draw (3.2,-1.3) node {$u_2$};
\draw (4.5,-1) -- (2.5,-2);
\draw (4.5,-1) -- (6.5,-2);

\draw (1.6,-2.4) node {$u_1$};
\draw (2.5,-2) -- (1.5,-3);
\draw (2.5,-2) -- (3.5,-3);
\draw (6.5,-2) -- (5.5,-3);
\draw (6.5,-2) -- (7.5,-3);

\draw (0.75,-3.5) node {$u_0$};
\draw (1.5,-3) -- (1,-4);
\draw (1.5,-3) -- (2,-4);
\draw (3.5,-3) -- (3,-4);
\draw (3.5,-3) -- (4,-4);
\draw (5.5,-3) -- (5,-4);
\draw (5.5,-3) -- (6,-4);
\draw (7.5,-3) -- (7,-4);
\draw (7.5,-3) -- (8,-4);

\draw (12.5,-1) -- (10.5,-2);
\draw (12.5,-1) -- (14.5,-2);

\draw (10.5,-2) -- (9.5,-3);
\draw (10.5,-2) -- (11.5,-3);
\draw (14.5,-2) -- (13.5,-3);
\draw (14.5,-2) -- (15.5,-3);

\draw (9.5,-3) -- (9,-4);
\draw (9.5,-3) -- (10,-4);
\draw (11.5,-3) -- (11,-4);
\draw (11.5,-3) -- (12,-4);
\draw (13.5,-3) -- (13,-4);
\draw (13.5,-3) -- (14,-4);
\draw (15.5,-3) -- (15,-4);
\draw (15.5,-3) -- (16,-4);


	{\footnotesize
\draw (1,-4.33) node {$0000$};
\draw (2,-4.33) node {$1111$};
\draw (3,-4.33) node {$0101$};
\draw (4,-4.33) node {$1010$};
\draw (5,-4.33) node {$0011$};
\draw (6,-4.33) node {$1100$};
\draw (7,-4.33) node {$0110$};
\draw (8,-4.33) node {$1001$};
\draw (9,-4.33) node {$0001$};
\draw (10,-4.33) node {$1110$};
\draw (11,-4.33) node {$0100$};
\draw (12,-4.33) node {$1011$};
\draw (13,-4.33) node {$0010$};
\draw (14,-4.33) node {$1101$};
\draw (15,-4.33) node {$0111$};
\draw (16,-4.33) node {$1000$};
}
\end{tikzpicture}}
	\caption{Partitioning of $(4,4,1)$}\label{fig:partdouble}
\end{figure}

With this example we have shown that  
 GCC  with  a   particular partition of the inner code
 $\mathcal B (2,2,1)$ is equivalent to a Plotkin construction
For the inner code $\mathcal B (4,4,1)$ the GCC is equivalent to a double Plotkin construction. 
This could be continued for inner codes
 $\mathcal B (2^m,2^m,1)$ which are equivalent to recursive Plotkin constructions.

\section{Soft Decision Decoding}\label{sec:softdec}
We use binary phase shift keying (BPSK) and  
	the usual mapping of the binary code symbols  $u_i=0 \leftrightarrow x_i = 1$ 
	and $u_i=1 \leftrightarrow x_i = -1$, and for vectors $\mathbf u \leftrightarrow \mathbf x$.
	In the additive white Gaussian noise (AWGN) channel we 
	receive $y_i = x_i + z_i$ where $z_i$ denotes the Gaussian noise.
	As usual, the signal-to-noise ratio for information bits in dB is called $E_b/N_0$. 
	Then, the signal-to-noise ratio  for a code bit in dB is $E_s/N_0 = E_b/N_0 - 10 \cdot \log_{10}(1/R)$
	where $R$ is the code rate of the used code.
	The mean is zero and the variance of the noise is  $\sigma^2=1/(2 \cdot 10^{E_s/(10 N_0)})$.
We define the hard decision received vector  $\mathbf r$ by the mapping $y_i \geq 0 \leftrightarrow r_i = 0$ 
	and $y_i<0 \leftrightarrow r_i = 1$.
The addition of two binary code words is a component-wise multiplication of 
the BPSK-modulated code words  $\mathbf u_0 + \mathbf u_1 \leftrightarrow \mathbf x_0 \mathbf x_1$,
where $$\mathbf x_0 \mathbf x_1 = ( x_{0,0} x_{1,0}, x_{0,1} x_{1,1}, \ldots ,  x_{0,n-1} x_{1,n-1}), \ x_{i,j} = \pm 1.$$

\subsection{Decoding Functions}
We need decoders for all codes used in a recursive Plotkin construction.
Since we might use  a different decoding algorithm for each code, possibly one with the best known performance or  
one with the lowest complexity, we introduce here decoding functions. With these, we
can describe the different decoding variants without details of the used algorithms. 

Given a code $\mathcal C$ with the BPSK modulated code words $ \mathbf x \in \mathcal C$
and a received vector  $\mathbf y = \mathbf x + \mathbf z$, where  $ \mathbf z$ is the noise,
we denote the decoding function by 
\begin{equation}\Delta (\mathbf y)=\hat{\mathbf{x}}.
\end{equation}
If $ \hat{\mathbf{x}} = \mathbf x$, the decoding is correct and otherwise, we have a decoding error.
For an ML decoder the probability $P(\hat{\mathbf{x}} | \mathbf y)$ is the maximum over all 
probabilities $P(\mathbf{x} | \mathbf y)$ for all code words $ \mathbf x \in \mathcal C$.
We will also use list decoders which are denoted by
\begin{equation}\Lambda (\mathbf y)=\{ \hat{\mathbf{x}}_1, \hat{\mathbf{x}}_2, \ldots, \hat{\mathbf{x}}_L \},
\end{equation}
where the output is a list of $L$ decoding decisions, possibly the ones with the largest probabilities.
For list size $L=1$ we have $\Delta (\mathbf y)=\Lambda (\mathbf y)$.

In addition, we also use an indicator function which indicates whether a vector is a valid code word of a code $\mathcal C$
\begin{equation}\Gamma (\mathbf v)= \left\{ \begin{array}{l} 
0, \mathbf v \not\in \mathcal C \\ 
1, \mathbf v \in \mathcal C .
\end{array} \right. 
\end{equation}

\subsection{Correlation and ML Decoding}
Given the modulated code words $\mathbf x$ and a received vector  $\mathbf y$
then the ML decision is equivalent to decide for the code word with the smallest Euclidean distance. 
However, the squared Euclidean distance can also be used which is defined by
\begin{equation}\label{eq.speuclid}
d_E^2=\sum\limits_{i=0}^{n-1} (x_i-y_i)^2 =\sum\limits_{i=0}^{n-1} x_i^2- 2\sum\limits_{i=0}^{n-1} x_i y_i + \sum\limits_{i=0}^{n-1} y_i^2 .
\end{equation}
From Eq. \eqref{eq.speuclid} it follows that
the minimization of the squared Euclidean distance is equivalent to the maximization   
of the correlation of the received vector $\mathbf y$ with all
code words $\mathbf x$. The correlation is defined by
\begin{equation}\label{def:corr}
\Phi (\mathbf x, \mathbf y) = \sum\limits_{i=0}^{n-1} x_i  y_{i}
\end{equation}
and has the property $\Phi (\mathbf x, \mathbf y)= - \Phi (-\mathbf x, \mathbf y)$
which can be exploited in cases where the all-one (modulated the all-minus-one) vector is a valid code word.
Any code $\mathcal C$ which contains the all one code word can be partitioned into two
subsets $\mathcal C^-$ and $\mathcal C^+$ such that for any  $\mathbf c \in \mathcal C^- $ 
the inverted code word  is in the other subset $\mathbf c + \mathbf 1 \in \mathcal C^+ $. 
For the modulated code words $\mathbf x$ the inverted code word is  $-\mathbf x$.

\begin{lemma}\label{corrallone}
If a code contains the all-one codeword then 
maximizing  $|\Phi (\mathbf x, \mathbf y)|$ over $\mathcal C^-$ is ML decoding. 
The decoding decision is $\mathbf x_m$ if $\Phi (\mathbf x_m, \mathbf y)>0$
and $-\mathbf x_m$ if $\Phi (\mathbf x_m, \mathbf y)<0$.
\end{lemma}
\begin{proof}
Let $\Phi (\mathbf x_m, \mathbf y)$ be the maximum correlation for 
 $\mathbf x_m \in \mathcal C$.
If  $\mathbf x_m \in \mathcal C^-$ then $|\Phi (\mathbf x_m, \mathbf y)|= \Phi (\mathbf x_m, \mathbf y)$.  
If  $\mathbf x_m \in \mathcal C^+$ then $|\Phi (\mathbf x_m, \mathbf y)|= -\Phi (\mathbf x_m, \mathbf y)$
due to the property of the correlation $\Phi (\mathbf x, \mathbf y)= - \Phi (-\mathbf x, \mathbf y)$.
\end{proof}
\subsection{ML Decoding Bound}
We will use an old result \cite[Th. 9.28, p. 384]{Boss-eng} to simulate a lower bound for ML decoding.
Assume that $\mathbf x_t$ was transmitted and is known when simulating a decoding algorithm.
For a received $\mathbf y$  the decoding algorithm decides for $\hat{\mathbf x}$.
In case $\Phi (\mathbf x_t, \mathbf y) < \Phi (\hat{\mathbf x}, \mathbf y)$ the decoding is wrong
and an ML decoder would also not decode correctly.
There could possibly exist
other code words than $\hat{\mathbf x}$ with larger correlation, however there exists at least one,
namely $\hat{\mathbf x}$,
and the ML decoder will decode wrongly.
In case $\Phi (\mathbf x_t, \mathbf y) > \Phi (\hat{\mathbf x}, \mathbf y)$
we assume that an ML decoder would have decoded correctly which might not be the case since there could possibly exist
other code words than the transmitted one with larger correlation. 
Because of this fact we only obtain a bound based on these considerations. 
Nevertheless,
if the number of incorrectly decoded transmissions of the decoding algorithm
coincides with the value of the bound the decoding algorithm has the same performance as ML decoding.

\section{Decoding Strategies for Plotkin Constructions}\label{sec:decstrat}
The addition of two binary code words corresponds to a component-wise multiplication of 
the BPSK modulated code words.
For the Plotkin construction we get
$$ |\mathbf u_0|\mathbf u_0 + \mathbf u_1| \leftrightarrow |\mathbf x_0|\mathbf x_0 \mathbf x_1|.$$
The addition of the first and the second block becomes a multiplication
$$ \mathbf u_0 + \mathbf u_0  + \mathbf u_1 =  \mathbf u_1 \leftrightarrow \mathbf x_0 \mathbf x_0 \mathbf x_1 = \mathbf x_1,$$
since  $\mathbf x_0 \mathbf x_0   = (1,1, \ldots,1)$.
Transmitting the modulated code word over a Gaussian channel we receive
$|\mathbf y_0 = \mathbf x_0 + \mathbf z_0| \mathbf y_1 = \mathbf x_0 \mathbf x_1 + \mathbf z_1|$
which has the hard decision interpretation
$|\mathbf r_0 = \mathbf u_0 + \mathbf e_0| \mathbf r_1 = \mathbf u_0 + \mathbf u_1 + \mathbf e_1|$.
According to \cite{Boss95}
we define the join operation $\Join$ by
\begin{equation}y_{i} \Join y_{j} = \mathrm{sign}(y_{i}y_{j})  \min \{|y_{i}|, |y_{j}|\}.
\end{equation}
\begin{remark} 
	The join operation 
	is an approximation of the exact calculation of the probability 
	(see \cite{Boss-eng, Geiselhart, Freudenberger})
	and is sometimes called minsum approximation. 
	For the presented decoding algorithms
this approximation is sufficient.
It can be interpreted as the combination of two reliability values to a new one 
by taking the bit value as the sign of the product and the reliability value of the more unreliable one. 
\end{remark}
In general, for $\ell + 1$ values we get
$$y_{i_0} \Join y_{i_1} \Join \ldots \Join  y_{i_\ell}  = 
\mathrm{sign}( y_{i_0}y_{i_1} \ldots y_{i_\ell})\min \{|y_{i_0}|,|y_{i_1}|, \ldots, |y_{i_\ell}|\}.$$
Note that the join operation is commutative and associative.
Since the modulated code symbols are   $x_\ell= \pm 1$ we have the 
property $(y_i \Join y_j x_\ell) = (y_i  x_\ell \Join y_j) =  x_\ell (y_i \Join y_j)$. 
For vectors  the join operation is done for each coordinate seperately.

The join operation can be used to uncover hidden code words.
For example, if we join the first and the second block of the Plotkin construction, called join-two, we get
 \begin{equation}\label{eq:jpin-two}
 \mathbf y_0 \Join \mathbf y_1 =
 \mathrm{sign}((\mathbf x_0 + \mathbf z_0) \cdot (\mathbf x_0 \mathbf x_1+ \mathbf z_1))\mu 
= \mathrm{sign}(\mathbf x_0 \mathbf x_0 \mathbf x_1+ \mathbf z')\mu = \mathrm{sign}(\mathbf x_1+ \mathbf z')\mu 
 \end{equation} 
where $\mu$ denotes the minimum and $\mathbf z'$ is a combination of the noise values $\mathbf z_0$ and $\mathbf z_1$.
The join-two operation is equivalent to the addition of the hard decision
received blocks
$\mathbf r_0 = \mathbf u_0 + \mathbf e_0 $ and 
$\mathbf r_1 = \mathbf u_0 + \mathbf u_1 + \mathbf e_1$ 
which gives $\mathbf r_0  + \mathbf r_1 = \mathbf u_1 + \mathbf e_0 + \mathbf e_1$.
The result of this join-two operation is that we get a noisy version of the hidden code word 
$\mathbf x_1$. 
An error occurs if the added noise is flipping the sign of the modulated bit $\pm 1$. Thus, 
in case of join-two of two received values
the sign of the product 
is correct if either both symbols are error free or if both symbols are erroneous.
This is equivalent to the hard decision case where an error at the same position
in $\mathbf e_0$ and $\mathbf e_1$ cancels when we add the two errors $\mathbf e_0 + \mathbf e_1$.
The sign of the join operation is wrong if one of the two  values is in error.
If more than two values are joined the sign is correct
if zero or an even number of values are in error
and wrong if an odd number of values are in error. 

Assume that we have four codes $\mathcal C_0$, $\mathcal C_1$, $\mathcal C_2$, and
$\mathcal C_3$ and a double Plotkin construction.
Then, a code word is 
$ \mathbf x = | \mathbf x_0 | \mathbf x_0 \mathbf x_1 | \mathbf x_0 \mathbf x_2 |\mathbf x_0 \mathbf x_1 \mathbf x_2\mathbf x_3 |$
which we transmit over a Gaussian channel with noise
$\mathbf z = |\mathbf z_0|  \mathbf z_1 |\mathbf z_2 |\mathbf z_3|$. We receive $\mathbf y = \mathbf x + \mathbf z = |\mathbf y_0|  \mathbf y_1 |\mathbf y_2 |\mathbf y_3|$.
The basic idea of the presented decoding based on hidden code words
is to use several decoding variants where each variant consists of four decoding steps.
The first decoding step decodes a noisy version of a hidden code word which is uncovered
by a suitable join operation.
Any decoding variant we introduce in Sec. \ref{decvariants} will start with
the decoding of a hidden code word which is from one of the codes 
$\mathcal C_0$, $\mathcal C_1$, $\mathcal C_2$, or $\mathcal C_3$. 
In order to show how hidden code words can be uncovered 
we consider the join-two operation according to Eq. \eqref{eq:jpin-two}. 
For example, $\mathbf y_2 \Join \mathbf y_3 = \mathbf x_1 \mathbf x_3 + \mathbf z'$,
where  $\mathbf z'$ is again a combination of noise values.
It uncovers a noisy version of the hidden code word $\mathbf x_1 \mathbf x_3$ which is 
according to Lemma \ref{subcodedouble}  a code word of $\mathcal C_1$.
There exist six join-two combinations of the four blocks which all uncover different hidden code words shown in 
Table \ref{table:join} where also the hard decision interpretation is included. 
We can decode each of the six combinations by a suitable decoder. 
The join-four $\mathbf y_0 \Join \mathbf y_1 \Join \mathbf y_2 \Join \mathbf y_3 = \mathbf x_3 + \mathbf z'$ uncovers a noisy version of a code word of $\mathcal C_3$.
\begin{table}
	\caption{Soft and Hard Join-Two of the Four Blocks \label{table:join}}
\begin{center}
\begin{tabular}{l|l}
	$\mathbf y_0 \Join \mathbf y_1 = \mathbf x_1 + \mathbf z'$ &  $\mathbf r_0  + \mathbf r_1 = \mathbf u_1 + \mathbf e_0 + \mathbf e_1$ \\[1ex]
\hline

	$\mathbf y_0 \Join \mathbf y_2 = \mathbf x_2 + \mathbf z'$ &  $\mathbf r_0  + \mathbf r_2 = \mathbf u_2 + \mathbf e_0 + \mathbf e_2$ \\[1ex]
\hline

	$\mathbf y_0 \Join \mathbf y_3 = \mathbf x_1\mathbf x_2 \mathbf x_3  + \mathbf z'$ &  $\mathbf r_0  + \mathbf r_3 = \mathbf u_1 +\mathbf u_2 + \mathbf u_3 +  \mathbf e_0 + \mathbf e_3$ \\[1ex]
\hline
	$\mathbf y_1 \Join \mathbf y_2 = \mathbf x_1 \mathbf x_2 + \mathbf z'$ &  $\mathbf r_1  + \mathbf r_2 = \mathbf u_1 + \mathbf u_2 + \mathbf e_1 + \mathbf e_2$ \\[1ex]
\hline
	$\mathbf y_1 \Join \mathbf y_3 = \mathbf x_2\mathbf x_3  + \mathbf z'$ &  $\mathbf r_1  + \mathbf r_3 = \mathbf u_2 +\mathbf u_3 +  \mathbf e_1 + \mathbf e_3$ \\[1ex]
\hline
	$\mathbf y_2 \Join \mathbf y_3 = \mathbf x_1\mathbf x_3  + \mathbf z'$ &  $\mathbf r_2  + \mathbf r_3 = \mathbf u_1 + \mathbf u_3 + \mathbf e_2 + \mathbf e_3$ \\[1ex]
\end{tabular}
\end{center}
\end{table}

If we assume that the decoding in the first step was correct we know one $\mathbf{x}_i$
or a combination $\mathbf{x}_i\mathbf{x}_j$
which can be used for the remaining three decoding steps.
In fact we can uncover some code words twice or four times by different
combinations of the blocks $\mathbf{y}_j$.
We can add these different noisy versions of the same hidden code word
since the modulated bits are identical in any position.
One example is the add-two operation $\mathbf y_0 + \mathbf y_1 {\mathbf{x}}_1$
when $\mathbf{x}_1$ is known.
Since $\mathbf y_0 = \mathbf x_0 + \mathbf z_0 $
and $\mathbf y_1\mathbf x_1  = \mathbf x_0 + \mathbf z_1$
we have twice the modulated code word $\mathbf{x}_0$ disturbed by different noise values.
So both are noisy versions of the same code word $\mathbf{x}_0$
and later we show that the addition of both will cancel errors. 
If only $\mathbf{x}_0$ is not known
we can get four noisy versions of $\mathbf{x}_0$ which can be added
by the add-four operation  
 $\mathbf y_0 + \mathbf y_1\mathbf x_1  + \mathbf y_2 \mathbf x_2 + \mathbf y_3\mathbf x_1\mathbf x_2\mathbf x_3$.
 In case $\mathbf{x}_1$, $\mathbf{x}_2$, and $\mathbf{x}_3$
 are correct,  many errors are canceled as we will see in the following.
If $\mathbf{x}_3$ is known 
we have another possibility: 
we can get  two noisy versions of  $\mathbf{x}_1$, namely,
 $\mathbf y_0 \Join \mathbf y_1$ and 
 $\mathbf y_2 \Join \mathbf y_3 \mathbf{x}_{3}$.
We can add them 
 $(\mathbf y_0 \Join \mathbf y_1) + (\mathbf y_2 \Join \mathbf y_3 \mathbf{x}_{3})$
 which we call the join-add operation and also this operation will cancel errors.
Finally, if  $\mathbf{x}_2$ and  $\mathbf{x}_3$ are known
we can get two noisy versions of $\mathbf{x}_0$ which can be added to obtain a vector.  
Also two noisy versions of $\mathbf{x}_0\mathbf{x}_1$
can be added to a second vector. If we join both vectors 
$(\mathbf y_0 + \mathbf y_2  \mathbf x_2) \Join (\mathbf y_1 + \mathbf y_3  \mathbf x_2\mathbf x_3) $
we get a noisy version of $\mathbf{x}_1$.
The examples for the add and join operations are listed in Table \ref{table:add-join}.
In the decoding variants we will use also further versions of add-join and join-add operations
combining other blocks.

\begin{table}
	\caption{Join and Add Operations of the Four Blocks \label{table:add-join}}
\begin{center}
\begin{tabular}{l|l}
	join-four &  $\mathbf y_0 \Join \mathbf y_1   \Join \mathbf y_2 \Join \mathbf y_3 $\\
\hline
	add-two & $\mathbf y_0 + \mathbf y_2  \mathbf x_2 $\\
\hline
	add-four &  $\mathbf y_0 + \mathbf y_1\mathbf x_1  + \mathbf y_2 \mathbf x_2 + \mathbf y_3\mathbf x_1\mathbf x_2\mathbf x_3$\\
\hline
	join-add  & $(\mathbf y_0 \Join \mathbf y_1) + (\mathbf y_2 \Join \mathbf y_3 \mathbf{x}_{3})$\\
\hline
	add-join  & $(\mathbf y_0 + \mathbf y_2  \mathbf x_2) \Join (\mathbf y_1 + \mathbf y_3  \mathbf x_2\mathbf x_3) $

\end{tabular}
\end{center}
\end{table}

\subsection{Error Cancellation}\label{sec:cancel}
Before we describe the decoding variants, in the next section, we show  
how errors are cancelled applying different add and join combinations.
We will analyze two cases when errors are cancelled.
The first case is the join operation of blocks
and the second is the combination of blocks when some code words $\mathbf x_\ell$ are known
which is the case if they have been already correctly decoded in previous decoding steps.

We start with the join operation. Let us consider the join-four case
\begin{equation*}
\mathbf w = \mathbf y_0 \Join \mathbf y_1 \Join \mathbf y_2 \Join \mathbf y_3.
\end{equation*}
Assume $\tau$ errors are randomly distributed in the four blocks and each block has $n$ positions
and let $\tau_j$ be the number of errors in $\mathbf w$.
An error means that the sign of the position is not equal to the sign of the transmitted symbol.
The join operation of two positions in error cancels the error.
For any position $\ell$ of the $n$ positions of each block the join-four operation is
$y_{0,\ell } \Join y_{1,\ell} \Join  y_{2,\ell}  \Join  y_{3, \ell}$ and
has five possibilities: 
i) all four values $y_{i,\ell }$ are error free,
ii) there is one error in only one of the four values, 
iii) there are two errors in two of the four values,
iv) three errors in three values, and v) all four values are erroneous.
The first error-free case does not influence the number of errors $\tau_j$ in $\mathbf w$.
For the second case the error in one block is also an error in $\mathbf w$.
In the third case the number of errors  in $\mathbf w$ is reduced by $2$ and $\tau_j = \tau -2$
because the error is cancelled. 
The fourth case reduces the number of errors in $\mathbf w$ also by  $2$ since the join of three errors is again
an error but two of them cancel.
Finally, in the fifth case  the number of errors in $\mathbf w$  is reduced by  $4$. 
Therefore, in the  $n$ positions of the  $\mathbf w$ obtained by join-four
are $0 \leq \tau_j = \tau - 2 \cdot i, \ i = 0, 1, \ldots$ errors.
However,  there exists only one possibility for $\tau_j = \tau$, namely, 
 that each of the  $\tau$ errors, no matter in which of the four blocks it lays, is at a different
position (of the $n$).
The probability for this case can be calculated.
\footnote{The calculation is according to the so called Birthday-Paradoxon 
where the name comes from the fact that only 23 persons are needed that the probablity that two 
of them have birthday at the same day (of the year with 365 days) is larger than 1/2.}
The first of the $\tau$ errors has $n$ positions to choose, the second has only $n-1$,
the third $n-2$, and the last has only  $n - (\tau - 1)$ choices. All together there are $n^\tau$
possible choices. Therefore, the probability that all $\tau$ errors are at different positions
(i.e., that there is no error reduction) is
\begin{equation}\label{birthparadox}
	\mathrm{P} (\tau_j = \tau) = \frac{n(n-1)(n-2) \cdots (n-\tau+1)}{n^\tau}.
\end{equation}
The probability in Eq. \eqref{birthparadox} can be approximated using $(n-i)/n = (1- i/n)$ and
$(1-\mu) \approx e^{-\mu}$
\begin{equation}\label{birthapprox}
	\mathrm{P} (\tau_j = \tau) = \prod\limits_{i=0}^{\tau-1} (1 -\frac{i}{n}) \approx 
	\prod\limits_{i=0}^{\tau-1} e^{-\frac{i}{n}} = 
	e^{-\sum_{i=0}^{\tau -1} \frac{i}{n}}
	=e^{-\frac{\tau (\tau - 1)}{2 n}}.
\end{equation}
So the more errors are in the four blocks the smaller becomes the probability that there is no error reduction. 
It is known that the probability is about one half for $\tau \approx \sqrt{n}$.
We can estimate the probability for no error reduction  
for the average number of errors.
Assuming an error probability of $p$, the expected number of errors is $\tau_e = 4 n p$
which gives $\mathrm{P} (\tau_j = \tau_e) \approx e^{- 8 n p^2}$. 

\begin{remark}\label{rem:GCC}
	The error reduction by join-four in the GCC construction can be explained as follows.
	In the first step of GCC decoding we need to decide if
	the received word belongs to the left or to the right 
	part of the partitioning of the inner $(4,4,1)$ code of the graph in Fig. \ref{fig:partdouble}.
In other words, we have to decide whether the weight of the received vector is even or odd.
The received vector can contain zero, one, two, three, or four errors.
For the cases of one or three errors we make a wrong decision, 
	since these errors change the weight of a vector from even to odd and vice versa.
	However, in
the cases of zero, two, and four errors our decision is correct since the even or odd weights stay the same.
\end{remark}

A very important fact is that any error reduction in join-four 
 is also an error reduction in at least one of the six
join-two from Table \ref{table:join}. 
More specific, when two errors are at the same position in two of the four blocks there is
exactly one of the six join-two from Table \ref{table:join} where the number of errors is reduced by two.
An error reduction when three errors are at the same position in three of the four blocks 
lead to an error reduction by two in three of the six join-two cases. Finally,  
if all four blocks have an error at the same position all six join-two have an error reduction by two.
This has proved the following Lemma.
\begin{lemma}[Error Cancellation]\label{genialparadox}
	Any error reduction in join-four leads to an error reduction
	in at least one of the six join-two of Table  \ref{table:join}.
\end{lemma}

Now we consider the case when
having decoded  code words $\mathbf x_\ell$ or combinations of them. 
According to Lemma \ref{subcodedouble} we have several possibilities to combine the different received blocks 
$\mathbf y_i$ in order to decode the other code words $\mathbf x_i$. 
Some of them will cancel errors 
which is indicated by an improved channel situation. 
We start assuming $\mathbf x_1$ is known.
Note that this represents the case that we  have decoded 
 $\mathbf{x}_1$ by joining two blocks (join-two)
\begin{equation}\label{jointwo}
	\hat{\mathbf{x}}_{1} = \Delta_1(\mathbf y_0 \Join \mathbf y_1) = \Delta_1(\mathbf x_1 + \mathbf z').
\end{equation}
If the decoding was correct we know $\mathbf x_1$ and
we can decode  $\hat{\mathbf{x}}_0 = \Delta_0 (\mathbf y_0 + \mathbf y_1 {\mathbf{x}}_1 )$
where the addition (add-two) cancels errors  as shown in Lemma \ref{gain3db}. 
Note that the add-two was already used for decoding in \cite{Boss95} but the gain was not explicitely mentioned. 

\begin{lemma}[$3$ dB Gain]\label{gain3db}
If $\mathbf x_1$
is known the decoding of $\mathbf y_0 + \mathbf y_1 {\mathbf{x}}_1$ for the code $\mathcal C_0$ 
results in a channel with factor $2$ respectively $3$ dB better signal-to-noise ratio than the original channel.
\end{lemma}
\begin{proof}
For any position $i$ of $\mathbf y_0$ we have $y_i=x_i + z_i$ and 
for the same position in $\mathbf y_1 \mathbf x_1$ we have $y'_i=x_i + z'_i$ where the 
$x_i=1$ or $x_i=-1$. Therefore, the addition of both gives  $x_i + z_i+x_i + z'_i= \pm 2 + z_i + z'_i$.
The variance of the addition of two Gaussian random variables with variance $\sigma^2$ is $2 \sigma^2$. 
However, the signal energy is $2^2=4$ and thus the signal-to-noise ratio is increased by a factor of $2$
	which corresponds to a $3$ dB better signal-to-noise ratio.
\end{proof}

Knowing  $\mathbf x_1$
we have another possibility, namely to decode $\mathbf{x}_3$ by join-two
\begin{equation}\label{c3jtwo}
 \hat{\mathbf x}_3=\Delta_3(\mathbf y_2  \Join \mathbf y_3 {\mathbf x}_{1}) = \Delta_3(\mathbf x_3 + \mathbf z').
\end{equation}
For decoding of $\mathbf{x}_3$ we can also use join-four
\begin{equation}\label{joinall}
\hat{\mathbf{x}}_{3} = \Delta_3(\mathbf y_0 \Join \mathbf y_1 \Join \mathbf y_2 \Join \mathbf y_3 ) = \Delta_3(\mathbf x_3 + \mathbf z').
\end{equation}
In the worst case the number of errors is the sum of the errors in two blocks for Eq. \eqref{c3jtwo} 
and in four blocks for Eq. \eqref{joinall}
which means that the number of errors could be doubled or four times more than in one block. 
Therefore, it is obvious that with very high probability the decoding according to Eq. \eqref{c3jtwo} is much better than the one according to Eq. \eqref{joinall}, however, $\mathbf x_1$ must be known.

Assume now that $\mathbf{x}_3$ is known. 
Then we can decode $\mathbf{x}_1$ by join-add
\begin{equation}\label{joinadd}
\hat{\mathbf{x}}_{1} =	\Delta_1((\mathbf y_0 \Join \mathbf y_1) + (\mathbf y_2 \Join \mathbf y_3 \mathbf{x}_{3}))
	= \Delta_1((\mathbf x_1 + \mathbf z') +(\mathbf x_1 + \mathbf z''))
\end{equation}
and the addition cancels errors. However, we can not quantify the gain according to Lemma \ref{gain3db} 
in this case since
$\mathbf z'$  and  $\mathbf z''$ are combinations of Gaussian distributions and we do not know the variance of the sum.
But the error reduction is plausible by the following consideration.
For position $\ell$ we get $ y_{0,\ell} \Join y_{1,\ell} +  y_{2,\ell} \Join y_{3,\ell} x_{3,\ell} $.
If there is one error in four symbols $ y_{0,\ell}, y_{1,\ell},  y_{2,\ell}, y_{3,\ell}$, say at  $ y_{0,\ell} $
then the error disappears if the minimum of $|y_{2,\ell}| $ and $|y_{3,\ell}|$
is larger than the minimum of $|y_{0,\ell}| $ and $|y_{1,\ell}|$.

Next, we assume  that $\mathbf{x}_3$ and $\mathbf{x}_2$ are known. 
Then, $\mathbf{x}_1$ can be decoded by add-join
\begin{equation}\label{addjoin}
\hat{\mathbf{x}}_{1} = \Delta_1((\mathbf y_0 + \mathbf y_2\mathbf{x}_2 ) \Join (\mathbf y_1 + \mathbf y_3 \mathbf{x}_{2} \mathbf{x}_{3})).
\end{equation}
According to  Lemma \ref{gain3db} there is a gain of $3$ dB for both additions which cancels errors 
and afterwards the join-two    
follows.

\begin{remark} The add-join operation in the case when two code words are already decoded 
	could be replaced by 
$(\mathbf y_0  \Join \mathbf y_1) + (\mathbf y_2 \hat{\mathbf{x}}_{2} \Join  \mathbf y_1) +
(\mathbf y_2 \hat{\mathbf{x}}_{2} \Join \mathbf y_3 \hat{\mathbf{x}}_{2}\hat{\mathbf{x}}_{3}) +
(\mathbf y_0  \Join \mathbf y_3 \hat{\mathbf{x}}_{2}\hat{\mathbf{x}}_{3})$.
This is the addition of four join-two terms where each term is $\mathbf x_1 + \mathbf z'$.
However, simulations have shown that the decoding performance gets slightly worse. 
	The difference is that in add-join 
	we have a gain by the addition (Lemma \ref{gain3db}) and then a loss by the 
	join-two and in the other case we have first a loss
in each of the four join-two and then a gain by the addition of the four terms, however,
	the reliability values have been changed by the join-two.
\end{remark}

Finally, assume  that $\mathbf{x}_3$, $\mathbf{x}_2$, and $\mathbf{x}_1$ are known. 
Then, we decode $\mathbf{x}_0$ by add-four
\begin{equation}\label{addall}
\hat{\mathbf{x}}_{0} = \Delta_0(\mathbf y_0 + \mathbf y_1\mathbf{x}_1 + \mathbf y_2 \mathbf{x}_1   + \mathbf y_3 \mathbf{x}_1  \mathbf{x}_{2} \mathbf{x}_{3}).
\end{equation}
We have the addition of four Gaussian random variables and it follows from Lemma \ref{gain3db} that we have a 
gain of $6$ dB.
There exists an interesting connection between the correlation (Eq. \eqref{def:corr}) and the add-four 
which is given in the following Lemma.
\begin{lemma}[Add-Four and Correlation]\label{simplexdoplo}
Given a received vector $\mathbf{y}$ of length $4 n$, the decoding decisions of all four decoders 
with 
$\hat{\mathbf x} = \hat{\mathbf x}_0| \hat{\mathbf x}_0\hat{\mathbf x}_1
|\hat{\mathbf x}_0\hat{\mathbf x}_2 |\hat{\mathbf x}_0  \hat{\mathbf x}_1\hat{\mathbf x}_2\hat{\mathbf x}_3$, 
and the add-four $\mathbf w  = \mathbf y_0 + \hat{\mathbf x}_1\mathbf y_1 + \hat{\mathbf x}_2 \mathbf y_2 
	+ \hat{\mathbf x}_1\hat{\mathbf x}_2\hat{\mathbf x}_3 \mathbf y_3$,
	the correlation $\Phi (\hat{\mathbf x}, \mathbf y)$ is 
 \begin{equation*} \Phi (\hat{\mathbf x}, \mathbf y) = \sum\limits_{i=0}^{n - 1} \hat{x}_{0,i} w_i\end{equation*}
\end{lemma}
\begin{proof}
	The correlation can be written as
	\[\begin{array}{rcl}
		\Phi (\hat{\mathbf x}, \mathbf y) &=& \sum\limits_{i=0}^{n - 1} \hat{x}_{0,i} y_i
	+  \hat{x}_{0,i} \hat{x}_{1,i} y_{n + i}+  \hat{x}_{0,i} \hat{x}_{2,i} y_{2n + i}
+  \hat{x}_{0,i} \hat{x}_{1,i} \hat{x}_{2,i}\hat{x}_{3,i}  y_{3n + i}\\
		&=&  \sum\limits_{i=0}^{n - 1} \hat{x}_{0,i} (y_i
	+  \hat{x}_{1,i} y_{n + i}+  \hat{x}_{2,i} y_{2n + i}
	+   \hat{x}_{1,i} \hat{x}_{2,i} \hat{x}_{3,i}y_{3n + i})\\
		&=&  \sum\limits_{i=0}^{n - 1} \hat{x}_{0,i} w_i
	\end{array}
		\]
because the add-four $\mathbf w= \mathbf y_0 + \hat{\mathbf x}_1\mathbf y_1 + \hat{\mathbf x}_2 \mathbf y_2 
	+ \hat{\mathbf x}_1\hat{\mathbf x}_2\hat{\mathbf x}_3 \mathbf y_3$
	has the terms
$w_i =y_i
	+  \hat{x}_{1,i} y_{n + i}+  \hat{x}_{2,i} y_{2n + i}
	+   \hat{x}_{1,i} \hat{x}_{2,i} \hat{x}_{3,i}y_{3n + i}$.
\end{proof}

Whenever we need to calculate the correlation of a final decoding decision we can use 
the result of Lemma \ref{simplexdoplo}.
The effect of all the gains and losses of the described add and join
combinations are simulated in Ex. \ref{ex:gains}.
\begin{example}[Error Cancellation]\label{ex:gains}
We will simulate the expected values for the error cancellations  
of the introduced operations.
We assume BPSK and a  Gaussian channel at $2$ dB signal-to-noise ratio with a code of rate one half.
The channel error probability is for this case $0.1041$. 
	It is assumed that all used $\mathbf{x}_i$ are known 
	which makes the results  independent of the used codes.
	The following table shows the simulated probabilities. 
	\vskip5pt
	\begin{center}
		\begin{tabular}{c|c|c|c|c|c|c}
			Channel & join-two & join-four & join-add & add-join & add-two & add-four\\
			$2$~dB  & Eq.~\eqref{jointwo} & Eq.~\eqref{joinall} & Eq.~\eqref{joinadd}& Eq.~\eqref{addjoin}&  Lemma \ref{gain3db} & Eq.~\eqref{addall}\\
			\hline
			$0.1041$ & $0.1872$ & $0.3036$ & $0.1006$ & $0.0725$ & $0.039$ & $0.0056$
		\end{tabular}
	\end{center}
	\vskip5pt
According to Eq.~\eqref{birthparadox} or Eq.~\eqref{birthapprox} the error reduction of the join-four should be significant. 
In the worst case for join-four the error probability would be four times the channel error probability which is $0.4164$. 
But it is  $0.3036$ which is about $25 \%$ less errors (each cancellation reduces by two errors).
	The join-two, Eq.~\eqref{jointwo}, which in the worst case 
	would have the doubled error probability of the channel $0.2082$,  has the value $0.1872$ 
	due to cancellation which is about $10 \%$ less errors.
Note that these are expected values and in the six cases of Table \ref{table:join} there will be a variance.

	Even there is no analytical derivation for join-add, Eq.~\eqref{joinadd}, 
	we have an error probability of  $0.1006$ which means the addition
	must reduce errors since each of the join-two has an error probability of  $0.1872$.
	Therefore, the addition reduces by $0.0866$.
	This is comparable to  Lemma \ref{gain3db} which reduces the error probability from
	the channel $0.1041$ to $0.039$ by $0.0651$.
	That errors may be cancelled is also observed for add-join which is less than two times add-two.
	According to Remark \ref{rem:GCC} we could calculate the join-four probability by
	the expected number of errors in the four bits of the inner code.
	The probability for $\tau$ errors is $P(\tau) = \binom{4}{\tau} p^\tau (1-p)^{n-\tau}$.
	Now the decision of the weight of the received four bits is wrong
	if  one or three errors have ocurred. We get $P(1) + P(3) = 0.3034$ which has to be compared to the 
	simulated join-four which is  $0.3036$. Thus, both interpretations lead to the same result.

	In Fig.\ref{fig:jastat} the gains and losses are graphically depicted.
	Note that the gain of add-four is extremely large which indicates that the code $\mathcal{C}_0$
	has to decode very little errors. 
\end{example}

\begin{figure}[htb]
	\begin{center}
	\includegraphics[width=0.65\textwidth]{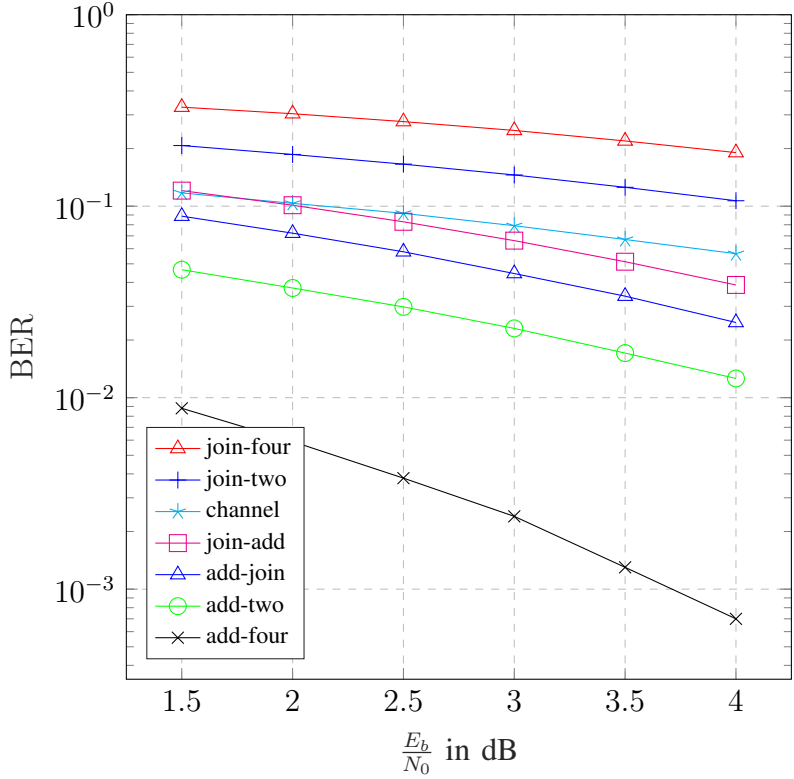}
	\end{center}
	\caption{BER for Different (join, add) Combinations}\label{fig:jastat}
\end{figure}
Note that the simulated BER of Fig. \ref{fig:jastat} are expected values and due to the variance 
there are cases with smaller and larger error reduction. Namely,
$(\mathbf y_0 + \mathbf y_2\mathbf{x}_2 ) \Join (\mathbf y_1 + \mathbf y_3 \mathbf{x}_{2} \mathbf{x}_{3})$ and 
$(\mathbf y_0 + \mathbf y_1\mathbf{x}_1 ) \Join (\mathbf y_2 + \mathbf y_3 \mathbf{x}_{1} \mathbf{x}_{3})$
will have a different error reduction with high probability.
Further, combining different blocks is to combine different channel variabilities. 
Consider the two combinations $y_{0,\ell} \Join y_{1,\ell}$ and $y_{0,\ell} \Join y_{2,\ell}$
where $y_{0,\ell}$ is erroneous with large absolute value.
Both others are correct.
If the absolute value of $y_{1,\ell}$ is much smaller than the one of $y_{2,\ell}$
in the first join-two the reliability of an error position is smaller than in the second join-two
and this might be an advantage in soft-decision decoding.
These facts will be exploited when using different variants for the decoding.

\subsection{Decoding Variants}\label{decvariants}
A decoding variant is a particular decoder realization which calculates a code word, given a received vector  $\mathbf{y}$.
Assume that we have four codes 
and a code word 
$\mathbf x_0 | \mathbf x_0 \mathbf x_1 |\mathbf x_0 \mathbf x_2 | \mathbf x_0 \mathbf x_1 \mathbf x_2\mathbf x_3|$
of a double Plotkin construction.
The classical decoder would start decoding $\mathcal{C}_3$ and then the other codes.
But it is also possible to decode $\mathcal{C}_2$ first and then to continue with the others.
Similarly, we can start with $\mathcal{C}_1$.
Even starting the decoding with  $\mathcal{C}_0$ is possible.
Any variant will combine  different received blocks 
$\mathbf y_i$ to uncover hidden code words (see Lemma \ref{subcodedouble}) and decode them. 
Thus, any variant calculates an estimated code word
	$\hat{\mathbf{x}} =\hat{\mathbf{x}}_{0}|\hat{\mathbf{x}}_{0}\hat{\mathbf{x}}_{1}|\hat{\mathbf{x}}_{0}\hat{\mathbf{x}}_{2}|\hat{\mathbf{x}}_{0}\hat{\mathbf{x}}_{1}\hat{\mathbf{x}}_{2}\hat{\mathbf{x}}_{3} $.
With this estimated code word the correlation 	$\Phi(\hat{\mathbf{x}}, \mathbf{y})$ can be  calculated and
 used as a measure in order to compare the results of the different variants.
 The decision of the variant which has the largest correlation is choosen as the final decoding decision.
We have shown in the last section that when a decoding decision was made then  the decoding of
the other codes benefits from error cancellation.
The idea is that some variant will decode correctly.  

In the following, we will describe  the possible decoding variants.
For most variants $\mathbf{x}_0$ is decoded in the last step by		
		$$\hat{\mathbf{x}}_{0}=\Delta_0(\mathbf y_0 + \mathbf y_1 \hat{\mathbf{x}}_{1} +\mathbf y_2 \hat{\mathbf{x}}_{2} + \mathbf y_3 \hat{\mathbf{x}}_{1} \hat{\mathbf{x}}_{2}  \hat{\mathbf{x}}_{3})$$
and we will omit this equation in the descriptions. 
However, in these cases we can use Lemma \ref{simplexdoplo} in order to calculate the correlation
exploiting that the join-four is known already.


\subsubsection{Variants Starting Decoding with $\mathcal C_3$}\label{sec:startc3} 
We will denote the three possible variants by V$(4, \mathbf y_0 \Join \mathbf y_1)$,
V$(4, \mathbf y_0 \Join \mathbf y_2)$, and V$(4, \mathbf y_1 \Join \mathbf y_2)$ in order to indicate
that all three variants start to decode $\mathcal C_3$ by the join-four
\begin{equation*}
	\hat{\mathbf{x}}_{3,i}=\Lambda_3(\mathbf y_0 \Join \mathbf y_1 \Join \mathbf y_2 \Join \mathbf y_3 )
\end{equation*}
which possibly provides a list. The variants differ in the second and third decoding step.
Note that for this decoding variants no subcode structure is necessary.

\begin{itemize}

	\item 	V$(4, \mathbf y_0 \Join \mathbf y_2)$: 
		
		$\hat{\mathbf{x}}_{2}=\Delta_2((\mathbf y_0 \Join \mathbf y_2) + (\mathbf y_1 \Join \mathbf y_3 \hat{\mathbf{x}}_{3}))$, \ 
				$\hat{\mathbf{x}}_{1}=\Delta_1((\mathbf y_0 + \mathbf y_2 \hat{\mathbf{x}}_{2}) \Join (\mathbf y_1 + \mathbf y_3 \hat{\mathbf{x}}_{2}\hat{\mathbf{x}}_{3}))$, \ $\hat{\mathbf{x}}_{0}$. 

	\item 	V$(4, \mathbf y_0 \Join \mathbf y_1)$: 
		
		$\hat{\mathbf{x}}_{1}=\Delta_1((\mathbf y_0 \Join \mathbf y_1) + (\mathbf y_2 \Join \mathbf y_3 \hat{\mathbf{x}}_{3}))$, \ 
		$\hat{\mathbf{x}}_{2}=\Delta_2((\mathbf y_0 + \mathbf y_1 \hat{\mathbf{x}}_{1}) \Join (\mathbf y_2 + \mathbf y_3 \hat{\mathbf{x}}_{1}\hat{\mathbf{x}}_{3}))$, \ $\hat{\mathbf{x}}_{0}$. 

	\item 	V$(4, \mathbf y_1 \Join \mathbf y_2)$: 
		
		$(\hat{\mathbf{x}}_{1}\hat{\mathbf{x}}_{2})=\Delta_1((\mathbf y_1 \Join \mathbf y_2) + (\mathbf y_0 \Join \mathbf y_3 \hat{\mathbf{x}}_{3}))$, \ 
				$\hat{\mathbf{x}}_{1} =\Delta_1((\mathbf y_0 + \mathbf y_3 (\hat{\mathbf{x}}_{1}\hat{\mathbf{x}}_{2})\hat{\mathbf{x}}_{3}) \Join (\mathbf y_1 + \mathbf y_2 (\hat{\mathbf{x}}_{1}\hat{\mathbf{x}}_{2}))$, \ 
	$\hat{\mathbf{x}}_{2}=(\hat{\mathbf{x}}_{1}\hat{\mathbf{x}}_{2})\hat{\mathbf{x}}_{1} $, \ $\hat{\mathbf{x}}_{0}$.
\end{itemize}	
The variant 
V$(4, \mathbf y_0 \Join \mathbf y_2)$ corresponds to the
 recursive decoding in \cite{Boss95}. Further, 
 $\mathrm V(4, \mathbf y_1 \Join \mathbf y_3)=\mathrm V(4, \mathbf y_0 \Join \mathbf y_2)$, 
 $\mathrm V(4, \mathbf y_2 \Join \mathbf y_3)=\mathrm V(4, \mathbf y_0 \Join \mathbf y_1)$, and 
 $\mathrm V(4, \mathbf y_0 \Join \mathbf y_3)=\mathrm V(4, \mathbf y_1 \Join \mathbf y_2)$. 
This is because the second decoding step
of  V$(4, \mathbf y_1 \Join \mathbf y_3)$ and V$(4, \mathbf y_0 \Join \mathbf y_2)$ is the same which also holds for the other cases. 

\subsubsection{Variants Starting Decoding with $\mathcal C_1$, $\mathcal C_2$}\label{sec:starteizw} 
We will denote the six possible variants by V$(\mathbf y_i \Join \mathbf y_j)$ 
where $i$ and $j$ are the two blocks we join for the first decoding step  which could be list decoding.
We assume that $\mathcal C_3 \subset \mathcal C_2 \subseteq \mathcal C_1$. In all six variants
the second decoding step decodes $ \hat{\mathbf x}_3$ and the last decoding step $ \hat{\mathbf x}_0$. 
Note that the code $\mathcal C_0$ can possibly be independent from the other codes.

\begin{itemize}
	\item 	V$(\mathbf y_0 \Join \mathbf y_1)$:  \quad
				$\hat{\mathbf x}_{1,i}=\Lambda_1(\mathbf y_0 \Join \mathbf y_1)$,

		$ \hat{\mathbf x}_3=\Delta_3(\mathbf y_2  \Join \mathbf y_3 \hat{\mathbf x}_{1})$, \ 
		$ \hat{\mathbf x}_2=\Delta_2((\mathbf y_0  + \mathbf y_1 \hat{\mathbf x}_{1}) \Join (\mathbf y_2 +\mathbf y_3 \hat{\mathbf x}_{1} \hat{\mathbf x}_3))$, \ $\hat{\mathbf{x}}_{0}$. 

	\item 	V$(\mathbf y_0 \Join \mathbf y_2)$: \quad 
		$\hat{\mathbf x}_{2,i}=\Lambda_2(\mathbf y_0 \Join \mathbf y_2)$, 

		$ \hat{\mathbf x}_3=\Delta_3(\mathbf y_1  \Join \mathbf y_3 \hat{\mathbf x}_{2})$, \ 
		$ \hat{\mathbf x}_1=\Delta_1((\mathbf y_0  + \mathbf y_2 \hat{\mathbf x}_{2}) \Join (\mathbf y_1 +\mathbf y_3 \hat{\mathbf x}_{2} \hat{\mathbf x}_3))$, \  $\hat{\mathbf{x}}_{0}$. 

	\item 	V$(\mathbf y_0 \Join \mathbf y_3)$:  \quad
		$(\hat{\mathbf x}_{1,i}\hat{\mathbf x}_{2,i}\hat{\mathbf x}_{3,i})=\Lambda_1(\mathbf y_0 \Join \mathbf y_3)$,

		$ \hat{\mathbf x}_3=\Delta_3(\mathbf y_1  \Join \mathbf y_2 (\hat{\mathbf x}_{1}\hat{\mathbf x}_{2}\hat{\mathbf x}_{3}))$, $(\hat{\mathbf x}_{1}\hat{\mathbf x}_{2}) = (\hat{\mathbf x}_{1}\hat{\mathbf x}_{2}\hat{\mathbf x}_{3})\hat{\mathbf x}_{3}$,
		$ \hat{\mathbf x}_1=\Delta_1((\mathbf y_0  + \mathbf y_3 (\hat{\mathbf x}_{1}\hat{\mathbf x}_{2}\hat{\mathbf x}_{3})) \Join (\mathbf y_1 +\mathbf y_2 (\hat{\mathbf x}_{1} \hat{\mathbf x}_{2})))$, \ $\hat{\mathbf{x}}_{2}=(\hat{\mathbf{x}}_{1}\hat{\mathbf{x}}_{2})\hat{\mathbf{x}}_{1} $, \  $\hat{\mathbf{x}}_{0}$.

	\item 	V$(\mathbf y_1 \Join \mathbf y_2)$:  \quad
		$(\hat{\mathbf x}_{1,i}\hat{\mathbf x}_{2,i})=\Lambda_1(\mathbf y_1 \Join \mathbf y_2)$, 

		$ \hat{\mathbf x}_3=\Delta_3(\mathbf y_0  \Join \mathbf y_3 (\hat{\mathbf x}_{1}\hat{\mathbf x}_{2}))$, \ 
		$ \hat{\mathbf x}_1=\Delta_1((\mathbf y_0  + \mathbf y_3 (\hat{\mathbf x}_{1}\hat{\mathbf x}_{2})\hat{\mathbf x}_{3}) \Join (\mathbf y_1 +\mathbf y_2 (\hat{\mathbf x}_{1} \hat{\mathbf x}_{2})))$, \  $\hat{\mathbf{x}}_{2}=(\hat{\mathbf{x}}_{1}\hat{\mathbf{x}}_{2})\hat{\mathbf{x}}_{1} $, $\hat{\mathbf{x}}_{0}$.
	\item 	V$(\mathbf y_1 \Join \mathbf y_3)$: \quad 
		$(\hat{\mathbf x}_{2,i}\hat{\mathbf x}_{3,i})=\Lambda_2(\mathbf y_1 \Join \mathbf y_3)$,

		$ \hat{\mathbf x}_3=\Delta_3(\mathbf y_0  \Join \mathbf y_2 (\hat{\mathbf x}_{2}\hat{\mathbf x}_{3}))$,  
		$\hat{\mathbf{x}}_{2}=(\hat{\mathbf x}_{2}\hat{\mathbf x}_{3})\hat{\mathbf{x}}_{3} $,  
		$ \hat{\mathbf x}_1=\Delta_1((\mathbf y_0  + \mathbf y_2 \hat{\mathbf x}_{2}) \Join (\mathbf y_1 +\mathbf y_3 (\hat{\mathbf x}_{2} \hat{\mathbf x}_{3})))$,   $\hat{\mathbf{x}}_{0}$.

	\item 	V$(\mathbf y_2 \Join \mathbf y_3)$:  \quad
		$(\hat{\mathbf x}_{1,i}\hat{\mathbf x}_{3,i})=\Lambda_1(\mathbf y_2 \Join \mathbf y_3)$, 

		$ \hat{\mathbf x}_3=\Delta_3(\mathbf y_0  \Join \mathbf y_1 (\hat{\mathbf x}_{1}\hat{\mathbf x}_{3}))$,  
$\hat{\mathbf{x}}_{1}=(\hat{\mathbf x}_{1}\hat{\mathbf x}_{3})\hat{\mathbf{x}}_{3} $, 
		$ \hat{\mathbf x}_2=\Delta_1((\mathbf y_0  + \mathbf y_1 \hat{\mathbf x}_{1}) \Join (\mathbf y_2 +\mathbf y_3 (\hat{\mathbf x}_{1} \hat{\mathbf x}_{3})))$,   $\hat{\mathbf{x}}_{0}$.
	
\end{itemize}
Another interpretation is possible when combining two variants. We get three
codes by combining 
V$(\mathbf y_0 \Join \mathbf y_1)$ and  V$(\mathbf y_2 \Join \mathbf y_3)$, 
V$(\mathbf y_0 \Join \mathbf y_2)$ and  V$(\mathbf y_1 \Join \mathbf y_3)$, and  
V$(\mathbf y_1 \Join \mathbf y_2)$ and  V$(\mathbf y_0 \Join \mathbf y_3)$. 
The first gives a codeword  $\mathbf x_1|\mathbf x_1 \mathbf x_3$, the second 
 $\mathbf x_2|\mathbf x_2 \mathbf x_3$, and the last
  $\mathbf x_1\mathbf x_2|\mathbf x_1 \mathbf x_2\mathbf x_3$ which are Plotkin constructions.
  Let us consider the case V$(\mathbf y_0 \Join \mathbf y_1)$ and  V$(\mathbf y_2 \Join \mathbf y_3)$
  with a list decoder.
  This can be written as
$$
\Lambda_1(\mathbf y_0 \Join \mathbf y_1)| \Lambda_1(\mathbf y_2 \Join \mathbf y_3) = \{ \hat{\mathbf x}_{1,i} \} | \{(\hat{\mathbf x}_{1,j}\hat{\mathbf x}_{3,j})\}.
$$
Clearly, the addition of a codeword from the list 
$\{ \hat{\mathbf x}_{1,i} \}$ and from the list
$\{ (\hat{\mathbf x}_{1,j}\hat{\mathbf x}_{3,j})\}$
must be a code word from $\mathcal C_3$. For modulated code words the multiplication
can be checked by the indicator function
$\Gamma_3(\hat{\mathbf x}_{1,i} (\hat{\mathbf x}_{1,j}\hat{\mathbf x}_{3,j}))$.
Note that the addition of two codewords of $\mathcal C_1$
is a code word of $\mathcal C_1$ and not necessarily a code word of $\mathcal C_3$.
But since it is a Plotkin constructed code word the addition must be a code word of $\mathcal C_3$
otherwise it is not a valid combination.
This decoding strategy will also be used in the next section.

\begin{remark}\label{zerocwtrick}
The calculation for the third decoding step 
can be modified which is described for variant V$(\mathbf y_0 \Join \mathbf y_1)$. 
We know for the third step  $\hat{\mathbf x}_{1}$ and $\hat{\mathbf x}_{3}$ from the two previous steps
and decode $\hat{\mathbf x}_{2}$ by
		$ \hat{\mathbf x}_2=\Delta_2((\mathbf y_0  + \mathbf y_1 \hat{\mathbf x}_{1}) \Join (\mathbf y_2 +\mathbf y_3 \hat{\mathbf x}_{1} \hat{\mathbf x}_3))$.
		We could also calculate
	$ \mathbf y_0  \Join \mathbf y_1 \hat{\mathbf x}_{1}$ and
	$ \mathbf y_2  \Join \mathbf y_3 \hat{\mathbf x}_{1} \hat{\mathbf x}_3$ instead
	which both must be the all zero-code word.
Therefore, all positions with values less than zero are erroneous.
However, we know already the erroneous positions of both terms 
	from the decoding decisions $\hat{\mathbf x}_{1}$ and $\hat{\mathbf x}_{3}$ of the first two steps.
	The addition 	$ (\mathbf y_0  \Join \mathbf y_1 \hat{\mathbf x}_{1})
	+ (\mathbf y_2  \Join \mathbf y_3 \hat{\mathbf x}_{1} \hat{\mathbf x}_3)$ must also be the all-zero
code word and several errors are cancelled. If the errors which are not cancelled
would be the same as in
		$ (\mathbf y_0  + \mathbf y_1 \hat{\mathbf x}_{1}) \Join (\mathbf y_2 +\mathbf y_3 \hat{\mathbf x}_{1} \hat{\mathbf x}_3)$ we would have a very simple decoder for $\hat{\mathbf x}_{2}$.
		But simulations have shown that only in very few cases the uncancelled error positions are identical for the two
		calculations.
Nevertheless, this interpretation could possibly support the decoding 
of $\hat{\mathbf x}_{2}$. The improvement for such a support in decoding was not studied and remains an open problem.
\end{remark}

\subsubsection{Variants Starting Decoding with $\mathcal C_0$}\label{sec:startnull}
We assume that $\mathcal C_3 \subset \mathcal C_2 \subseteq \mathcal C_1 \subset \mathcal C_0$.
The variant V$(\mathbf y_0, \mathbf y_1)$ uses list decoders $\Lambda_0$ for $\mathcal C_0$
for the received blocks $\mathbf y_0$  and $\mathbf y_1$.
Let 
$$
\Lambda_0 (\mathbf y_0) = \{ \hat{\mathbf x}_{0,1}, \hat{\mathbf x}_{0,2}, \ldots, \hat{\mathbf x}_{0,L}\}
$$
and
$$
\Lambda_0 (\mathbf y_1) = \{ (\hat{\mathbf x}_{0,1}\hat{\mathbf x}_{1,1}), (\hat{\mathbf x}_{0,2}\hat{\mathbf x}_{1,2}), \ldots, (\hat{\mathbf x}_{0,L}\hat{\mathbf x}_{1,L})\}
$$
where $(\hat{\mathbf x}_{0,j}\hat{\mathbf x}_{1,j})$ is a code word from $\mathcal C_0$.
Now, we can check with the indicator function
if the addition of the code words of the two lists is a code word of $\mathcal C_1$
\begin{equation*}\Gamma_1 (\hat{\mathbf x}_{0,i}(\hat{\mathbf x}_{0,j}\hat{\mathbf x}_{1,j}) )= \left\{ \begin{array}{l} 
0, \\ 
1 .
\end{array} \right. 
\end{equation*}
Only solutions with indicator equal $1$ are valid solutions.
For any valid solution we have a pair $\hat{\mathbf{x}}_{0}$ and $\hat{\mathbf{x}}_{1}$ with which we can continue.
This leads to the following six variants.

\begin{itemize}
	\item 	V$(\mathbf y_0, \mathbf y_1)$: \quad $\Lambda_0(\mathbf y_0)$,  $\Lambda_0(\mathbf y_1)$,
		$\Gamma_1 (\hat{\mathbf x}_{0,i}(\hat{\mathbf x}_{0,j}\hat{\mathbf x}_{1,j}) )$

			$ \hat{\mathbf x}_3=\Delta_3(\mathbf y_2  \Join \mathbf y_3 \hat{\mathbf x}_{1})$, \
		$ \hat{\mathbf x}_2=\Delta_2((\mathbf y_0  + \mathbf y_1 \hat{\mathbf x}_{1}) \Join (\mathbf y_2 +\mathbf y_3 \hat{\mathbf x}_{1} \hat{\mathbf x}_3))$.
	
	\item 	V$(\mathbf y_0, \mathbf y_2)$: \quad $\Lambda_0(\mathbf y_0)$,  $\Lambda_0(\mathbf y_2)$,
		$\Gamma_2 (\hat{\mathbf x}_{0,i}(\hat{\mathbf x}_{0,j}\hat{\mathbf x}_{2,j}) )$

		$ \hat{\mathbf x}_3=\Delta_3(\mathbf y_1  \Join \mathbf y_3 \hat{\mathbf x}_{2})$, \ 
		$ \hat{\mathbf x}_1=\Delta_1((\mathbf y_0  + \mathbf y_2 \hat{\mathbf x}_{2}) \Join (\mathbf y_1 +\mathbf y_3 \hat{\mathbf x}_{2} \hat{\mathbf x}_3))$.

\item 	V$(\mathbf y_0, \mathbf y_3)$: \quad $\Lambda_0(\mathbf y_0)$,  $\Lambda_0(\mathbf y_3)$,
		$\Gamma_1 (\hat{\mathbf x}_{0,i}(\hat{\mathbf x}_{0,j}\hat{\mathbf x}_{1,j}\hat{\mathbf x}_{2,j}\hat{\mathbf x}_{3,j}) )$  

	$ \hat{\mathbf x}_3=\Delta_3(\mathbf y_1  \Join \mathbf y_2 (\hat{\mathbf x}_{1}\hat{\mathbf x}_{2}\hat{\mathbf x}_{3}))$, $(\hat{\mathbf x}_{1}\hat{\mathbf x}_{2}) = (\hat{\mathbf x}_{1}\hat{\mathbf x}_{2}\hat{\mathbf x}_{3})\hat{\mathbf x}_{3}$,
		$ \hat{\mathbf x}_1=\Delta_1((\mathbf y_0  + \mathbf y_3 (\hat{\mathbf x}_{1}\hat{\mathbf x}_{2}\hat{\mathbf x}_{3})) \Join (\mathbf y_1 +\mathbf y_2 (\hat{\mathbf x}_{1} \hat{\mathbf x}_{2})))$, \ $\hat{\mathbf{x}}_{2}=(\hat{\mathbf{x}}_{1}\hat{\mathbf{x}}_{2})\hat{\mathbf{x}}_{1} $.

	\item 	V$(\mathbf y_1, \mathbf y_2)$: \quad $\Lambda_0(\mathbf y_1)$,  $\Lambda_0(\mathbf y_2)$,
		$\Gamma_1 ((\hat{\mathbf x}_{0,i} \hat{\mathbf x}_{1,i})(\hat{\mathbf x}_{0,j}\hat{\mathbf x}_{2,j}) )$
		
		$ \hat{\mathbf x}_3=\Delta_3(\mathbf y_0  \Join \mathbf y_3 (\hat{\mathbf x}_{1}\hat{\mathbf x}_{2}))$, \ 
		$ \hat{\mathbf x}_1=\Delta_1((\mathbf y_0  + \mathbf y_3 (\hat{\mathbf x}_{1}\hat{\mathbf x}_{2})\hat{\mathbf x}_{3}) \Join (\mathbf y_1 +\mathbf y_2 (\hat{\mathbf x}_{1} \hat{\mathbf x}_{2})))$, \  $\hat{\mathbf{x}}_{2}=(\hat{\mathbf{x}}_{1}\hat{\mathbf{x}}_{2})\hat{\mathbf{x}}_{1} $, \   $\hat{\mathbf{x}}_{0}=(\hat{\mathbf{x}}_{0}\hat{\mathbf{x}}_{1})\hat{\mathbf{x}}_{1} $.

	\item 	V$(\mathbf y_1, \mathbf y_3)$: \quad  $\Lambda_0(\mathbf y_1)$,  $\Lambda_0(\mathbf y_3)$,
		$\Gamma_2 ((\hat{\mathbf x}_{0,i}\hat{\mathbf x}_{1,i})(\hat{\mathbf x}_{0,j}\hat{\mathbf x}_{1,j}\hat{\mathbf x}_{2,j}\hat{\mathbf x}_{3,j}) )$ 

		$ \hat{\mathbf x}_3=\Delta_3(\mathbf y_0  \Join \mathbf y_2 (\hat{\mathbf x}_{2}\hat{\mathbf x}_{3}))$,  
		$\hat{\mathbf{x}}_{2}=(\hat{\mathbf x}_{2}\hat{\mathbf x}_{3})\hat{\mathbf{x}}_{3} $,  
		$ \hat{\mathbf x}_1=\Delta_1((\mathbf y_0  + \mathbf y_2 \hat{\mathbf x}_{2}) \Join (\mathbf y_1 +\mathbf y_3 (\hat{\mathbf x}_{2} \hat{\mathbf x}_{3})))$, \   $\hat{\mathbf{x}}_{0}=(\hat{\mathbf{x}}_{0}\hat{\mathbf{x}}_{1})\hat{\mathbf{x}}_{1} $  .

	\item 	V$(\mathbf y_2, \mathbf y_3)$: \quad $\Lambda_0(\mathbf y_2)$,  $\Lambda_0(\mathbf y_3)$,
		$\Gamma_1 ((\hat{\mathbf x}_{0,i}\hat{\mathbf x}_{2,i})(\hat{\mathbf x}_{0,j}\hat{\mathbf x}_{1,j}\hat{\mathbf x}_{2,j}\hat{\mathbf x}_{3,j}) )$

		$ \hat{\mathbf x}_3=\Delta_3(\mathbf y_0  \Join \mathbf y_1 (\hat{\mathbf x}_{1}\hat{\mathbf x}_{3})$, \ 
$\hat{\mathbf{x}}_{1}=(\hat{\mathbf x}_{1}\hat{\mathbf x}_{3})\hat{\mathbf{x}}_{3} $, \ 
		$ \hat{\mathbf x}_2=\Delta_1((\mathbf y_0  + \mathbf y_1 \hat{\mathbf x}_{1}) \Join (\mathbf y_2 +\mathbf y_3 (\hat{\mathbf x}_{1} \hat{\mathbf x}_{3})))$, \   $\hat{\mathbf{x}}_{0}=(\hat{\mathbf{x}}_{0}\hat{\mathbf{x}}_{2})\hat{\mathbf{x}}_{2} $.

\end{itemize}

The decoding strategy  starting with $\mathcal C_0$ can 
be modified to get a list decoder or a decoder as follows. 
We will denote the decoding variant by
V$((\mathbf y_0, \mathbf y_1), (\mathbf y_2, \mathbf y_3))$. 
In the first step we decode each block with a list decoder $\Lambda_0$ 
for $\mathcal C_0$.
With the list decoders
 $\Lambda_0(\mathbf y_0)$,  $\Lambda_0(\mathbf y_1)$,  $\Lambda_0(\mathbf y_2)$, and
 $\Lambda_0(\mathbf y_3)$ we get the four lists
$\{\hat{\mathbf x}_{0,j}\}$,
$\{(\hat{\mathbf x}_{0,j}\hat{\mathbf x}_{1,j})\}$,
$\{(\hat{\mathbf x}_{0,j}\hat{\mathbf x}_{2,j})\}$, and
$\{(\hat{\mathbf x}_{0,j}\hat{\mathbf x}_{1,j} \hat{\mathbf x}_{2,j}\hat{\mathbf x}_{3,j})\}$.
In the second step we use the indicator function $\Gamma_1$ 
for the sum of code words from the first and second list
and for the sum of the third and fourth list which is
$\hat{\mathbf x}_{0,i}(\hat{\mathbf x}_{0,j}\hat{\mathbf x}_{1,j})$
and 
$(\hat{\mathbf x}_{0,i}\hat{\mathbf x}_{2,i})(\hat{\mathbf x}_{0,j}\hat{\mathbf x}_{1,j}\hat{\mathbf x}_{2,j}\hat{\mathbf x}_{3,j})$.
For all cases  where the indicator function is one we have a code word
of  a Plotkin construction, namely
$\mathbf x_0 | \mathbf x_0 \mathbf x_1$ 
from the first two lists and
$\mathbf x_0 \mathbf x_2 | \mathbf x_0 \mathbf x_2 \mathbf x_1 \mathbf x_3 $ 
from the second two lists.
Both are code words of a Plotkin construction with the 
codes $\mathcal C_0$ and $\mathcal C_1$.
Thus, with these valid code words we have now two new lists, namely 
$\{\hat{\mathbf x}_{0,i}|\hat{\mathbf x}_{0,i}\hat{\mathbf x}_{1,i}\}$
and $\{ \hat{\mathbf x}_{0,j}\hat{\mathbf x}_{2,j}|\hat{\mathbf x}_{0,j}\hat{\mathbf x}_{1,j}\hat{\mathbf x}_{2,j}\hat{\mathbf x}_{3,j} \}$.
The last step is to use the indicator function $\Gamma_3$  
for the sum of the four blocks which must be a code word of the code $\mathcal C_3$
and is for the modulated code words
$\hat{\mathbf x}_{0,i} \hat{\mathbf x}_{0,i}\hat{\mathbf x}_{1,i}
 \hat{\mathbf x}_{0,j}\hat{\mathbf x}_{2,j} \hat{\mathbf x}_{0,j}\hat{\mathbf x}_{1,j}\hat{\mathbf x}_{2,j}\hat{\mathbf x}_{3,j}$.
When the indicator function is one we have a valid code word of the double Plotkin construction.
Thus, we have a list of possible decoding decisions.
Therefore, the described decoder is a list decoder for a double Plotkin construction
\begin{equation}\label{dplistdec}
	\Lambda_{dP} = \mathrm{V}((\mathbf y_0, \mathbf y_1), (\mathbf y_2, \mathbf y_3)).
\end{equation}
When selecting the solution of the list with the largest correlation we have
a decoder.

The disadvantage of this decoder 
V$((\mathbf y_0, \mathbf y_1), (\mathbf y_2, \mathbf y_3))$  is that
all lists $\Lambda_0(\mathbf y_i)$
must contain the correct solution, otherwise
the decoding can not be correct.
By using several variants V$(\mathbf y_i, \mathbf y_j)$
only the two lists of one variant must contain the correct solutions
which is a much weaker condition.

\subsection{Recursive Decoding}\label{sec:recdec}
For RM codes it is possible to recursively apply the double Plotkin construction
which can be seen also from the illustration in Fig. \ref{fig:rm-tree}.
Namely, a  $\mathcal R(r,m), 2 \leq r \leq m-2$ code is a double Plotkin construction of the codes   
$\mathcal C_0 = \mathcal R(r,m-2)$, 
$\mathcal C_1 = \mathcal C_2 =  \mathcal R(r-1,m-2)$, 
and $\mathcal C_3 = \mathcal R(r-2,m-2)$.
In order to decode for example $\mathcal R(r-1,m-2)$
we can again use the double Plotkin construction with
$\mathcal C'_0 = \mathcal R(r-1,m-4)$, 
$\mathcal C'_1 = \mathcal C'_2 =  \mathcal R(r-2,m-4)$, 
and $\mathcal C'_3 = \mathcal R(r-3,m-4)$. This can be continued
until the codes have short enough length to be decoded by correlation. 

According to Sec. \ref{sec:starteizw}, the variants
V$(\mathbf y_0 \Join \mathbf y_1)|$V$(\mathbf y_2 \Join \mathbf y_3)$,
V$(\mathbf y_0 \Join \mathbf y_2)|$V$(\mathbf y_1 \Join \mathbf y_3)$, and
V$(\mathbf y_1 \Join \mathbf y_2)|$V$(\mathbf y_0 \Join \mathbf y_3)$
are $\mathcal R(r-1,m-1)$
and each of them can be interpretated again as a double Plotkin construction with
$\mathcal C'_0 = \mathcal R(r-1,m-3)$, 
$\mathcal C'_1 = \mathcal C'_2 =  \mathcal R(r-2,m-3)$, 
and $\mathcal C'_3 = \mathcal R(r-3,m-3)$.

\begin{example}[Recursive Decoding]\label{ex:r37}
	The code $\mathcal R(3,7) = (128,64,16)$ is a double Plotkin construction with the four codes   
	$\mathcal C_0 = \mathcal R(3,5) = (32,26,4)$, 
$\mathcal C_1 = \mathcal C_2 =  \mathcal R(2,5) = (32,16,8)$, 
	and $\mathcal C_3 = \mathcal R(1,5) = (32,6,16)$ (see Fig. \ref{fig:rm-tree}).
Then, the $\mathcal R(2,5) = (32,16,8)$ is a double Plotkin construction with the codes
	$\mathcal C'_0 = \mathcal R(2,3)=(8,7,2)$, 
	$\mathcal C'_1 = \mathcal C'_2 =  \mathcal R(1,3)=(8,4,4)$, 
	and $\mathcal C'_3 = \mathcal R(0,3) = (8,1,8)$.
Further, the $\mathcal R(3,5) = (32,26,4)$ is a double Plotkin construction with the codes
	$\mathcal C'_0 = \mathcal R(3,3)=(8,8,1)$, 
	$\mathcal C'_1 = \mathcal C'_2 =  \mathcal R (2,3)=(8,7,2)$, 
	and $\mathcal C'_3 = \mathcal R(1,3) = (8,4,4)$.
The code $\mathcal R(1,5) = (32,6,16)$ is a simplex code and thus a Plotkin construction
with	$\mathcal C_0 = \mathcal R(1,4)=(16,5,8)$ and 
	$\mathcal C_1 = \mathcal R(0,4)=(16,1,16)$. 
	Therefore, for decoding of the $\mathcal R(3,7) = (128,64,16)$
	we need to decode multiple times the codes
	$(8,8,1)$, $(8,7,2)$, $(8,4,4)$, $(8,1,8)$, $(16,5,8)$, and $(16,1,16)$.

Another possibility is to consider a Plotkin construction where one code is the $(64,22,16)$ code which
can be interpreted as a double Plotkin construction with	$\mathcal C_0 = \mathcal R(3,4)=(16,11,4)$, 
	$\mathcal C_1 = \mathcal C_2 =  \mathcal R(2,4)=(16,5,8)$, 
	and $\mathcal C_3 = \mathcal R(1,4) = (16,1,16)$.
\end{example}
For Plotkin constructions with other code classes the recursive decoding is only possible if the 
used codes have the necessary subcode properties.

\section{Performance of Double Plotkin Constructions}\label{sec:decperform}

We will first introduce decoding algorithms which are the realizations of the decoding
functions. Short codes will be decoded by the correlation with all code words.
To decode simplex codes we introduce two ML decoders.
Also, repetition and parity-check codes can be ML decoded.
These algorithms will be used to decode the $\mathcal R(2,5)=(32,16,8)$ RM code.
We study the performance of different variants and their combinations
without and with list decoding of the hidden code words in the first decoding step.
Some examples of the complexity for the decoding of some variants are given.
Finally, we show that also constructions which are not RM codes 
as well as constructions based on BCH codes can be decoded with variants which decode hidden code words.
We introduce two classes of half-rate codes based on RM codes.
\subsection{Decoding Algorithms}
We will describe possible realizations of the decoding functions.
Since we use the recursive decoding described in Sec. \ref{sec:recdec}
we need to decode only short codes which we can do by calculating the correlation
with all code words according to Lemma \ref{corrallone}.
In addition, we have to decode simplex, repetition, and parity-check codes.
We will describe possible algorithms in the following.

Let us recall an old result \cite[Th. 9.27, p. 381]{Boss-eng} for ML decoding of 
first-order RM codes $\mathcal R(1,m)$ (simplex codes).
For these codes the code $\mathcal C_1$ is a repetition code $(2^{m-1},1,2^{m-1})$
and  $\mathcal C_0$ is also a simplex code $(2^{m-1},m,2^{m-2})$.
The following Lemma is a modification of the old result.
\begin{lemma}[ML Decoding of Simplex Codes]\label{mldecr1m}
If the code $\mathcal C_1$ in a Plotkin construction is a repetition code and 
 the code $\mathcal C_0$ contains the all-one code word, then
ML decoding of the code $\mathcal C = | \mathbf u_0 | \mathbf u_0 + \mathbf u_1 | \leftrightarrow 
	| \mathbf x_0 | \mathbf x_0  \mathbf x_1 | $
can be done by  $|\Phi (\mathbf x_0, \mathbf y_0)| + |\Phi (\mathbf x_0, \mathbf y_1)|$
over all code words of $\mathcal C^-_0$ which (according to Lemma \ref{corrallone}) 
contains only either a code word or the inverted code word but not both.
\end{lemma}
\begin{proof}
Let the maximum of the correlations of all code words of $\mathcal C$ be $\phi_m = \Phi (\mathbf x, \mathbf y)$
and the corresponding code word
is $\mathbf x_m | \mathbf x_m \mathbf x_1 \in \mathcal C$
where  $\mathbf x_1$ is either all-one or the all-minus-one vector.
Thus, either $\mathbf x_m | \mathbf x_m$ or  $\mathbf x_m | -\mathbf x_m$.
For the code word  $\mathbf x_m \in \mathcal C_0$ we can write
$\phi_m=\Phi (\mathbf x_m, \mathbf y_0) + \Phi (\mathbf x_m, \mathbf x_1\mathbf y_1)$.
Now, for  $\mathbf x_m \in \mathcal C^-_0$ 
the maximum of the correlation 
$\Phi (\mathbf x_0 , \mathbf y_0)$ for  $\mathbf x_0 \in \mathcal C^-_0$ is $\phi_m$
for $\mathbf x_0 = \mathbf x_m $.
In case of  $\mathbf x_m \not\in \mathcal C^-_0$ according to the definition 
$\mathbf -x_m \in \mathcal C^-_0$
and the correlation is  
$-\phi_m = \Phi (\mathbf -x_m , \mathbf y_0)$. Taking the absolute value  
	will give this maximum.
The same argument holds for $\Phi (\mathbf x_m \mathbf x_1, \mathbf y_1)$
 since $\mathbf x_1$
can only be the all-one or the all-minus-one vector.
If the sum $|\Phi (\mathbf x_0 , \mathbf y_0)| + |\Phi (\mathbf x_0, \mathbf y_1)|$
is maximized at 
 $\mathbf x_d$ then the decision is $\hat{\mathbf x}_0 = \mathbf x_d$ if $\Phi (\mathbf x_d, \mathbf y_0) > 0$
	and $\hat{\mathbf x}_0 = -\mathbf x_d$ otherwise.
The decision is $\hat{\mathbf x}_0 \hat{\mathbf x}_1 = \hat{\mathbf x}_0$ if the sign of $\Phi (\mathbf x_d, \mathbf y_0) $
is the same as the sign of $\Phi (\mathbf x_d, \mathbf y_1)$
and $\hat{\mathbf x}_0 \hat{\mathbf x}_1 = -\hat{\mathbf x}_0$ otherwise. 
\end{proof}

The interpretation of a simplex code of length $4n$ as 
double Plotkin construction yields another ML decoding algorithm.
Only two different codes are used, a simplex code $\mathcal C_0 $ and a repetition
code $\mathcal C_1 = \mathcal C_2$ ($\mathcal C_3 $ is not used or consists of the all-one code word only)
where the length of each code is $n$. 
A code word is then
 $\mathbf x_0 | \mathbf x_0 \mathbf x_1 | \mathbf x_0 \mathbf x_2 | \mathbf x_0 \mathbf x_1 \mathbf x_2 $.
The add-four is 
$\mathbf w = \mathbf y_0 + \mathbf x_1\mathbf y_1+ \mathbf x_2\mathbf y_2+ \mathbf x_1\mathbf x_2\mathbf y_3$.
The two repetition codes $\mathbf x_1$ and $\mathbf x_2$ can be either all-one or all-minus-one
 which gives four possible cases  
$\mathbf w = \mathbf y_0 + \mathbf y_1+ \mathbf y_2+ \mathbf y_3$,
$\mathbf w= \mathbf y_0 - \mathbf y_1+ \mathbf y_2- \mathbf y_3$,
$\mathbf w = \mathbf y_0 + \mathbf y_1 - \mathbf y_2 - \mathbf y_3$, and
$\mathbf w = \mathbf y_0 - \mathbf y_1 - \mathbf y_2 + \mathbf y_3$.
We get $6$ dB gain and  $\mathbf w$
is a noisy code word of a simplex code of length $n$.
We can decode each and get four estimates $\hat{\mathbf x}_{0}$.
In order to compare the four different solutions we have to calculate
 $\hat{\mathbf x} = \hat{\mathbf x}_0 | \hat{\mathbf x}_0 \hat{\mathbf x}_1| \hat{\mathbf x}_0 \hat{\mathbf x}_2 |
  \hat{\mathbf x}_0 \hat{\mathbf x}_1 \hat{\mathbf x}_2 $ and the correlation
$\Phi (\hat{\mathbf x}, \mathbf y)$.
However, according to Lemma \ref{simplexdoplo}, this can be done by
summation of the terms $x_{0,i}w_{i}$.
If the decoding of $\hat{\mathbf x}_{0}$ is ML then the correlation
is maximal in each of the four cases. If we choose the maximum of the four cases
we get the largest possible correlation and have an ML decoder which has proved the following lemma.

\begin{lemma}[ML Decoding of Simplex Codes II]\label{mldoplr1m}
A double Plotkin construction of a simplex code  of length $4 n$
	consists of a simplex code $\mathcal C_0 $ and a repetition
code $\mathcal C_1 = \mathcal C_2$
where the length of each code is $n$. 
Using a list of length four with all possible combinations of the two repetition codes
gives the add-four  
$\mathbf w = \mathbf y_0  \pm \mathbf y_1 \pm \mathbf y_2 \pm \mathbf y_3$.
Decoding by $\mathbf w $ by an ML decoder for the simplex code $\mathcal C_0$
	is ML decoding.
\end{lemma}

\begin{example}[ML Decoding of Simplex Code]\label{ex:mlrm13}
The code  $\mathcal C_0(4,3,2)$ consists of the four code words of $\mathcal C^-_0$, namely 
$0000$, $0011$, $0101$, and $0110$
and the inverted code words  $\mathcal C^+_0$ are
$1111$, $1100$, $1010$, and $1001$.
We use the repetition code $\mathcal C_1(4,1,4)$  
	and with the Plotkin construction we get the code $\mathcal R(1,3)=\mathcal C(8,4,4)$.  
According to Lemma \ref{mldecr1m} we can ML decode this code by the four correlations
\[
|y_0+y_1+y_2+y_3|+|y_4+y_5+y_6+y_7|= |\phi_{0,0}| +|\phi_{0,1}| 
\]
\[
|y_0+y_1-y_2-y_3|+|y_4+y_5-y_6-y_7|= |\phi_{1,0}| +|\phi_{1,1}| 
\]
\[
|y_0-y_1+y_2-y_3|+|y_4-y_5+y_6-y_7|= |\phi_{2,0}| +|\phi_{2,1}| 
\]
\[
|y_0-y_1-y_2+y_3|+|y_4-y_5-y_6+y_7|= |\phi_{3,0}| +|\phi_{3,1}| 
\]
We choose the maximal value of these four, say it has index $i$.  
If the sign of $\phi_{i,0}>0$ then it is the code word and otherwise the inverted code word.
	If the sign of  $\phi_{i,0}$ is the
same as for  $\phi_{i,1}$ then it is the all-zero codeword of the repetition code otherwise the all-one code word.
Note that one correlation consists of $7$ additions, thus, all together $28$ additions;  finding the maximum 
needs $4$ comparisons.

For the ML decoding according to Lemma \ref{mldoplr1m}
we use the codes $\mathcal C_0(2,2,1)$ and  $\mathcal C_1(2,1,2) = \mathcal C_2(2,1,2)$
for a double Plotkin construction. 
For all four cases of possible repetition codes we calculate the add-four
\[
w_0 = y_0+y_2+y_4+y_6,\quad w_1 = y_1+y_3+y_5+y_7  
\]
\[
w_0 = y_0-y_2+y_4-y_6, \quad w_1 = y_1-y_3+y_5-y_7  
\]
\[
w_0 = y_0+y_2-y_4-y_6, \quad w_1 = y_1+y_3-y_5-y_7  
\]
\[
w_0 = y_0-y_2-y_4+y_6, \quad w_1 = y_1-y_3-y_5+y_7.
\]
Now we decide the code word $\mathbf x_0$
	as $(1,1)$ if $w_0 > 0$ and $w_1 > 0$,
	as $(1,-1)$ if $w_0 > 0$ and $w_1 < 0$,
	as $(-1,1)$ if $w_0 < 0$ and $w_1 > 0$, and
	as $(-1,-1)$ if $w_0 < 0$ and $w_1 < 0$.
	We choose the decision with maximum  $|w_0| + |w_1| $
	which has the largest correlation of all possible code words.
	We need $3$ additions per $w_i$ which gives the total mumber $24$ for all eight coefficients.
	Then,  $4$ additions and $4$ comparisons. 
\end{example}
Finally, we recall the well known ML decoding of repetition and parity-check codes.
\begin{lemma}[ML Decoding of Repetition and Parity-Check Codes]\label{mlpcrep}
	A repetition code $\mathcal R(0,m) = (2^m,1,2^m)$ can be ML decoded given the received vector $\mathbf y$ by
calculating $\Phi  = \sum\limits_{i=0}^{n-1}  y_{i}$. For $\Phi > 0$ the decision is the all-zero code word
and otherwise the all-one code word.
A parity-check code $\mathcal R(m-1,m) = (2^m,2^m-1,2)$ can be ML decoded by flipping the position with the smallest reliability if the weight is odd.
\end{lemma}
\subsection{Variants and Lists}
\begin{table}
	\caption{Decoding Variants for Simulation}\label{tab:decvar}
\begin{center}
	\begin{tabular}{c|c|c}
		Starting with & Variant & Number of Variants \\
		\hline
		$\mathcal C_3$  & V$(4, \mathbf y_i \Join \mathbf y_j)$ & 3 \\ 
		\hline
		$\mathcal C_1, \mathcal C_2$  & V$(\mathbf y_i \Join \mathbf y_j)$ & 6 \\ 
		\hline
		$\mathcal C_0$  & V$(\mathbf y_i, \mathbf y_j)$ & 6 \\ 
	\end{tabular}
\end{center}
\end{table}
In Sec. \ref{decvariants} we have introduced decoding variants
which are repeated in Table \ref{tab:decvar}.
We will simulate the performance  of different variants and their combinations
without and with list decoding in the first decoding step.
For this we use 
 the $\mathcal R(2,5)=(32,16,8)$ RM code.
This code is a double Plotkin construction of the codes
$\mathcal C_0=(8,7,2)$, $\mathcal C_1=\mathcal C_2=(8,4,4)$, and
$\mathcal C_3=(8,1,8)$ which is a repetition code.
A code word consists of four blocks
$ \mathbf x_0 | \mathbf x_0 \mathbf x_1 | \mathbf x_0 \mathbf x_2 | \mathbf x_0 \mathbf x_1 \mathbf x_2\mathbf x_3$.
Transmitting a code word over an AWGN channel we receive   
$$	\mathbf y = |\mathbf y_0  | \mathbf y_1  | \mathbf y_2 |\mathbf y_3 | =
	|\mathbf x_0 +\mathbf z_0 |\mathbf x_0 \mathbf x_1 + \mathbf z_1|\mathbf x_0  \mathbf x_2 +\mathbf z_2  |\mathbf x_0 \mathbf x_1 \mathbf x_2\mathbf x_3 + \mathbf z_3|
$$ 
where $\mathbf z_i$ is the noise vector.
As decoder $\Delta_1$ and $\Delta_2$ we use the one from Ex. \ref{ex:mlrm13} which is ML and provides also a list
of the best decisions if necessary. 
For the other two codes we use decoders according to Lemma \ref{mlpcrep}.

\subsubsection{Starting Decoding with $\mathcal C_3$}
\begin{figure}[htb]
	\begin{center}
	\includegraphics[width=0.65\textwidth]{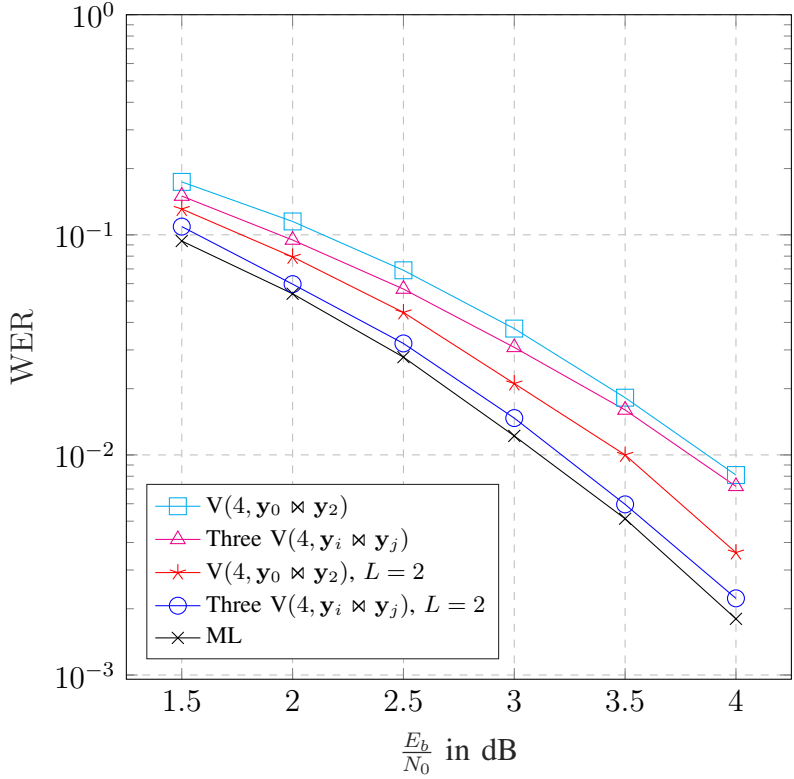}
	\end{center}
	\caption{WER $\mathcal R(2,5)$ starting with $\mathcal C_3$}\label{fig:r25-cdrei}
\end{figure}
Interpreting the double Plotkin construction as GCC, 
traditionally  $\mathcal C_3$ is decoded first, then $\mathcal C_2$, $\mathcal C_1$, and 
	finally, $\mathcal C_0$ which corresponds to
	variant V$(4, \mathbf y_0 \Join \mathbf y_2)$ in Sec. \ref{sec:startc3}. 
	The performance of V$(4, \mathbf y_0 \Join \mathbf y_2)$ is shown in Fig. \ref{fig:r25-cdrei}
	where the word error rate (WER) is plotted versus the signal-to-noise ratio. 
	The other two variants  V$(4, \mathbf y_0 \Join \mathbf y_1)$ and  V$(4, \mathbf y_1 \Join \mathbf y_2)$ give the same performance as V$(4, \mathbf y_0 \Join \mathbf y_2)$.
The first decoding step of all three variants 
	is the same and done with join-four 
	$\mathbf y_0 \Join \mathbf y_1 \Join \mathbf y_2 \Join \mathbf y_3$ which is a disturbed
	hidden code word of $\mathcal C_3$ and decoded by $\Delta_3$.
Using all three variants
together and deciding for the one with the 
	largest correlation gives hardly any improvement as shown in Fig. \ref{fig:r25-cdrei}. 
	The  improvement even gets smaller for larger signal-to-noise ratios.
	The explanation is that the first decoding step is the same in all three variants and 
dominates the performance.
Using only one variant but a list decoder for $\mathcal C_3$ with $L=2$, the improvement is larger than with the three 
variants without list decoding.
Note that $L=2$ corresponds to a complete list of all possible code words since $\mathcal C_3$ is a repetition code.
We also see in Fig. \ref{fig:r25-cdrei} that using the three variants 
in case of a list decoder $\Lambda_3$ with $L=2$
gives an improvement and in fact it reaches almost ML performance. 
The ML performance is known for this code and so we do not need to use the ML-bound.

\begin{figure}[htb]
	\begin{center}
	\includegraphics[width=0.65\textwidth]{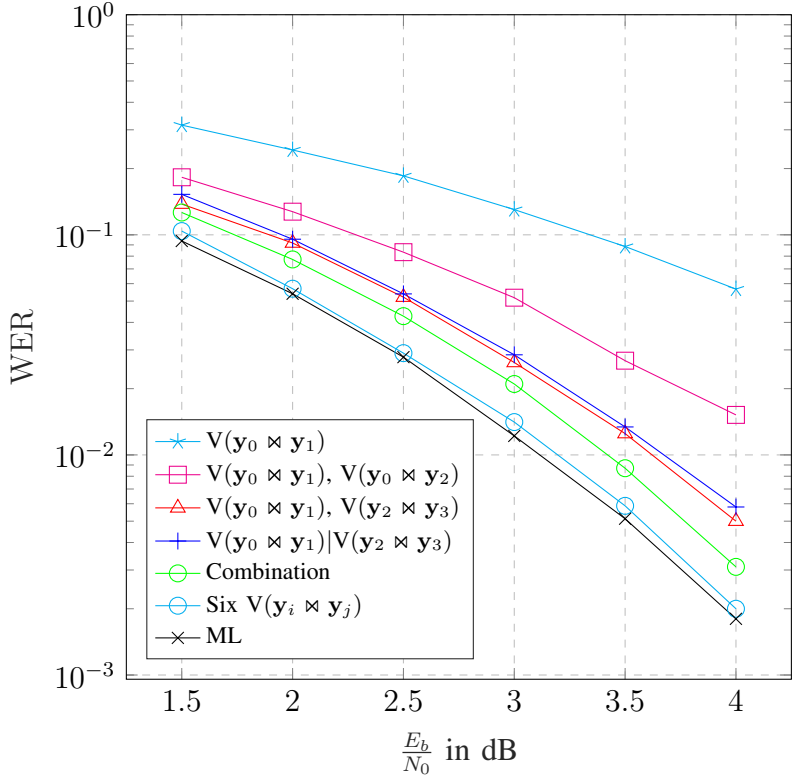}
	\end{center}
	\caption{WER $\mathcal R(2,5)$ starting with  $\mathcal C_1$, $\mathcal C_2$}\label{fig:r25-ceizw}
\end{figure}
\subsubsection{Starting Decoding with $\mathcal C_1$,$\mathcal C_2$}\label{sec:c12r25}
We consider first the variants of Sec. \ref{sec:starteizw} without list decoding in the first step.
The simulation results are shown in Fig. \ref{fig:r25-ceizw}.
Using only  variant V$(\mathbf y_0 \Join \mathbf y_1)$ the performance is more than 1 dB worse at a WER of $ 10^{-1}$ than using the
two variants  V$(\mathbf y_0 \Join \mathbf y_1)$ and V$(\mathbf y_0 \Join \mathbf y_2)$ and deciding for the codeword with larger correlation. However, using 
the two variants  V$(\mathbf y_0 \Join \mathbf y_1)$ and V$(\mathbf y_2 \Join \mathbf y_3)$  the performance 
is about $1.5$ dB better at a WER of $ 10^{-1}$ than using only  variant V$(\mathbf y_0 \Join \mathbf y_1)$.
The explanation is that in the case  V$(\mathbf y_0 \Join \mathbf y_1)$ and V$(\mathbf y_0 \Join \mathbf y_2)$ 
both use the block $\mathbf y_0$ for each first decoding step which means they are dependent. 
In the second case V$(\mathbf y_0 \Join \mathbf y_1)$ and V$(\mathbf y_2 \Join \mathbf y_3)$ 
there is no common block for each first decoding step.
Using all six variants  V$(\mathbf y_0 \Join \mathbf y_1)$, V$(\mathbf y_2 \Join \mathbf y_3)$,
V$(\mathbf y_0 \Join \mathbf y_2)$, V$(\mathbf y_1 \Join \mathbf y_3)$, V$(\mathbf y_1 \Join \mathbf y_2)$,
and V$(\mathbf y_0 \Join \mathbf y_3)$, 
for decoding gives allmost
	ML performance and is about $0.5$ dB better at a WER of $ 10^{-1}$ than variants  V$(\mathbf y_0 \Join \mathbf y_1)$ and 
	V$(\mathbf y_2 \Join \mathbf y_3)$.
We notice that the performance improvement when adding additional variants
is much larger here than in Fig. \ref{fig:r25-cdrei} when we start decoding $\mathcal C_3$.
	The explanation is that according to Eq. \eqref{birthapprox} 
	the probability 
that errors are cancelled is very large for the join-four operation
and according to Lemma \ref{genialparadox}
this error reduction occurs at least in one of the join-two 
of the six variants we start the decoding with.
The second decoding in each variant is $\mathcal C_3$, however, with join-two  and not join-four as in V$(4, \mathbf y_0 \Join \mathbf y_2)$.
Therefore, the decoding is correct with a much larger probability than for join-four.
For the third decoder we have the add-join operation which according to Ex. \ref{ex:gains} 
reduces the number of errors considerably.
The last decoder operates in a $6$ dB better channel.
So, if the first decoder finds the correct solution the three following are also correct with large probability.
Further, it can be observed that the slope of the performance curves gets larger with the number of used variants.
This effect is known exploiting diversity in communications. 
But using one variant only
the performance of variant V$(\mathbf y_0 \Join \mathbf y_1)$ is more than 1 dB worse at a WER of $ 10^{-1}$ than  V$(4, \mathbf y_0 \Join \mathbf y_2)$ from Fig.~\ref{fig:r25-cdrei} which follows from the fact that $\mathcal C_1$ has minimum distance $4$, while $\mathcal C_3$ has $8$. 

In Fig. \ref{fig:r25-ceizw} we also show simulation results of a further aspect of the decoding variants.
In Sec. \ref{sec:starteizw} an alternative interpretation was described
using $\mathbf y_0 \Join \mathbf y_1 | \mathbf y_2 \Join \mathbf y_3$
	which is  a Plotkin construction $ \mathbf x_1 | \mathbf x_1 \mathbf x_3$
	and in our example a $\mathcal R(1,4)=(16,5,8)$ code which can be ML decoded according to Lemma \ref{mldecr1m}.
	After this, with a add-join operation $\mathcal C_2$ is decoded and then after a add-four 
	operation $\mathcal C_0$.  
	Using this decoder  V$(\mathbf y_0 \Join \mathbf y_1)|$V$(\mathbf y_2 \Join \mathbf y_3)$  in the first step the performance is suprisingly slightly worse than
	the case V$(\mathbf y_0 \Join \mathbf y_1)$ and V$(\mathbf y_2 \Join \mathbf y_3)$ as shown in Fig. \ref{fig:r25-ceizw}. 
	The simulation using  a combination of the three variants V$(\mathbf y_0 \Join \mathbf y_1)$, V$(\mathbf y_2 \Join \mathbf y_3)$, and V$(\mathbf y_0 \Join \mathbf y_1)|$V$(\mathbf y_2 \Join \mathbf y_3)$ gives a visible improvement
	(named Combination in Fig. \ref{fig:r25-ceizw}).
	This means that there are cases when V$(\mathbf y_0 \Join \mathbf y_1)$ and V$(\mathbf y_2 \Join \mathbf y_3)$ 
	cannot decode but 
	V$(\mathbf y_0 \Join \mathbf y_1)|$V$(\mathbf y_2 \Join \mathbf y_3)$ decodes correctly and
	vice versa.
	The explanation for this effect is as follows.
	The ML decoding calculates the sum of the absolute values of the correlation of the code words of the code $\mathcal C_1$
	with the two vectors $\mathbf y_0 \Join \mathbf y_1$ and $\mathbf y_2 \Join \mathbf y_3$
	and chooses the maximum. 
	The result is $ \mathbf x_1 | \mathbf x_1 \mathbf x_3$ and with  
	$ \mathbf x_1$ and $ \mathbf x_3$ the remaining two decoding steps are done. 
	The other two variants decide the two correlations independently.
	The first variant calculates the correlation of the code words of the code $\mathcal C_1$
	with the vector $\mathbf y_0 \Join \mathbf y_1$ and chooses the maximum as decison for $ \mathbf x_1$. 
	Then, in the second decoding step, it calculates a join-two 
	$\mathbf y_2 \Join \mathbf y_3 \mathbf x_1$ in order to decode $ \mathbf x_3$.
	The remaining two decoding steps are identical to the above case.
	The second variant calculates the correlation of the code words of the code $\mathcal C_1$
	with the vector $\mathbf y_2 \Join \mathbf y_3$ and chooses the maximum as decison for 
	$ \mathbf x_1 \mathbf x_3$ 
	and then calculates a join-two $\mathbf y_0 \Join \mathbf y_1 \mathbf x_1 \mathbf x_3$ in order to decode
	$ \mathbf x_3$.
	Therefore, it can happen that the correlation with $\mathbf y_0 \Join \mathbf y_1$
	gives the correct result but the one with $\mathbf y_2 \Join \mathbf y_3$ the wrong one
	and vice versa. Then, the sum of both absolute values of this correlation 
	might give the wrong decision in the ML decoding.
	The opposite case is that both separate correlations give the wrong result
	but the sum the correct one.

\begin{figure}[htb]
	\begin{center}
	\includegraphics[width=0.65\textwidth]{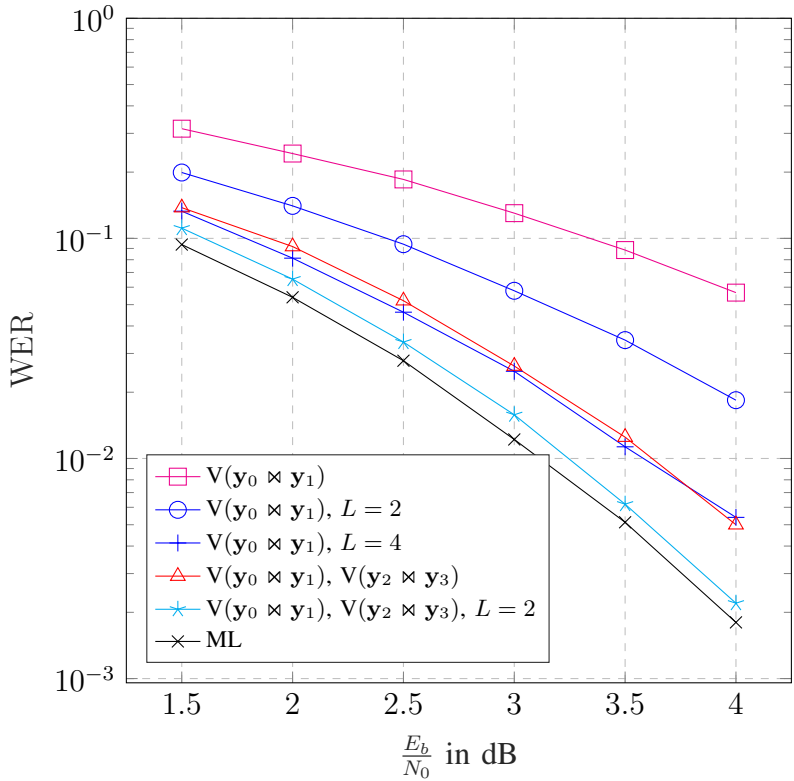}
	\end{center}
	\caption{WER $\mathcal R(2,5)$ starting with  $\mathcal C_1$, $\mathcal C_2$ with list decoding}\label{fig:r25-list}
\end{figure}

Now we consider the variants from Sec. \ref{sec:starteizw} with list decoding in the first step.
The simulation results are shown in Fig. \ref{fig:r25-list}. Variant V$(\mathbf y_0 \Join \mathbf y_1)$ 
with list $ L = 2$ is about $1$ dB better at a WER of $ 10^{-1}$ than without list. 
However, variants V$(\mathbf y_0 \Join \mathbf y_1)$  and V$(\mathbf y_2 \Join \mathbf y_3)$  
are more than $1.5$ dB better at a WER of $ 10^{-1}$  compared to variant V$(\mathbf y_0 \Join \mathbf y_1)$. 
This shows that in case of two decoders two variants give a better performance than one variant with
 list size $ L = 2$.  The explanation is that one variant does not use the variance within the received blocks
 but rather starts with a list decoder for a single join-two $\mathbf y_0 \Join \mathbf y_1$. 
 The probability that this list does not contain the correct code word must therefore be larger
 than the probability that both join-two, $\mathbf y_0 \Join \mathbf y_1$ and $\mathbf y_2 \Join \mathbf y_3$,
 are not decoded correctly. 
  Using two variants V$(\mathbf y_0 \Join \mathbf y_1)$ and V$(\mathbf y_2 \Join \mathbf y_3)$ with a list $ L = 2$ gives another significant improvement.
This should be compared to  V$(\mathbf y_0 \Join \mathbf y_1)$ with list $ L = 4$ since both use $4$ decoders.
However, as seen in Fig. \ref{fig:r25-list} the variant V$(\mathbf y_0 \Join \mathbf y_1)$ with list $ L = 4$ 
has very similar performance to variants V$(\mathbf y_0 \Join \mathbf y_1)$
and V$(\mathbf y_2 \Join \mathbf y_3)$ which are only two decoders.
Even more, this two curves are crossing each other at $3.75$ dB which can be explained 
by the fact that the two variants use the variability of the channel since 
both use different blocks and therefore, the performance curve has a steeper slope. 
Adding more variants with different list sizes brings the performance closer to the ML curve but since the use of
two variants with list size $ L = 2$ is already close to ML we omit these curves. 

\begin{figure}[htb]
	\begin{center}
	\includegraphics[width=0.65\textwidth]{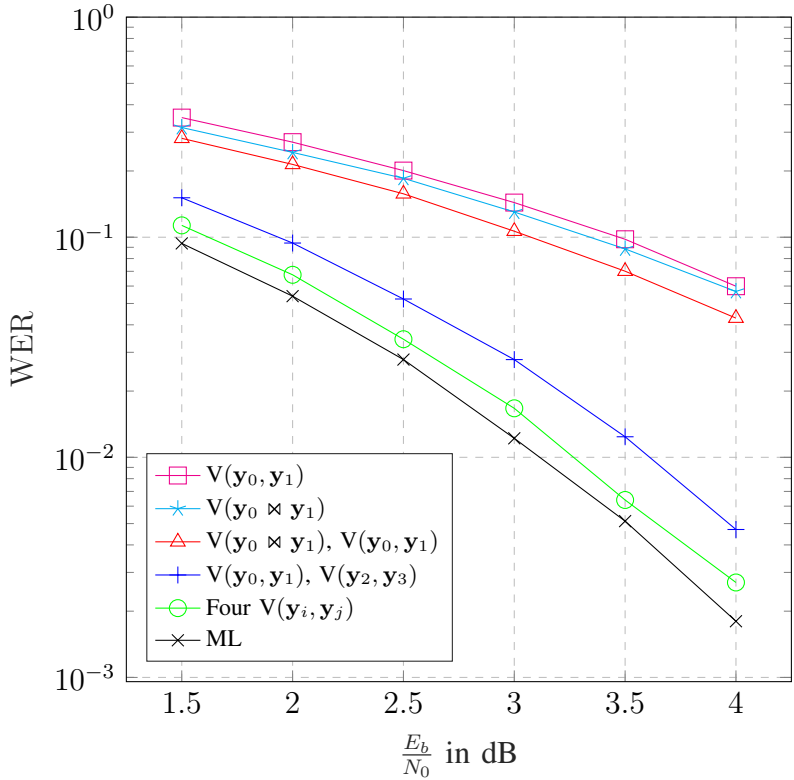}
	\end{center}
	\caption{WER $\mathcal R(2,5)$ starting with  $\mathcal C_0$}\label{fig:r25-cnull}
\end{figure}
\subsubsection{Starting Decoding with $\mathcal C_0$}
For this variants it is necessary that $\mathcal C_1 \subset \mathcal C_0$,
$\mathcal C_2 \subset \mathcal C_0$, and
$\mathcal C_3 \subset \mathcal C_0$.
 The variants of Sec. \ref{sec:startnull} are  related to the ones of Sec. \ref{sec:starteizw} as follows.
In V$(\mathbf y_0 \Join \mathbf y_1)$ we use $\mathbf y_0 \Join \mathbf y_1$ and
then a decoder $\Delta_1$ for $\mathcal C_1$.
In V$(\mathbf y_0, \mathbf y_1)$  we use list decoders $\Lambda_0$ for both, 
$\mathbf y_0$ and $\mathbf y_1$ and then an indicator $\Gamma_1$ for $\mathcal C_1$
in order to select proper code words from the two lists.  
While the join-two operation  in V$(\mathbf y_0 \Join \mathbf y_1)$ 
might cancel errors the errors are unchanged in V$(\mathbf y_0, \mathbf y_1)$.
On the other hand, the reliability values of the channel are changed in V$(\mathbf y_0 \Join \mathbf y_1)$ but not in 
V$(\mathbf y_0, \mathbf y_1)$.
The simulation results are shown in Fig. \ref{fig:r25-cnull}.
It can be observed that V$(\mathbf y_0, \mathbf y_1)$ has a slightly worse performance than 
V$(\mathbf y_0 \Join \mathbf y_1)$ which is due to the cancellation of errors by the join-two. 
However, both variants correct different errors which can be seen by the small improvement when using both variants
and choosing the decision with larger correlation. 
This means that some errors can be corrected by V$(\mathbf y_0, \mathbf y_1)$ but not by 
V$(\mathbf y_0 \Join \mathbf y_1)$  and vice versa. 
Using two variants V$(\mathbf y_0, \mathbf y_1)$  and 
V$(\mathbf y_2, \mathbf y_3)$  gives also here an improvement of more than $1.5$dB
at a WER of $10^{-1}$
as for variants V$(\mathbf y_0 \Join \mathbf y_1)$  and V$(\mathbf y_2 \Join \mathbf y_3)$. 
The results from  Fig. \ref{fig:r25-cnull} and Fig. \ref{fig:r25-ceizw}
are very similar. 
The performance of four variants  
V$(\mathbf y_0, \mathbf y_1)$, V$(\mathbf y_2, \mathbf y_3)$, V$(\mathbf y_0, \mathbf y_2)$,
and V$(\mathbf y_1, \mathbf y_3)$
is already close to ML decoding. 
The performance of all six variants 
is almost identical to
the six variants  V$(\mathbf y_i \Join \mathbf y_j)$ and therefore omitted. 
Since the results are very similar to those starting with $\mathcal C_1$, $\mathcal C_2$we 
will not consider the variants starting decoding with 
$\mathcal C_0$ in the following simulation results.

\begin{figure}[htb]
	\begin{center}
	\includegraphics[width=0.65\textwidth]{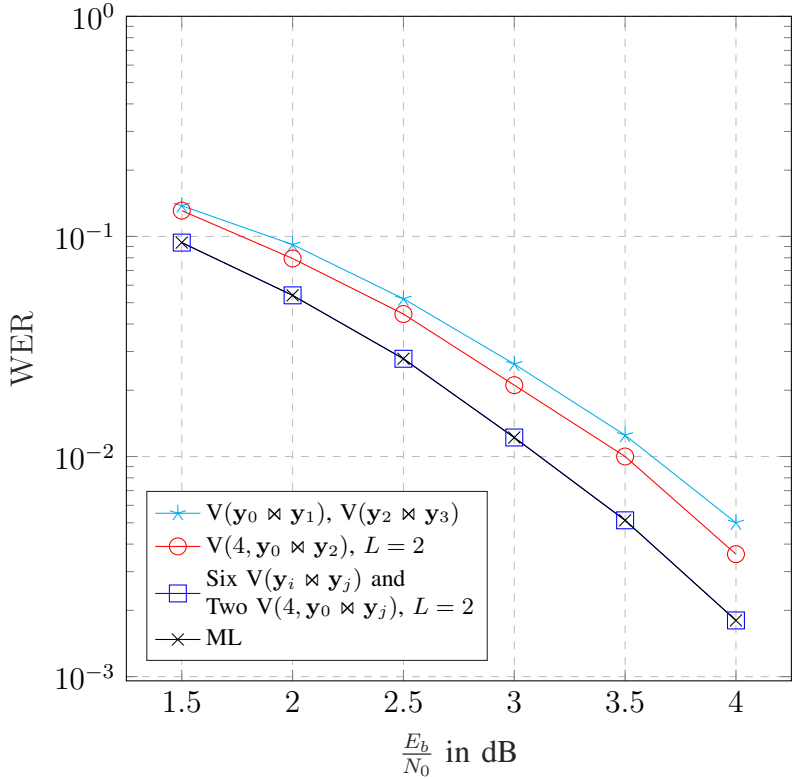}
	\end{center}
	\caption{WER $\mathcal R(2,5)$ Mixed Decoding Variants}\label{fig:r25-mixed}
\end{figure}

\subsubsection{Starting Decoding with $\mathcal C_1$, $\mathcal C_2$, $\mathcal C_3$}\label{sec:r25c123}
Here, we will consider mixed variants which start decoding 
with $\mathcal C_1$, $\mathcal C_2$ and $\mathcal C_3$.
In Fig. \ref{fig:r25-mixed} the performance curves of  
V$(\mathbf y_0 \Join \mathbf y_1)$ and V$(\mathbf y_2 \Join \mathbf y_3)$ 
is repeated for comparison.
It can be seen that the performance V$(4, \mathbf y_0 \Join \mathbf y_2), L=2$
is slightly better than the repeated curve.
Using all six variants V$(\mathbf y_i \Join \mathbf y_j)$
and two variants $\mathrm V (4, \mathbf y_0 \Join \mathbf y_1)$ and 
$\mathrm V (4, \mathbf y_0 \Join \mathbf y_2)$
with list size $L=2$ each of the performance curves in
 Fig.~\ref{fig:r25-mixed} is identical to ML performance.
 We will use this as the decoder 
in Sec.~\ref{sec:recurplot} for recursive decoding.

\subsubsection{Decoding Complexity}\label{sec:complexity}
In the decoding we use three operations, namely, correlation, 
join and add, and sorting or finding the maximum or minimum.
The correlation has as basic operation the addition
of real values and a sign operation. The join and add operations
consist of sign operations, comparisons, and additions.
The sorting or the search of the maximum or minimum consists of comparisons.
The complexity of all used join and add operations 
can be given in general when the four blocks have length $n$ each which 
means the constructed code has length $4 n$.
The values are given in Table  \ref{table:compjoinadd}.
We count the sign operations, the comparisons, and the additions separately.
If we ignore the sign operations and assume that a comparison has the same complexity as an addition
we denote this number as \^a\^c operations in the table.
\begin{table}[ht]
	\caption{Complexity of Join and Add Operations\label{table:compjoinadd}}
\begin{center}
\begin{tabular}{r|c|c|c|c}
	operation & sign & comp & add & \^a\^c \\ 
\hline
\hline
	join-two & $n$ & $n$ & & $n$\\[.5ex]
\hline
	join-four & $3n$ & $3n$ & & $3n$\\[.5ex]
\hline
	add-two & $n$ &  & $n$ & $n$\\[.5ex]
\hline
	add-four & $3n$ &  & $3n$ & $3n$\\[.5ex]
\hline
	join-add & $2n$ & $2n$  & $n$ & $3n$\\[.5ex]
\hline
	add-join & $n$ & $n$  & $2n$ & $3n$\\
\end{tabular}
\end{center}
\end{table}

The complexity of the decoders can not be given in general since 
it depends on the specific decoding algorithm used.
For the decoders used in this section we can count the operations as follows.  
The decoder of $\mathcal C_3$ needs $7$ additions and $1$ sign operation which means $7$ \^a\^c operations.
The decoding of $\mathcal C_1$ and $\mathcal C_2$ is done according to Ex. \ref{ex:mlrm13} with $28$ additions,
$32$ sign operations, and  $3$ comparisons to find the maximum of the 
correlations which corresponds to $31$ \^a\^c operations.
Finally, the decoding of $\mathcal C_0$ needs $8$ sign operations 
to check the weight and $7$ comparisons
to find the most unreliable position in case the weight was odd which is $7$ \^a\^c operations. 
With these complexities we can now derive the ones of different variants.

The decoding starting with  $\mathcal C_3$ consists
of the following join and add operations. One join-four,
one join-add, one add-join, and one add-four, thus,
according to Table \ref{table:compjoinadd},
$9n$ sign operations, $6n$ comparisons, and $6n$ additions
or $12n$ \^a\^c operations. Since here $n=8$ the total complexity is $96$ \^a\^c operations.
In addition, we need to decode one time $\mathcal C_3$, $\mathcal C_2$ and $\mathcal C_1$, and $\mathcal C_0$
which has the  complexity $7+31+31+7 = 76$ \^a\^c operations. 
For the measure to compare different variants  the correlation $\Phi(\hat{\mathbf x}, \mathbf y)$
is calculated using Lemma \ref{simplexdoplo} with $7$ additions and $8$ sign operations
since we use add-four which means $7$ \^a\^c operations.
Therefore, this variant has an overall complexity of $96+ 76+ 7 = 180$ \^a\^c operations.

Any variant from Sec. \ref{sec:starteizw} without list
needs two join-two, one add-join, and one add-four, thus,
according to Table \ref{table:compjoinadd},
$2n$, $3n$, and $3n$, alltogether $8n$ \^a\^c operations. Since here $n=8$ the complexity is $64$ \^a\^c operations.
The complexity for the decoding of the four codes and the final correlation
is the same as above  with $83$ \^a\^c operations. 
Therefore, the overall complexity for this variant is $147$ \^a\^c operations
which is $33$ \^a\^c operations less than for the variant above.
Note that the correlation of the final decision is already included.
\begin{table}[ht]
	\caption{Complexity of Decoding Variants for $\mathcal R (2,5)$\label{table:compdecvar}}
\begin{center}
\begin{tabular}{c|c}
	Variant & Number of \^a\^c Operations  \\ 
\hline
\hline
V$(4, \mathbf y_i \Join  \mathbf y_j)$	 & $180$ \\[.5ex]
\hline
V$( \mathbf y_i \Join  \mathbf y_j)$	 & $147$ \\[.5ex]
\hline
Six V$( \mathbf y_i \Join \mathbf y_j)$	 & $887$ 
\end{tabular}
\end{center}
\end{table}
If we use a decoder which uses all six variants from Sec. \ref{sec:starteizw} without list
we need $6 \cdot 147 + 5 = 887$ \^a\^c operations, 
assuming that the decoders are independent parallel implementations.
The examples for the complexity are summarized in Table \ref{table:compdecvar}.
Since the complexity is that small we do not analyze a possible
reduction by calculating some join and add operations 
only once and subsequently using it for all six variants.

Note that the decoding complexity for a list decoder depends also on the specific algorithm 
and so it cannot be given in general. 
However, when we decode a code by correlation 
with all code words we have a set of correlation values.
For a decoder we have to find the maximum correlation and for a list decoder we have to find the 
$L$ largest correlation values by sorting.
Here, the list decoding of $\mathcal C_1$ needs  $28$ \^a\^c operations to calculate the  $4$ 
correlation values.
Instead of the $3$ comparisions to find the maximum we need  $3 + 2 + 1 = 6$ comparisons  
in order to sort the correlation values.
For each code word of the list only three remaining codes have to be decoded.
Therefore, the complexity of a variant using listsize $L$  
in the first decoding step is less than $L$ times the complexity of this variant without list.
\begin{figure}[htb]
	\begin{center}
	\includegraphics[width=0.65\textwidth]{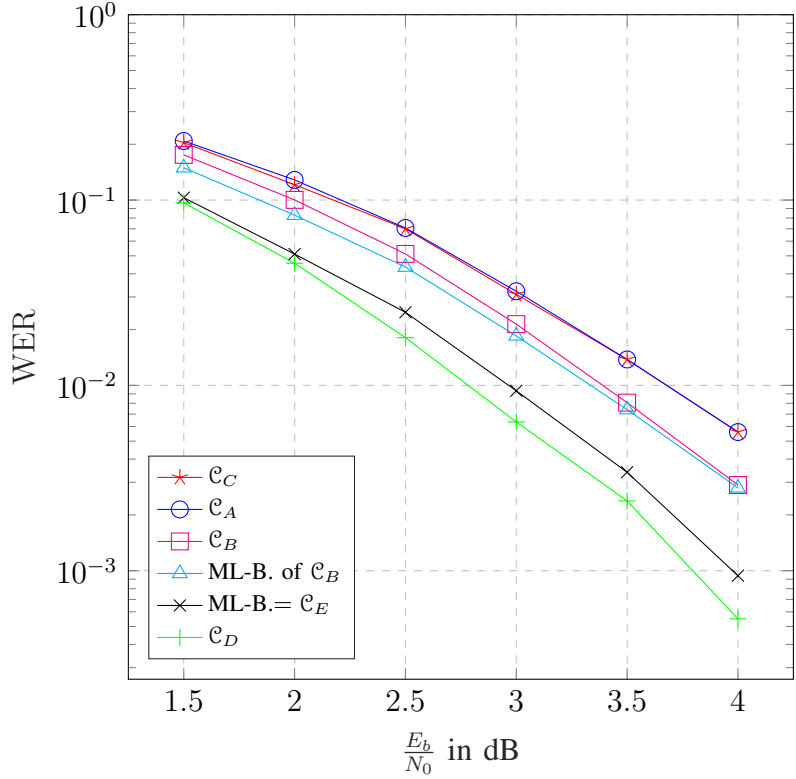}
	\end{center}
	\caption{$(64,32,\geq 8)$ WER for  different Codes}\label{fig:64}
\end{figure}
\subsection{Modified Code Constructions}
The main scope here is to show that the recursive Plotkin
construction is not limited to RM codes but can create a huge variety of codes.
In order to show some of these possibilities 
we consider five codes $\mathcal C(64,32,d)$ where the first three are based on RM codes 
but are not RM codes themselves and the other two use BCH codes as component codes.
However, all five codes have a subcode property such that hidden code words can be used for decoding them.
The first code $\mathcal C_A$ is 
the double Plotkin construction with the RM codes $\mathcal C_0=(16,15,2)$,  $\mathcal C_1=(16,11,4)$, $\mathcal C_2=(16,5,8)$, and
$\mathcal C_3=(16,1,16)$. Note that this is not a RM code. 
The second construction $\mathcal C_B$ is $\mathcal C_0=(16,11,4)$,  $\mathcal C_1 = \mathcal C_2=(16,8,4)$, and
$\mathcal C_3=(16,5,8)$.
For the $(16,8,4)$ code we use two possibilities.
One possibility concatenates two code words of the simplex code of length $8$, namely $\mathcal C_1 = \mathcal C_2 = |(8,4,4)|(8,4,4)| $
and the second possibility uses a double Plotkin construction with
$\mathcal C_0'=(4,3,2)$,  $\mathcal C'_1=(4,2,2)$, $\mathcal C_2'=(4,2,2)$, and
$\mathcal C'_3=(4,1,4)$. For the $(4,2,2)$ code we use the linear subcode of $(4,3,2)$ with the code words
$0000, 1100, 0011, 1111$.
Both possibilities preserve the subcode structure for the decoding with hidden code words.
The third construction $\mathcal C_C$
is  $\mathcal C_0=(16,11,4)$,  $\mathcal C_1=(16,11,4)$, $\mathcal C_2=(16,5,8)$, and
$\mathcal C_3=(16,5,8)$.
The  construction $\mathcal C_D$
uses the code $\mathcal C_0=(16,13,2)$ and the extended BCH codes  $\mathcal C_1=\mathcal C_2=(16,7,6)$, and
$\mathcal C_3=(16,5,8)$.
For the definition of BCH codes we refer to \cite[Ch. 4, pp. 81--94]{Boss-eng} or \cite{sloane}.
For the code $\mathcal C_0=(16,13,2)= |(6,5,2)|(5,4,2)|(5,4,2)|$ we use 
the concatenation of three parity-check codes.
Note that using only decoding variants from Sec. \ref{sec:starteizw} 
and from Sec. \ref{sec:startc3} no subcode property for the code $\mathcal C_0$ is necessary.
Finally, the construction $\mathcal C_E$
 uses the codes $\mathcal C_0=(16,15,2)$,  $\mathcal C_1=(16,7,6)$, and
$\mathcal C_2=\mathcal C_3=(16,5,8)$. Note that here we have 
$\mathcal C_3 =  \mathcal C_2 \subset \mathcal C_1 \subset \mathcal C_0$.
For convenience, the codes are listed in Table \ref{tab:codes64}.
\begin{table}
	\caption{Half-Rate Codes of Length $64$}\label{tab:codes64}
	\begin{center}
\begin{tabular}{r|cccc}
	Code &$\mathcal C_0$ &$\mathcal C_1$ & $\mathcal C_2$ & $\mathcal C_3$ \\
	\hline
	\hline
	$\mathcal C_A$ & $(16,15,2)$ &  $(16,11,4)$ &  $(16,5,8)$ &  $(16,1,16)$ \\ 
	\hline
	$\mathcal C_B$ & $(16,11,4)$ &  $(16,8,4)$ &  $(16,8,4)$ &  $(16,5,8)$ \\ 
	\hline
	$\mathcal C_C$ & $(16,11,4)$ &  $(16,11,4)$ &  $(16,5,8)$ &  $(16,5,8)$ \\ 
	\hline
	$\mathcal C_D$ & $(16,13,2)$ &  $(16,7,6)$ &  $(16,7,6)$ &  $(16,5,8)$ \\ 
	\hline
	$\mathcal C_E$ & $(16,15,2)$ &  $(16,7,6)$ &  $(16,5,8)$ &  $(16,5,8)$ \\ 
\end{tabular}
\end{center}
\end{table}

In Fig. \ref{fig:64} the related simulation results are shown.
For the decoding of all codes we use all six variants from Sec. \ref{sec:starteizw} with list size $L=8$
and two variants from Sec. \ref{sec:startc3} also with list size $L=8$.
The code $\mathcal C_A$ has a nearly identical performance as 
$\mathcal C_C$ and 
both show a performance very close to their ML bounds (the ML performance is not known), respectivly, which means that the
performance is very close to ML performance (we omit the curves of these bounds).
The difference of the codes $\mathcal C_A$ and $\mathcal C_C$ is only 
in the codes $\mathcal C_0$ and $\mathcal C_3$.  
The explanation that both have the same performance can only be that
neither a weaker  $\mathcal C_0$ nor a weaker $\mathcal C_3$
influences the performance. Recall that $\mathcal C_0$ is decoded in a six dB better channel and 
$\mathcal C_3$ is decoded after a join-two and not after a join-four
in six of the used variants.
Since both codes have according to Table \ref{tab:codes64} the same code $\mathcal C_1$ 
and the same $\mathcal C_2$ this fact must dominate the
performance.
The code $\mathcal C_B$ has the best performance of the three RM based codes.
It is the version with the two simplex codes of length $8$.
The version with the double Plotkin construction with the linear subcodes $(4,2,2)$
has a slightly worse performance.
As can be seen in  Fig. \ref{fig:64},  the decoding performance 
for the code $\mathcal C_B$ is getting closer to the ML-bound
for better channels and is already identical at $4$ dB.
The code $\mathcal C_B$  
has the worse $\mathcal C_3$ from $\mathcal C_C$
and the better $\mathcal C_0$  from $\mathcal C_C$
and $\mathcal C_1 = \mathcal C_2$.
The explanation for the improved performance is that the two concatenated simplex codes in 
$\mathcal C_B$ have minimum distance $4$ which is the same as the extended Hamming code $\mathcal C_1$
in the codes $\mathcal C_A$ and $\mathcal C_C$. 
However, more errors can be corrected depending on their distribution.
Assume that the concatenated code word contains four errors.
If in each of the two codes there are two errors the probability that they are corrected is very high
which does not hold for the extended Hamming code.
In contrast, when four errors are in one of the code words the decoding might be not correct
as for the extended Hamming code.

As can be seen in Fig. \ref{fig:64}  the codes $\mathcal C_D$  and $\mathcal C_E$ 
based on BCH codes
have a better performance compared to the other three based on RM codes.
The explanation is that the extended
BCH code $\mathcal C(16,7,6)$ exists
between the extended BCH codes $\mathcal C(16,5,8)$ and $\mathcal C(16,11,4)$
while for the other three codes we must use a code with smaller minimum distance.
An old result states that any RM code is equivalent to a very particular extended BCH code 
\cite[Ch. 13, Th. 11, p. 383]{sloane}.
The inverse direction is not true and in general
there exist many more BCH codes than RM codes of a certain length with different
dimensions and distances which are nested which means they  posses a subcode structure.
The code $\mathcal C_D$ has a better performance than  $\mathcal C_E$. 
For the latter the ML bound is nearly identical and therefore we show only one curve in Fig. \ref{fig:64}.
The difference between the codes $\mathcal C_D$ and $\mathcal C_E$ are the codes 
$\mathcal C_0$ and $\mathcal C_2$ while the other two are identical.
The only possible explanation for the  better performance can be
the different codes $\mathcal C_0$ as follows.
Since we use list decoders with list size $L=8$
for $\mathcal C_1$ and $\mathcal C_2$ 
the probability 
that the transmitted code word is in the list
should be the same for the codes $\mathcal C(16,5,8)$ and $\mathcal C(16,7,6)$.
Because if the latter would be better then  $\mathcal C_E$ should have the better performance.
So only the difference in $\mathcal C_0$ remains as the explanation of the performance
difference. 
Note that the fact that both $\mathcal C_0$ have minimum distance  $2$
does not mean that the ML decoding performance is the same since one has $2^{13}$
and the other $2^{15}$ code words in the same space.

The codes $\mathcal C_A$ and $\mathcal C_B$ are both based on RM codes.
In the following we will generalize these 
constructions for half-rate codes based on RM codes.
Half-rate RM codes $\mathcal R (r,m)$ only exist for $m$ odd which means 
	$\mathcal R (r=\nu, m=2 \nu +1)$, $\nu = 3, 4,  \ldots$.
        Thus, the length is  $n=2^{2 \nu +1}$,
	the dimension is $k=2^{2 \nu}$, and the minimum distance is $d=2^{\nu +1}$.
	The codes $\mathcal C_A$ and $\mathcal C_B$ give  general double Plotkin constructions for 
	half-rate codes of length $n=2^{2 \nu }$ based on RM codes as shown in the following Theorems.
	Recall that the dual code of an RM code $\mathcal C = \mathcal R (r, m)$
	is an RM code $\mathcal C^\perp = \mathcal R (m - r -1, m)$ and 
	$k + k^\perp = n = 2^m$.

\begin{theorem}[Half-Rate Codes of Length $2^{2 \nu}$ Based on RM Codes I]\label{rmratehalf}
	Let $\mathcal C_0  = \mathcal R (2\nu - \ell- 1, 2\nu - 2)$ and 
	$\mathcal C_3 = \mathcal C_0^\perp = \mathcal R (\ell , 2\nu -2)$
	be dual RM codes and let also
	 $\mathcal C_1 = \mathcal R (2\nu - \ell - 4, 2\nu - 2)$ and 
	$\mathcal C_2 = \mathcal C_1^\perp = \mathcal R (\ell + 1, 2\nu -2)$
	be dual RM codes, 
	where $\ell = 0, 1, \ldots , \nu - 3$.
	The codes obtained by a double Plotkin construction with these RM codes are half-rate codes  
	with the parameters $n= 2^{2\nu}, k = 2^{2\nu-1}$, and $d= \mathrm{min} \{d_3, 2 d_2, 2 d_1, 4 d_0 \}$. 
	All codes have the property $\mathcal C_3 \subset \mathcal C_2 \subseteq \mathcal C_1 \subset \mathcal C_0 $. 
\end{theorem}
\begin{proof}
The length is $4 \cdot 2^{2 \nu - 2} = 2^{2 \nu }$ because the length of the four used codes 
	is $ 2^{2 \nu - 2}$. 
Independent of the choice of $\ell$ the codes $\mathcal C_0$ and $\mathcal C_3$ are dual to each other 
	which means that the sum of their dimensions equals the length which is  
$ 2^{2 \nu - 2}$.
Also, the codes $\mathcal C_1$ and $\mathcal C_2$ are dual to each other 
	and the sum of their dimensions is again  
$ 2^{2 \nu - 2}$.
	Therefore, the dimension of the double Plotkin construction is $ 2^{2 \nu - 1}$.

For the minimum distance we follow the arguments of the proof of Th. \ref{uuplusvrm}.
Adding the four received blocks we get 
	$\mathbf u_3 + \mathbf e_0 +\mathbf e_1 +\mathbf e_2 +\mathbf e_3 $
	which can be correctly decoded if less than $ d_3/2$ errors have occured.
	Knowing $\mathbf u_3$ we can get two different vectors 
	$\mathbf u_2 + \mathbf e_0 +\mathbf e_2$ and
	$\mathbf u_2 + \mathbf e_1 +\mathbf e_3$ which can be corrected if less than  $ d_2$
	errors have occured.
	Similarly, we can get two  vectors 
	$\mathbf u_1 + \mathbf e_0 +\mathbf e_1$ and
	$\mathbf u_1 + \mathbf e_2 +\mathbf e_3$ which can be corrected if less than  $ d_1$
	errors have occured.
Finally, knowing $\mathbf u_1$, $\mathbf u_2$, and $\mathbf u_3$ we can get four vectors
$\mathbf u_0 + \mathbf e_0$, $\mathbf u_0 + \mathbf e_1$,  $\mathbf u_0 + \mathbf e_2$, and 
$\mathbf u_0 + \mathbf e_3$ which can be corrected if less than  $ 2 d_0$
	errors have occured.
	Therefore, the minimum distance is the minimum of $\{d_3, 2 d_2, 2 d_1, 4 d_0 \}$.

To prove the subcode structure we use the property of RM codes, namely, 
	$\mathcal C_3= \mathcal R (\ell, 2\nu -2) \subset \mathcal C_2= \mathcal R (\ell + 1, 2\nu -2)$,
	$\mathcal C_2= \mathcal R (\ell + 1, 2\nu -2) \subseteq \mathcal C_1= \mathcal R (2\nu - \ell - 4 , 2\nu -2)$
	since $\ell \leq \nu-3$,
	and
		$\mathcal C_1= \mathcal R (2\nu - \ell - 4, 2\nu -2) \subset \mathcal C_0= \mathcal R (2\nu - \ell - 1 , 2\nu -2)$.
\end{proof}

\begin{theorem}[Half-Rate Codes of Length $2^{2 \nu}$ Based on RM Codes II]\label{rmratehalfII}
	Let $\mathcal C_0= \mathcal R (\nu-1 + \ell, 2\nu - 2)$ and 
	$\mathcal C_3 = \mathcal C_1^\perp = \mathcal R (\nu -2 - \ell, 2\nu -2)$
	be the dual RM codes where $\ell = 0,1, \ldots \nu -2$. 
	Let $\mathcal C_1=\mathcal C_2$ be the concatenation of the half-rate RM code
	$|\mathcal R (\nu -2, 2\nu-3)|\mathcal R (\nu -2, 2\nu - 3)|$.
	The codes obtained by a double Plotkin construction with these RM codes are half-rate codes  
	with the parameters $n= 2^{2\nu}, k = 2^{2\nu-1}$, and $d= \mathrm{min} \{d_3, 2 d_2 = 2 d_1, 4 d_0 \}$. 
	The codes have the property $\mathcal C_3 \subset \mathcal C_2 = \mathcal C_1 \subset \mathcal C_0 $. 
\end{theorem}
\begin{proof}
The length is $4 \cdot 2^{2 \nu - 2} = 2^{2 \nu }$ because the length of the four used codes 
	is $ 2^{2 \nu - 2}$. 
Independent of the choice of $\ell$ the codes $\mathcal C_0$ and $\mathcal C_3$ are dual to each other 
	which means that the sum of their dimensions equals the length which is  
$ 2^{2 \nu - 2}$.
	For the half-rate code the dimension is $ 2^{2 \nu - 4}$ and because two times two code words are concatenated
	the sum of the dimensions is $ 2^{2 \nu - 2}$.
	Therefore, the dimension of the double Plotkin construction is $ 2^{2 \nu - 1}$.
For the minimum distance the same arguments hold as in  the proof of Th. \ref{rmratehalf}.
Note, that $d_1 = d_2$ is the minimum distance of the half-rate code $\mathcal R (\nu -2, 2\nu-3)$.

To prove the subcode structure we desribe  $\mathcal C_3= \mathcal R (\nu -2 - \ell, 2\nu -2)$
	as Plotkin construction 
$|\mathbf v_0 |\mathbf v_0 + \mathbf v_1|$ where $\mathbf v_0 \in  \mathcal R (\nu -2 - \ell, 2\nu-3)$ 
and $\mathbf v_1 \in \mathcal R (\nu -3 - \ell, 2\nu - 3)$. 
The code $\mathcal C_1 = \mathcal C_2$ 
consists of the code words $|\mathbf u_0|\mathbf u_0 |$ where $\mathbf u_0 \in  \mathcal R (\nu -2 , 2\nu-3)$.
	Since $\mathcal R (\nu -2 - \ell, 2\nu - 3) \subset \mathcal R (\nu -2, 2\nu-3)$ it follows that 
 $\mathcal C_3 \subset \mathcal C_2 | \mathcal C_2$.
	Now we desribe  $\mathcal C_0= \mathcal R (\nu -1 + \ell, 2\nu -2)$
	as Plotkin construction 
$|\mathbf v_0 |\mathbf v_0 + \mathbf v_1|$ where $\mathbf v_0 \in  \mathcal R (\nu -1 + \ell, 2\nu-3)$ and 
	$\mathbf v_1 \in \mathcal R (\nu -2+ \ell, 2\nu - 3)$. 
Since $ \mathcal R (\nu -2, 2\nu-3) \subset \mathcal R (\nu -1 + \ell, 2\nu - 3)$ it follows that 
 $\mathcal C_1 \subset \mathcal C_0$.
\end{proof}

For the half-rate codes according to Th. \ref{rmratehalf} and Th. \ref{rmratehalfII}
all $15$ decoding variants using different hidden code words for starting the decoding
can be applied. Also, recursive decoding is possible which will be described in next section.

\section{Performance of Recursive Plotkin Constructions}\label{sec:recurplot}
The main question is if the statistical properties of the join and add operations
are also valid if we apply the double Plotkin construction twice.
To show this we start with a half-rate code 
and afterwards we simulate the performance of a low and a high rate RM code.
\subsection{Half-Rate Codes}
We consider the half-rate RM code $\mathcal R(3,7) = (128,64,16)$  from Ex. \ref{ex:r37}
which is a double Plotkin construction with the four codes   
	$\mathcal C_0 = \mathcal R(3,5) = (32,26,4)$, 
$\mathcal C_1 = \mathcal C_2 =  \mathcal R(2,5) = (32,16,8)$, 
and $\mathcal C_3 = \mathcal R(1,5) = (32,6,16)$.
For decoding we use the variants of Sec.
\ref{sec:startc3} and Sec. \ref{sec:starteizw} in which we have to decode
hidden code words according to 
the codes $\mathcal C_1 = \mathcal C_2$ and   
$\mathcal C_3$.
For the decoding of the hidden code words of $\mathcal C_1 = \mathcal C_2$ we consider 
the code word again as a double Plotkin construction
with the codes $(8,7,2)$, $(8,4,4)$, and $(8,1,8)$ according to Ex. \ref{ex:recursive}. 
We use 
 all six variants 
 V$(\mathbf y_0 \Join \mathbf y_1)$, V$(\mathbf y_0 \Join \mathbf y_2)$ 
 up to  V$(\mathbf y_2 \Join \mathbf y_3)$ from Sec. \ref{sec:c12r25}
 which start with the decoding of the $(8,4,4)$ code.
 Additionally, we use two variants  
  V$(4, \mathbf y_0 \Join \mathbf y_2)$ and   V$(4, \mathbf y_0 \Join \mathbf y_1)$ with $L=2$ 
  from Sec. \ref{sec:r25c123}
  which starts with the $(8,1,8)$ repetition code.
 According to Fig. \ref{fig:r25-mixed} this decoder has ML performance.

  A list decoder for $\mathcal C_1 = \mathcal C_2$ uses the codes $(16,5,8)$  
  and $(16,11,4)$ of the Plotkin construction and creates a list for the $(16,5,8)$ code
  which needs $8$ correlations using Lemma \ref{mldecr1m}.
  With any element of this list the codes  
  $\mathcal C'_2=(8,4,4)$ and $\mathcal C'_3=(8,1,8)$ of the double Plotkin construction
  are known. Then, an add-join operation can be done
  which is used to calculate a list with four entries for the $\mathcal C'_1=(8,4,4)$.
  For each element of this list $\mathcal C'_0=(8,7,2)$ is calculated by an add-four operation.
  We have then $32$ code words of the code $\mathcal C_1 = \mathcal C_2$ and choose the 
  $L$ best which is determined by their correlation.

The decoding of $\mathcal C_3(32,6,16)$ is done with an ML decoder according  to Lemma \ref{mldecr1m}
and the $16$ correlation values can also be used to calculate a list of the best decoding decisions.
The extended Hamming code $\mathcal C_0(32,26,4)$ is considered as a double Plotkin construction
with the codes $(8,8,1)$, $(8,7,2)$, and $(8,4,4)$
and decoded with six variants from Sec. \ref{sec:c12r25}.


\begin{figure}[htb]
	\begin{center}
	\includegraphics[width=0.65\textwidth]{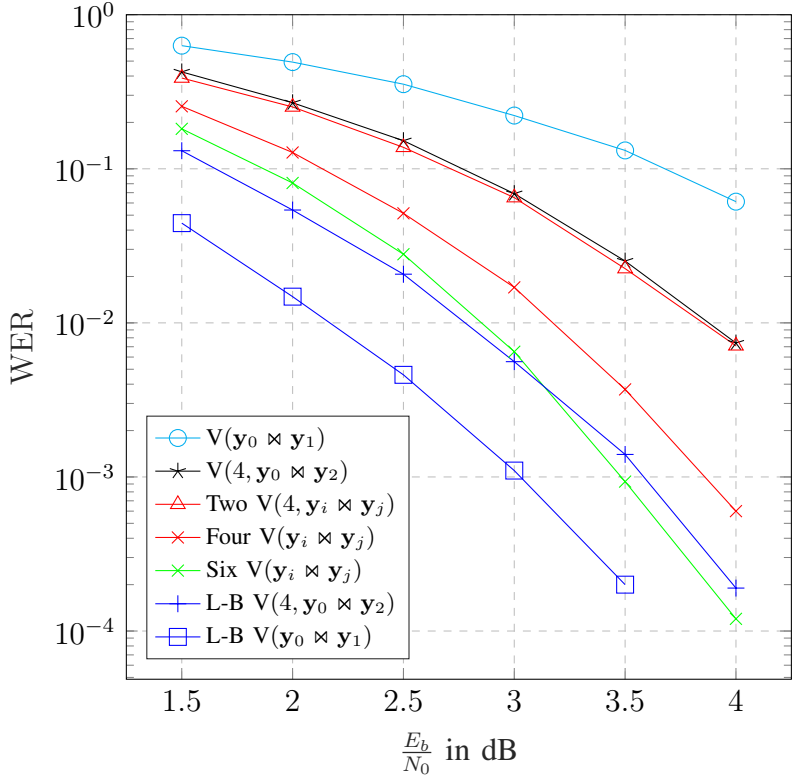}
	\end{center}
	\caption{$\mathcal R(3,7)$ WER for Variants Decoding $\mathcal C_1$, $\mathcal C_2$, or $\mathcal C_3$ First}\label{fig:ar37}
\end{figure}

In Fig. \ref{fig:ar37} the simulation results of recursive decoding are shown 
without list decoding in the first step.
Also, here the variant V$(\mathbf y_0 \Join \mathbf y_1)$ which decodes $\mathcal C_1$
in the first step has about one dB worse performance at a WER of $10^{-1}$
than the variant V$(4, \mathbf y_0 \Join \mathbf y_2)$  which decodes $\mathcal C_3$ first.
Using additionally the variants V$(4, \mathbf y_0 \Join \mathbf y_1)$
and deciding for the decision with the larger correlation
gives almost no improvement. 
Adding the third variant  V$(4, \mathbf y_1 \Join \mathbf y_2)$ would
give no visible improvement of the performance.
This was already observed in the decoding of
the $\mathcal R (2,5)$ code and is due to the fact that these variants decode in the first step the same
join-four combination which dominates the performance. 
Taking two variants V$(\mathbf y_0 \Join \mathbf y_1)$ and V$(\mathbf y_2 \Join \mathbf y_3)$
from Sec. \ref{sec:c12r25} shows a performance very close to the two variants 
V$(4, \mathbf y_0 \Join \mathbf y_2)$  and V$(4, \mathbf y_0 \Join \mathbf y_1)$ 
(the curve is not plotted in Fig. \ref{fig:ar37}).
Therefore, the performance difference of one variant V$(\mathbf y_0 \Join \mathbf y_1)$  
and two variants V$(\mathbf y_0 \Join \mathbf y_1)$ and V$(\mathbf y_2 \Join \mathbf y_3)$
and deciding for the result with larger correlation 
is  significant and amounts to about $1$ dB  at a WER of $10^{-1}$. 
As can be seen in Fig. \ref{fig:ar37}
using four variants, 
additionally V$(\mathbf y_0 \Join \mathbf y_2)$ and V$(\mathbf y_1 \Join \mathbf y_3)$, 
the performance gain is about $1.5$ dB compared to one variant only.
The gain when using six variants is almost $2$ dB at a WER of $5 \cdot 10^{-2}$ compared to one variant.
An observation in Fig. \ref{fig:ar37}
is that the performance curves show a steeper slope the more variants are used.
These results confirm the considerations according to Eq. \eqref{birthapprox}  and Lemma \ref{genialparadox}
which predict that within the six different hidden code words
there is at least one which can be decoded correctly with high probability.
When the decoding of the first step is correct the remaining three decoding steps
benefit from error cancellation as shown in Sec. \ref{sec:cancel}.
Using a list decoder in the first decoding step increases the probability that
the correct code word  of the first step is in the list.
In order to estimate the possible gain when in the first decoding step 
a list decoder is used, we introduce a list decoding bound.
\begin{lemma}[List Decoding Bound, L-Bound]\label{lbound}
The L-Bound is the performance of a variant 
	V$(4, \mathbf y_i \Join \mathbf y_j)$ or V$(\mathbf y_i \Join \mathbf y_j)$
assuming that the code word of the first decoding step is known. The remaining three code words
	are decoded according the used variant.
The decoding of a variant with a list decoder in the first decoding step 
	cannot perform  better than the L-Bound. 
\end{lemma}
\begin{proof}
	If the transmitted code word is not in the list
of the decoder of the first decoding step, correct decoding is impossible.
However, correct decoding might be possible if the transmitted code word of the first decoding step is known.
\end{proof}

It can be seen in Fig. \ref{fig:ar37} that a gain of more than one dB at a WER of $10^{-1}$ 
is possible according to the L.Bound from Lemma \ref{lbound}. Therefore, for variant V$(4, \mathbf y_0 \Join \mathbf y_2)$ with list decoding the performance cannot be better.
For the list decoding of V$(\mathbf y_0 \Join \mathbf y_1)$ a gain of more than $2.5$ dB at a WER of $10^{-1}$ 
is possible.
Note that there is a significant difference of more than  one dB  of the gain for the L-Bound
between the variants from Sec. \ref{sec:startc3} starting with  $\mathcal C_3$   and those from 
\ref{sec:starteizw}  starting with  $\mathcal C_1 = \mathcal C_2$.
This can be explained as follows. 
If the first decoding step for $\mathcal C_3$ decodes correctly the second decoding step uses
the join-add operation for decoding  $\mathcal C_2$,
whereas if the first decoding step for $\mathcal C_1$ decodes correctly the 
second decoding step uses the join-two operation for decoding  $\mathcal C_3$
which is designed for a join-four operation.  
Therefore, the probability that the second decoding step decodes correctly is larger 
in case when we start with a correct code word of $\mathcal C_1$.
The third and the fourth decoding step are in principal the same for both cases.
The add-join operation is used for decoding $\mathcal C_1$ or $\mathcal C_2$
and the add-four operation  
 provides a channel with $6$ dB gain for decoding $\mathcal C_0$.
The crossing of the performance curve of six variants and the L-Bound for variant V$(4, \mathbf y_0 \Join \mathbf y_2)$
in Fig. \ref{fig:ar37} supports the fact that the more variants are used the steeper the curve gets.

\begin{figure}[htb]
	\begin{center}
	\includegraphics[width=0.65\textwidth]{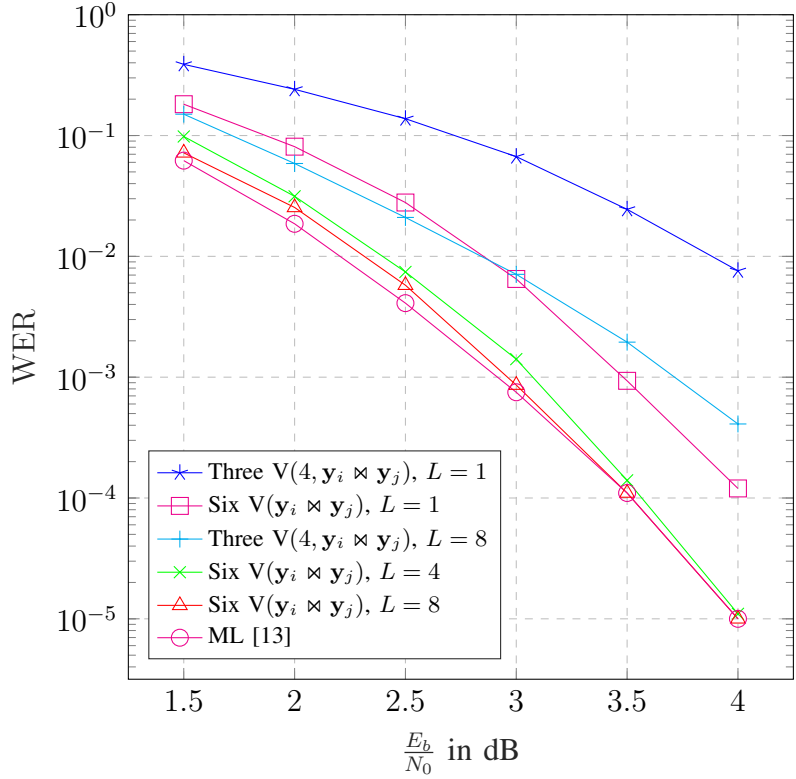}
	\end{center}
	\caption{$\mathcal R(3,7)$ WER for List Decoding in the First Step}\label{fig:br37}
\end{figure}

In  Fig. \ref{fig:br37} the simulation results for list decoding in the first decoding step are shown.
It can be seen that using list size $L=8$ and the three variants V$(4, \mathbf y_i \Join \mathbf y_j)$
a gain of a little more than one dB at a WER of $10^{-2}$ is achieved compared to $L=1$.
This fits well to the prediction by the L-Bound (Lemma \ref{lbound}) in Fig. \ref{fig:ar37}. 
The performance of the six variants V$(\mathbf y_i \Join \mathbf y_j)$
without list ($L=1$) has a larger slope and is crossing the curve of V$(4, \mathbf y_i \Join \mathbf y_j)$
with $L=8$ at about $3$ dB. 
Note that the six variants consist of six decoders while the other three variants with $L=8$ need $24$ decoders.
The explanation is that the probability that within the six different join-two
is at least one case which can be decoded correctly is large.
In addition, if the first decoding is correct the L-Bound in Fig. \ref{fig:ar37}
predicts that the remaining steps are decoded correctly. 
The curves of six variants with list sizes  $L=4$ and $L=8$ are approaching the ML performance.
It can be observed that the curves are getting closer to the ML performance
for a better channel. Both curves are nearly identical to ML at $4$ dB and the $L=8$
case already at $3.5$ dB which shows that
the performance difference between  $L=4$ and  $L=8$ is vanishing for better channels. 
We could combine the six variants V$(\mathbf y_i \Join \mathbf y_j)$ with  $L=8$ 
and three variants V$(4, \mathbf y_i \Join \mathbf y_j)$ with $L=8$ 
which would result in a performance curve even closer to the ML curve but would also increase the complexity. 

The complexity of the decoding follows straight forwardly from 
the complexities derived in Sec. \ref{sec:complexity}
and in the following we give some examples.
The decoding starting with  $\mathcal C_3$  needs the
 operations one join-four,
one join-add, one add-join, and one add-four, which
according to Table \ref{table:compjoinadd}  needs
$9n$ sign operations, $6n$ comparisons, and $6n$ additions
or $12n$ \^a\^c operations. Since $n=32$ the complexity is $384$ \^a\^c operations.
In addition, we need to decode one time $\mathcal C_3$ which according to Lemma \ref{mldoplr1m}
takes $256$ \^a\^c operations.
Then, the decoding of $\mathcal C_1$ and $\mathcal C_2$
using six variants takes $887$ \^a\^c operations each as derived in Sec. \ref{sec:complexity}.
We assume that $\mathcal C_0$ takes also $887$ \^a\^c operations
even though it takes much less.
Thus, the complexity for the decoders is $256 + 3 \cdot 887 = 2917$ \^a\^c operations.
All together we need $384 + 2917 = 3301$ \^a\^c operations.
Any variant from Sec. \ref{sec:starteizw} without list
needs two join-two, one add-join, and one add-four, thus
 $8n$ \^a\^c operations. Since $n=32$ the complexity is $256$ \^a\^c operations.
For the decoders the complexity is the same as before and therefore,
 we need $256 + 2917 =  3173$ \^a\^c operations which is a little less than the variant above.

\begin{figure}[htb]
	\begin{center}
	\includegraphics[width=0.65\textwidth]{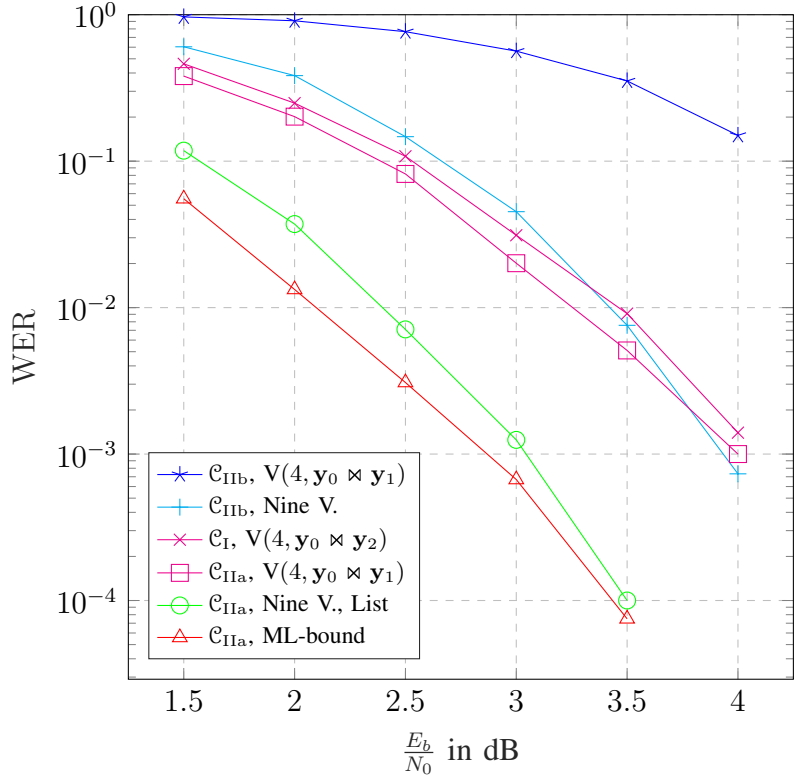}
	\end{center}
	\caption{WER for Three Codes $\mathcal C (256,128, 16)$}\label{fig:256}
\end{figure}

In the following example we will simulate 
 half-rate codes which are not RM codes, namely,
three $\mathcal C (256, 128, 16)$ codes which are
constructed according to Th. \ref{rmratehalf} with $\nu =4$ and $\ell = 1$
and Th. \ref{rmratehalfII} with $\nu =4$ and $\ell =0, 1$.  
For the first case with $\ell =1$ we obtain the code $\mathcal C_{\mathrm{I}} (256, 128, 16)$ using 
a double Plotkin construction with the four codes
$\mathcal C_0 = \mathcal R(4,6) = (64,57,4)$, 
$\mathcal C_1 = \mathcal R(3,6) = (64,42,8)$, 
$\mathcal C_2 = \mathcal R(2,6) = (64,22,16)$, 
and $\mathcal C_3 = \mathcal R(1,6) = (64,7,32)$.
For the second case with $\ell = 0$ we get
the code $\mathcal C_{\mathrm{IIa}} (256, 128, 16)$ 
$\mathcal C_0 = \mathcal R(4,6) = (64,57,4)$, 
$\mathcal C_1 = \mathcal C_2 =  |\mathcal R(2,5)|\mathcal R(2,5)|$  
where  $\mathcal R(2,5) = (32,16,8)$, 
and $\mathcal C_3 = \mathcal R(1,6) = (64,7,32)$.
For $\ell =1$ the code $\mathcal C_{\mathrm{IIb}} (256, 128, 16)$ uses
$\mathcal C_0 = \mathcal R(3,6) = (64,42,8)$  
and $\mathcal C_3 = \mathcal R(2,6) = (64,22,16)$.
The codes
$\mathcal C_1 = \mathcal C_2 $ are the same as for  $\mathcal C_{\mathrm{IIa}}$. 
Note that $\ell =2$ would lead to a minimum distance of $8$ and is not considered. 

For the decoding of the codes  $\mathcal R(2,6)$, $\mathcal R(3,6)$,
and $\mathcal R(4,6)$ we use double Plotkin constructions
with $\mathcal R(r,4)$ codes.
The list decoding for the code $ |\mathcal R(2,5)|\mathcal R(2,5)|$
is done as follows. Each half is list decoded
seperately by a list decoder for $\mathcal R(2,5)$ which was described at the begin of this section. 
The best code word of the list for the $ |\mathcal R(2,5)|\mathcal R(2,5)|$ code
is the best of the first list  combined with the best of the second list.
The second best is the combination
of the best of the first list with the second best of the second list.
The third best is the combination
of the second best of the first list with the best of the second list and so on
until listsize $L$.
Note that this list decoder does not give a list with code words where the 
code words are sorted according to their probabilities.

In Fig. \ref{fig:256} the related simulation results are shown. 
Using only variant  V$(4, \mathbf y_0 \Join \mathbf y_1)$ without list the performance 
of code $\mathcal C_{\mathrm{IIb}}$ is almost two dB worse than for code $\mathcal C_{\mathrm{IIa}}$ at a WER of $10^{-1}$.
The explanation is that the join-four creates an average error rate which
is about three times larger than the one in the channel.
The code $\mathcal C_{3}$ of $\mathcal C_{\mathrm{IIa}}$ has minimum distance
$32$ while the one for
$\mathcal C_{\mathrm{IIb}}$ has only $16$ and
if the first decoding step is wrong the following three steps are also wrong with 
very large probability.
The variant V$(4, \mathbf y_0 \Join \mathbf y_2)$ without list 
for code $\mathcal C_{\mathrm{I}}$ has a slightly worse performance than $\mathcal C_{\mathrm{IIa}}$.
The explanation is that the code $\mathcal C_{1}$ in $\mathcal C_{\mathrm{I}}$ has only minimum distance $8$
while the concatenated code in $\mathcal C_{\mathrm{IIa}}$ has distance $8$ in both halfs.
Using all nine variants without list decoding 
the performance improvement is about $1.5$ dB for $\mathcal C_{\mathrm{IIa}}$.
The improvement  using nine variants for decoding of $\mathcal C_{\mathrm{IIb}}$ 
is only  $0.5$ dB  at a WER of $10^{-1}$ and we do not plot this curve.
 Also, here it can be observed that the slope of the curves with more variants is steeper.
 Even the code with the worst performance is crossing with nine variants the curves
of $\mathcal C_{\mathrm{I}}$  and $\mathcal C_{\mathrm{IIa}}$ with one variant only 
 which leads to the crossing of the curves of V$(4, \mathbf y_0 \Join \mathbf y_1)$ 
for  $\mathcal C_{\mathrm{I}}$ and the one for nine variants for $\mathcal C_{\mathrm{IIb}}$.
Using all nine variants for $\mathcal C_{\mathrm{IIa}}$ with list decoding $L=16$ aproaches the ML bound
which means the performance is close to the ML performance.
The performance improvement using list decoding is larger for code $\mathcal C_{\mathrm{IIa}}$ 
than for $\mathcal C_{\mathrm{I}}$. The explanation is again the difference in the $\mathcal C_{1}$ 
of the two constructions.

\begin{figure}[htb]
	\begin{center}
	\includegraphics[width=0.65\textwidth]{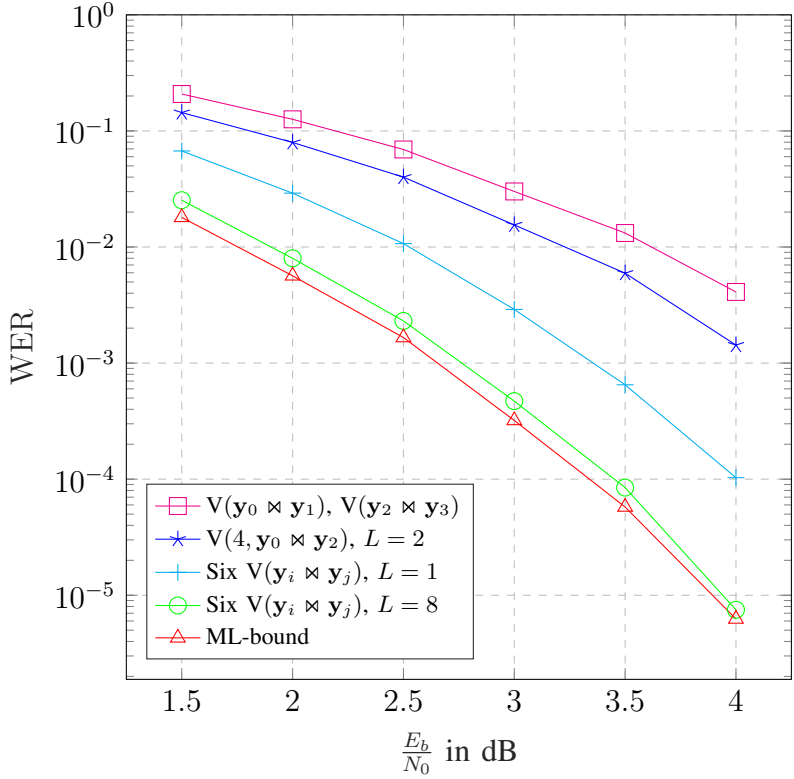}
	\end{center}
	\caption{$R(2,7)$ WER for Different Variants}\label{fig:low27}
\end{figure}

\subsection{Low-Rate Code}
In this section we will show that the novel decoding works well for low-rate codes.
For a low-rate code
we simulate the performance of the $\mathcal R(2,7) = (128,29,32)$ code 
which is a double Plotkin construction with the four codes   
	$\mathcal C_0 = \mathcal R(2,5) = (32,16,8)$, 
$\mathcal C_1 = \mathcal C_2 =  \mathcal R(1,5) = (32,6,16)$, 
and $\mathcal C_3 = \mathcal R(0,5) = (32,1,32)$.
For decoding we use the variants of Sec.
\ref{sec:startc3} and Sec. \ref{sec:starteizw} in which we have to decode
hidden code words according to 
the codes  $\mathcal C_1 = \mathcal C_2$ and   
$\mathcal C_3$.
For the decoding of the hidden code words of $\mathcal C_1 = \mathcal C_2(32,6,16)$ we use 
 an ML decoder according to Lemma \ref{mldecr1m}.
 The $16$ correlation values can also be used to calculate a list of the $L$ best decisions.
The code $\mathcal C_0$ is considered  as a double Plotkin construction
with the codes $(8,7,2)$, $(8,4,4)$, and $(8,1,8)$. 
For this decoding we use 
 all six variants 
 V$(\mathbf y_0 \Join \mathbf y_1)$, V$(\mathbf y_0 \Join \mathbf y_2)$ 
 up to  V$(\mathbf y_2 \Join \mathbf y_3)$ from Sec. \ref{sec:c12r25}.
The decoding of the repetition code $\mathcal C_3$ is done by correlation.
Fig. \ref{fig:low27} shows the simulation results. 
The performance of V$(4, \mathbf y_0 \Join \mathbf y_2)$ with $L=2$
is better than two variants V$(\mathbf y_0 \Join \mathbf y_1)$ and V$(\mathbf y_2 \Join \mathbf y_3)$.
Six variants V$(\mathbf y_i \Join \mathbf y_j)$ have a gain of more than 
 one dB at a WER of $10^{-2}$ compared to two variants.
These six variants with list size $L=8$ in the first decoding step
have a gain of more than half a dB at a 
WER of $10^{-2}$ compared to list size $L=1$ and 
show a performance close to the ML-Bound.
Again, the performance curve for six variants and listsize $L=8$ 
could be improved by
adding the three variants starting with $\mathcal C_3$ and listsize $L=2$.

\begin{figure}[htb]
	\begin{center}
	\includegraphics[width=0.65\textwidth]{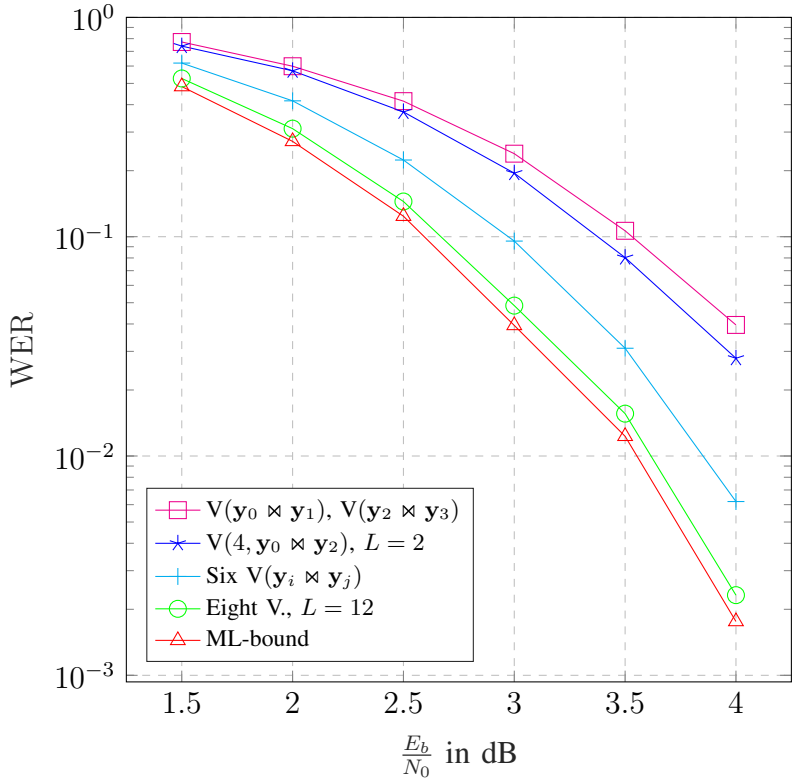}
	\end{center}
	\caption{$R(4,7)$ WER for Different Variants}\label{fig:high47}
\end{figure}

\subsection{High-Rate Code}
For a high-rate code
we simulate the performance of the $\mathcal R(4,7) = (128,99,8)$ code 
which is a double Plotkin construction with the four codes   
	$\mathcal C_0 = \mathcal R(4,5) = (32,31,2)$, 
$\mathcal C_1 = \mathcal C_2 =  \mathcal R(3,5) = (32,26,4)$, 
and $\mathcal C_3 = \mathcal R(2,5) = (32,16,8)$.
For decoding we use the variants of Sec.
\ref{sec:startc3} and Sec. \ref{sec:starteizw} in which we have to decode
hidden code words according to 
the codes  $\mathcal C_1 = \mathcal C_2$ and   
$\mathcal C_3$.
The decoder for $\mathcal C_3$ is the same as in the low-rate case 
and a list decoder is based on a list decoder for the simplex code
$\mathcal R(1,4) = (16,5,8)$.
For the decoding of the hidden code words of $\mathcal C_1 = \mathcal C_2(32,26,4)$ we consider 
it as a double Plotkin construction
with the codes $(8,8,1)$, $(8,7,2)$, and $(8,4,4)$. 
A list decoder uses lists for the $(8,7,2)$ code in the variants
 V$(\mathbf y_0 \Join \mathbf y_1)$ and V$(\mathbf y_2 \Join \mathbf y_3)$ and
 for each of the list elements the remaining three decoding steps are applied to calculate a code word
 of the $(32,26,4)$ code.
The code $\mathcal C_0$ is a parity-check code.
The simulation results are shown in Fig. \ref{fig:high47}.
The performance of  V$(4, \mathbf y_0 \Join \mathbf y_2)$ with $L=2$
has a better performance than 
the two variants V$(\mathbf y_0 \Join \mathbf y_1)$ and V$(\mathbf y_2 \Join \mathbf y_3)$. 
Using six variants has about $0.5$ dB better performance  at a WER of $10^{-1}$ 
compared to two variants.
The performance improvements using additional variants for decoding are smaller
than for low-rate codes. 
Because of the relatively small minimum distance the possible performance gain of ML decoding 
compared to a simple decoder using two variants is less than one dB 
while this gain was almost two dB in case of the low-rate code.
Using the six variants and a listsize $L=12$ for the first decoding step
and additionally two variants V$(4, \mathbf y_0 \Join \mathbf y_2)$ V$(4, \mathbf y_0 \Join \mathbf y_1)$,
also with a listsize $L=12$, gives the curve (named Eight V. $L=12$) in  Fig. \ref{fig:high47}
which is close to the ML-Bound.

\section{Conclusions}
We have described a novel decoding strategy of recursive Plotkin constructions which uses variants
which start the decoding with different hidden code words. 
The final decoding decision selects the best out of this variants.
Despite the fact that a single variant showed a performance that is much worse than  classical decoding,
using several variants the performance approches the ML performance.
Since the decoding variants are independent of each other
a tradeoff between performance and complexity is possible by using more or less variants for decoding.
The usage of list decoders in the first decoding step improves also the performance.
We explained this effect by showing that the remaining decoding steps 
benefit from error cancellation and simulated the L-Bounds which confirmed the explanations. 
The statistical properties which are exploited by the novel decoding strategy
holds as well for the recursive application of the decoding variants
as well as for codes of different rates.
Note that the three codes from Sec. \ref{sec:recurplot}
used in a double Plotkin construction would give the RM code $\mathcal R (4,9) = (512, 256, 32)$.
We also showed that constructions which are not RM codes can be decoded with a performance close to ML.

An old result states that any RM code is equivalent to an very particular extended BCH code 
\cite[Ch. 13, Th. 11, p. 383]{sloane}.
The inverse direction is not true and in general
there exist many more BCH codes than RM codes of a certain length with different
dimensions and distances which are nested which means they  posses a subcode structure.
In \cite{Freudenberger} results are given based on BCH codes with a Plotkin construction
but not a double Plotkin construction with decoding variants.
We have only given two examples 
that the use of double Plotkin construction with BCH codes improves the decoding performance
and also the ML performance which means that the constructed codes are better. 
An extensive study with double Plotkin constructions and BCH codes remains an open problem.

A main open question which was not considered is when fixing the complexity
which combination of variants gives the best performance. This should be studied especially
for possible applications.
Even though the decoding complexity is already very low,
another open question is if the exchange of information between the variants can 
reduce the complexity without loss in performance.
For example, the first two decoding steps of V$(\mathbf y_0 \Join \mathbf y_1)$ 
estimate $\mathbf x_1$ and $\mathbf x_3$ and the first decoding step of variant 
 V$(\mathbf y_2 \Join \mathbf y_3)$ estimates  $\mathbf x_1\mathbf x_3$
 and the results may fit or are a contradiction. However, both outcomes
 deliver information. 
 There are many more of such connections of the variants.
 How to use this information and if it is beneficial is an open problem.
Another open question is concerning the list decoding.
It is not studied if the list decoding used could be improved
in the sense that it is not proved if the lists contains 
all most reliable code words. 
However, this is not valid for list decoding of the simplex codes which cannot be improved.
Clearly, the performance cannot be improved if it is close to the ML performance already 
but maybe the decoding complexity can be reduced.

\section*{Acknowledgements}
The author thanks Sebastian Bitzer
for numerous valuable discussions.
The simulations have been done with the help of SageMath \cite{sage}.


\vspace{.2cm}
\begin{thebibliography}{1}


\bibitem{Muller}
 D.E. Muller,
"Application of boolean algebra in switching circuit design and to error detection", 
IEEE Trans. on Computers, vol. 3, pp. 6-12, 1954. 
 
\bibitem{Reed}
 I.S. Reed, "A class of multiple-error-correcting codes and the decoding scheme"
, 
IEEE Trans. on Inf. Theory, vol. 4, pp. 38-49, 1954. 

\bibitem{Plotkin}
 M. Plotkin, "Binary codes with specific minimum distances"
, 
IEEE Trans. on Inf. Theory, vol. 6, pp. 445-450, 1960. 


\bibitem{BZ}
 E.L. Blokh, V.V. Zyablov, "Coding of generalized cascade codes", 
Problemy Peredachi Informatsii, vol. 10, no. 2,  pp. 45-50, 1974. 


\bibitem{Litsyn}
 S. Litsyn, E. Nemirovski, O. Shekhovtsov, L. Mikhailovskaya,
		"The fast decoding of first order Reed--Muller codes in the Gaussian Channel",
Problems of Control and Information Theory, vol. 14, pp. 189-201, 1985. 

\bibitem{Karyakin}
 Y.D. Karyakin, "Fast Correlation Decoding of Reed--Muller Codes",
Problemy Peredachi Informatsii, vol. 23,  pp. 40-49, 1987. 

\bibitem{Boss95}
 G. Schnabl, M. Bossert, Soft decision decoding of RM codes as generalized multiple concatenated codes
, 
IEEE Trans. on Inf. Theory, vol. 41, pp. 304-308, 1995. 


\bibitem{sloane}
 F.J. MacWilliams, N.J.A. Sloane,
The Theory of Error Correcting Codes, North-Holland, 1977. 


\bibitem{Boss-eng}
 M. Bossert,
Channel Coding for Telecommunications, John Wiley and Sons Ltd., 1999. 

\bibitem{DBoss}
 M. Bossert,
		Kanalcodierung, Teubner, Germany, 1998 (German version of \cite{Boss-eng}). 


\bibitem{Stolte}
N. Stolte, "Recursive codes with the Plotkin construction and their decoding", PhD Dissertation, 
		TH Darmstadt, Germany, 2002.



\bibitem{Dumer}
I. Dumer and K. Shabunov, "Soft-decision decoding of Reed--Muller codes: Recursive lists", IEEE Trans. Inf. Theory, vol. 52, no. 3, pp. 1260-1266, Mar. 2006.

\bibitem{Geiselhart}	 
  M. Geiselhart, A. Elkelesh, M. Ebada, S. Cammerer, S. ten Brink,
  "Automorphism Ensemble Decoding of Reed–Muller Codes", 
  IEEE Trans. on Com.,
   vol.69, no.10, pp. 6424-6438, 2021.

\bibitem{Freudenberger}
D. Bailon, M. Bossert, J.-P. Thiers, J. Freudenberger, "Concatenated codes
based on the Plotkin construction and their soft-input decoding", IEEE Trans. on
Com., vol. 70, no. 5, pp. 2939-2950, 2022.


\bibitem{Kamenev}
M. Kamenev, "Recursive Decoding of Reed--Muller Codes Starting With the Higher-Rate Costituent Codes",
arXiv:2101.11328v3 [cs.IT], 2022.

\bibitem{sage}
The Sage Developers, "SageMath, the Sage Mathematics Software System", 2018

\end{thebibliography}
\end{document}